\documentclass[a4paper]{article}
\RequirePackage{pdf14}
\RequirePackage{fix-cm}
\usepackage{jheppub}

\usepackage{graphicx}
\usepackage{amsfonts,amsmath,amssymb,amsthm}
\usepackage{tikz,tikz-cd}
\usetikzlibrary{positioning}
\usepackage{verbatim}
\usepackage[parfill]{parskip}
\usepackage{array,setspace,mathrsfs,yfonts,dsfont,bbm,colonequals,amscd,euscript}
\usepackage{relsize,suffix,mathtools,cancel,bbm}

\newtheorem{thm}{Theorem}[section]
\newtheorem{lem}[thm]{Lemma}
\newtheorem{conj}[thm]{Conjecture}
\newtheorem{prop}[thm]{Proposition}
\newtheorem{corr}[thm]{Corollary}
\newtheorem{rem}[thm]{Remark}

\newtheorem*{sthm}{Theorem}

\def\IC{\mathbb{C}}
\def\ID{\mathbb{D}}

\def\IN{\mathbb{N}}

\def\IR{\mathbb{R}}

\def\IV{\mathbb{V}}
\def\IZ{\mathbb{Z}}

\def\DD{{\cal D}}

\def\II{{\cal I}}
\def\JJ{{\cal J}}
\def\KK{{\cal K}}

\def\NN{{\cal N}}
\def\OO{{\cal O}}

\def\SS{{\cal S}}

\def\WW{{\cal W}}
\def\XX{{\cal X}}
\def\YY{{\cal Y}}
\def\ZZ{{\cal Z}}

\def\SU{\mathrm{SU}}

\def\Tr{\mathop{\mathrm{Tr}}\nolimits}

\def\dim{{\rm dim}}

\def\af{\mathfrak{a}}

\def\gf{\mathfrak{g}}

\def \zf{\mathfrak{z}}
\def \nf{\mathfrak{n}}

\def \slf{\mathfrak{sl}}
\def \Zf{\mathfrak{Z}}

\def \eg{{\it e.g.}}
\def \ie{{\it i.e.}}
\def \cf{{\it cf.}}

\def \id{\mathbbm{1}}
\def \KT{\mathrm{KT}}
\def \rk{\mathrm{rk}}

\newcommand{\Hi}[1]{\mathrm{H}^{\frac{\infty}{2}+#1}}

\def\D{\mathcal{D}^{ch}}
\def\W{\mathbf{W}}
\def\V{\mathbf{V}}

\title{Free field realisation of the chiral universal centraliser}
\author{Christopher Beem and Sujay Nair}
\affiliation{Mathematical Institute, University of Oxford, Woodstock Road, Oxford, OX2 6GG, UK}
\abstract{In the TQFT formalism of Moore--Tachikawa for describing Higgs branches of theories of class $\mathcal{S}$, the space associated to the unpunctured sphere in type $\mathfrak{g}$ is the universal centraliser $\Zf_G$, where $\mathfrak{g}=Lie(G)$. In more physical terms, this space arises as the Coulomb branch of pure $\mathcal{N}=4$ gauge theory in three dimensions with gauge group $\check{G}$, the Langlands dual. In the analogous formalism for describing chiral algebras of class $\mathcal{S}$, the vertex algebra associated to the sphere has been dubbed the \emph{chiral universal centraliser}. In this paper, we construct an open, symplectic embedding from a cover of the Kostant--Toda lattice of type $\mathfrak{g}$ to the universal centraliser of $G$---extending a classic result of Kostant. Using this embedding and some observations on the Poisson algebraic structure of $\Zf_G$, we propose a free field realisation of the chiral universal centraliser for any simple group $G$. We exploit this realisation to develop free field realisations of chiral algebras of class $\mathcal{S}$ of type $\mathfrak{a}_1$ for theories of genus zero with $n=1,\ldots,6$ punctures. These realisations make generalised $S$-duality completely manifest, and the generalisation to $n\geqslant7$ punctures is conceptually clear, though technically burdensome.}
\begin{document}
\maketitle
\flushbottom

%----------------------------------------------------------------------%
\section{\label{sec:intro}Introduction}
%----------------------------------------------------------------------%

The occurrence of vertex operator algebra (VOA) structures within the operator algebras of supersymmetric quantum field theories in various dimensions (see, \eg, \cite{Beem:2013sza, Beem:2014rza, Gaiotto:2017euk, Costello:2018swh, Costello:2019, Costello:2020ndc}) has been a major theme in the study of such systems in recent years, and has led to powerful new approaches to the analysis of these (usually strongly coupled) theories. Furthermore, physical intuition can potentially provide new insights into the study of VOAs arising through these correspondences at a purely formal level, and indeed developments in this area have led to a rich dialogue between physics and pure mathematics (see, \eg, \cite{Arakawa:2018egx,Creutzig:2017qyf,Creutzig:2018lbc,Costello:2018swh,Xie:2019yds,Xie:2019vzr,Creutzig:2021ext,Beem:2022mde} for an extremeley incomplete selection).

A surprising property of many VOAs arising in a physical context is their \emph{geometric} nature. This is perhaps most clearly exemplified through the Higgs Branch Conjecture of \cite{Beem:2017ooy}, which implies that the VOAs arising in connection with four-dimensional $\mathcal{N}=2$ SCFTs are \emph{chiral quantisations} of their Higgs branches (\cf\ \cite{Arakawa:2012c2}). A related, but less well formalised, idea is that the VOAs arising in this setting admit certain (generalised) free field realisations (FFRs) the structure of which is closely connected with the Poisson geometry of the Higgs branch \cite{Beem:2019tfp,Beem:2019snk, Beem:2021jnm}. Physically, these free field realisations should be thought of as encoding the operator algebra of an SCFT in terms of degrees of freedom that survive in the effective field theory on the moduli space after (partial) Higgsing. This intuition (and the examples that bear it out) has thus far been primarily restricted to the setting of VOAs associated to four-dimensional $\mathcal{N}=2$ SCFTs, but it seems likely that the principle holds somewhat more generally.

In this paper, we will be investigating the application of this free field philosophy to a family of VOAs that are exotic from a four-dimensional perspective, but potentially more conventional in the three-dimensional setting. These are the \emph{chiral universal centralisers} defined in \cite{Arakawa:2018egx}, which were previously identified as the VOAs associated to the \emph{sphere} in the setting of chiral algebras of class $\mathcal{S}$ \cite{Beem:2014rza}. The associated variety of this VOA is the \emph{universal centraliser} $\Zf_G$, the definition of which we recall in Section \ref{sec:universal_centraliser}---this is the holomorphic-symplectic variety associated to the sphere by the Moore--Tachikawa topological field theory for class $\mathcal{S}$ Higgs branches \cite{Moore:2011ee}. While in the context of class $\mathcal{S}$ the sphere is a somewhat pathological case, the universal centraliser appears in a more physically sensible way as the Coulomb branch of pure $\mathcal{N}=4$ gauge theory in three dimensions. Thus the chiral quantisation of this space might naturally be expected to appear as a boundary VOA in the sense of \cite{Costello:2019} for pure gauge theory.

Our approach to constructing a candidate free field realisation for the chiral universal centraliser will follow the general theme of \cite{Beem:2019tfp}. We start with a detailed understanding of the Poisson geometry of the universal centraliser, an introduction to which we provice in Section \ref{sec:universal_centraliser}. In Section \ref{sec:kostant_toda_lattice}, we give an explicit characterisation of a map, $\varphi$, from the Kostant--Toda lattice to a Zariski-open subset of $\Zf_G$. Such a map was known to Kostant for the case when $G$ is of adjoint type, and in Theorem~\ref{thm:Kostant_level_set} we construct a generalisation of~\cite[Theorem 2.6]{Kostant:1979195} for any semi-simple $G$. This gives rise to Propositions~\ref{prop:phi_open_immersion} and \ref{prop:varphi-symplectic}, in which we show that $\varphi$ is an open immersion and symplectic.

Altogether, this gives a finite-dimensional analogue of a free field realisation for the Poisson algebra of functions $\mathbb{C}[\Zf_G]$. In Section \ref{sec:chiralising_centraliser} we propose a chiralisation of this construction in terms of a lattice vertex algebra that we call the chiral Kostant--Toda lattice. In particular, we define a sub-VOA of this lattice VOA (by specifying a weak generating set) that we conjecture to match the chiral universal centraliser as defined by Arakawa. In defining this sub-VOA, we encounter an (apparently novel) generalisation of the Miura transform, and in Remark~\ref{rem:deformed_Miura} we explain how that the usual Miura transform---realising principal $\WW$-algebras in products of Heisenberg vertex algebras---can be generalised to a family of such embeddings based on the lower central series of the upper nilpotent subalgebra of $\gf$.

We verify that our construction reproduces the chiral universal centraliser (defined \emph{via} Drinfel'd--Sokolov reduction) explicity for the case $G={\rm SL}_2$, and we also study the $G={\rm SL}_3$ case in detail and confirm that the VOA according to our prescription is indeed a chiral quantisation of universal centraliser. Here a direct comparison to the definition in terms of DS reduction would require a detailed analysis of said reduction, and this has not appeared in the literature to date so far as the authors are aware. Some interesting structure regarding gradings and simplicity of these VOAs become apparent in our examples, and we elaborate on this in Section~\ref{subsec:interlude}, where we define a ``more universal'' centraliser that returns the universal centraliser by central quotient.

An additional benefit of our free field construction is that for chiral algebras of class $\mathcal{S}$ corresponding to physical theories (so with three or more punctures), the sphere free field realisation can act as the starting point for free field realisations that treat all punctures on equal footing (and so manifest generalised $S$-duality at the level of FFRs). In Section \ref{sec:FF_genus_zero} we consider the case of $\mathfrak{a}_1$ theories with various numbers of punctures. We give free field realisations for these vertex algebras that incorporate the sphere FFR, which leads to a very uniform construction. Analogous realisations for other class $\mathcal{S}$ families should arise in much the same way. We conclude in Section \ref{sec:conclusions} with several comments and open questions. The proof of Theorem~\ref{thm:Poisson_generation}, regarding Poisson generation in the ${\rm SL}_N$ case, is relegated to Appendix \ref{app:proof_Poisson}.

The OPE computations in this paper were performed using the \texttt{OPEdefs} package, which was developed by Kris Thielemans \cite{Thielemans:1994er}, for \texttt{Mathematica}.

%----------------------------------------------------------------------%
\section{\label{sec:universal_centraliser}The universal centraliser}
%----------------------------------------------------------------------%

The universal centraliser admits several equivalent definitions, some of which arise naturally by way of three- and four-dimensional physics. In this section we will review these definitions and some important properties of the universal centraliser, and also establish notation for what follows.

%----------------------------------------------------------------------%
\subsection{\label{subsec:defns_univ}Background and definitions}
%----------------------------------------------------------------------%

Let $G$ be a simple algebraic Lie group and $\gf=\mathrm{Lie}(G)$. Let $\gf^\ast$ be the linear dual of $\gf$, and using the Killing form $(\cdot,\cdot)$ one may identify $\gf$ with $\gf^\ast$ via the isomorphism $x\mapsto (x,\cdot)$ for $x\in\gf$.

Furthermore, let $B\subset G$ be a Borel subgroup with splitting $B=N\rtimes T$ for a maximal torus $T$ and unipotent group $N$. The Lie algebra $\gf$ admits a Cartan decomposition $\gf= \mathfrak{n_+}\oplus \mathfrak{t}\oplus \mathfrak{n_-}$, where $\mathfrak{t}=\mathrm{Lie}~T$ and $\mathfrak{n}_+ = \mathrm{Lie}~N$. (For future reference, we also denote by $N_-$ the lower unipotent group with $\mathfrak{n}_- = \mathrm{Lie}~N_-$.)

We denote adjoint and coadjoint actions of $G$ on $\gf$ by $\mathrm{Ad}$ and $\mathrm{Ad}^\ast$, respectively. Similarly, $\mathrm{ad}$ and $\mathrm{ad}^\ast$ denote the respective $\gf$-actions. The stabiliser, $G_x$, of an element $x\in\gf^\ast$ is the subgroup of $G$ that leaves $x$ invariant under the coadjoint action,
%----%
\begin{equation}
\label{eq:stabiliser_def}
    G_x = \{g\in G |\, \mathrm{Ad}^\ast_g\, x=x \}~.
\end{equation}
%----%
The Zariski open subset $\gf^\ast_{\rm reg}\subset \gf^\ast$ is defined to be the set of all elements whose centralisers have dimension equal to the rank of $\gf$. This is the complement of the union of non-principal nilpotent orbits in $\gf^\ast$, or alternatively, the union of all non-nilpotent orbits along with the principal nilpotent orbit in $\gf^\ast$.

Let $T^\ast G$ be the total space of the cotangent bundle of $G$; one has the trivialisation $T^\ast G\cong \gf^\ast\times G$ using the realisation of $\gf$ by left-invariant vector fields on $G$. On the cotangent bundle $T^\ast G$ there are two commuting, Hamiltonian $G$ actions, with moment maps $\mu_L$ and $\mu_R$ given by
%----%
\begin{equation}
\label{eq:T*G_moment_map}
    \mu_L(g,x) = x~,\qquad \mu_R(g,x) = \mathrm{Ad}_g^*\,x~,
\end{equation}
%----%
for $(g,x)\in G\times \gf^*\cong T^\ast G$. Let $\mu_{\rm diag} = \mu_L - \mu_R$ be the moment map for the diagonal $G$ action. The universal centraliser is then defined as the following diagonal Hamiltonian reduction \cite{Bezrukavnikov:2003},\!\footnote{Note that here our $\Zf_G$ is denoted by $\Zf^G_\gf$.}
%----%
\begin{equation}
\label{eq:centraliser_as_diagonal_DS}
	\Zf_G = \big(\mu_{\rm diag}^{-1}(0)\cap (\gf^\ast_{\rm reg}\times G)\big)/\!\!/G_{\rm diag}~.
\end{equation}
%----%
This definition gives the universal centraliser the structure of a Poisson variety (and in fact a symplectic variety). Furthermore, $\Zf_G$ is naturally a scheme over the orbit space $\gf^\ast_{\rm reg}/\!\!/G\cong \gf^*/\!\!/G$. We let $(P_i)_{i=1}^{\rk~\gf}$ denote the generators of the $\IC[\gf^\ast]^G$ (the \emph{Harish-Chandra centre}) corresponding to the fundamental invariants of $\gf$.

In terms of class $\SS$ physics, $\Zf_G$ can be formally identified as the Higgs branch associated to an unpunctured sphere. While there is no physical SCFT in four-dimensions for which this is the Higgs branch, one can still interpret such an object in terms of, \eg, the Moore--Tachikawa TQFT.\!\footnote{Or alternatively, in terms of compactification of the $(2,0)$ theory on a Riemann surface of finite size, as in \cite{Gaiotto:2011xs}.} In that setting, the operation of closing full punctures on the UV curve is implemented by performing Kostant--Whittaker (KW) reduction on the associated symplectic variety with respect to the Hamiltonian $G$-actions associated to the closed punctures. The sphere can then be obtained, for instance, by closing two punctures on the cylinder, which itself is associated to $T^\ast G$. This leads to another definition of $\Zf_G$ as the two-sided KW reduction of the cotangent bundle $T^*G$, as seen in~\cite{Ginzburg:2018Nil}. This perspective will also have the most natural link to the chiralisation of $\Zf_G$.

Let us make this characterisation more concrete. Let $\Delta$ denote the set of (positive) simple roots of $\gf$, and let
%----%
\begin{equation}
\label{eq:principal_lowering_def}
	f = \sum_{\alpha\in\Delta}e_{-\alpha}~,
\end{equation}
%----%
be the sum of the negative simple roots. By the Jacobson--Morozov theorem, such an $f$ can be completed to an $\slf_2$-triple $(e,h,f)$. Let $\chi\in\gf^*$ be the (Killing) dual of $f$.

The moment maps on $T^*G$ can be restricted to $\nf_-^*$-valued moment maps by means of the composition $\bar{\mu}_{L,R}:T^*G \rightarrow \gf^* \twoheadrightarrow \gf^*/\mathfrak{b}_+^*\cong \nf^*_-$, where $\mathfrak{b}_+ = Lie(B)$ is our Borel subalgebra. The maximal (upper) unipotent subgroups of the left and right $G$-actions are denoted by $N_L$ and $N_R$, respectively. The KW realisation of the universal centraliser then takes the following form,
%----%
\begin{equation}
\label{eq:double_KW_reduction}
    \Zf_G\cong T^* G \overset{~\chi}{\,/\!\!/\,} (N_L\times N_R) = (\bar{\mu}^{-1}_L(\chi)\cap\bar{\mu}^{-1}_R(\chi))/\!\!/(N_L\times N_R)~,
\end{equation}
%----%
which again encodes $\Zf_G$ as a symplectic variety.

Now let $\gf^e\subset \gf$ be the $\mathrm{ad}_e$-invariant subspace. The \emph{Slodowy slice} $S_f\subset\gf^\ast_{\rm reg}\subset\gf^\ast$ is the image of $f+\gf^e$ under the Killing isomorphism, or more concretely,
%----%
\begin{equation}
\label{eq:slodowy_def}
    S_f = \chi_f + \{(x,\cdot)\in \gf^* |\,x\in\gf^e \}~.
\end{equation}
%----%
All $G$-orbits in $\gf^\ast_{\rm reg}$ intersect $S_f$ transversally~\cite{Kostant:1978}, so in particular,
%----%
\begin{equation}
\label{eq:greg_from_slodowy}
    \gf^\ast_{\rm reg} = G\,\cdot S_f~.
\end{equation}
%----%
This structure leads to Kostant's isomorphism $\gf^*\,/\!\!/\,G \cong S_f$. At the level of varieties (schemes), performing KW reduction on a Hamiltonian $G$-space $X$ (with respect to the principal nilpotent conjugacy class) is equivalent to taking the fibre-product of $X$ with $S_f$ over $\gf^*$, see~\cite{Ginzburg:2008HC}. We can perform two-sided KW reduction in stages, and after taking just the reduction by the left action we have
%----%
\begin{equation}
\label{eq:first_KW_reduction}
	T^*G\overset{\chi}{\,/\!\!/\,}N_L \cong T^*G\!\!\!\underset{\mu_L,\gf^*}{\times}\!\!\!S_f = G\times S_f~.
\end{equation}
%----%
The variety $G\times S_f$ (with the symplectic structure induced by KW reduction) is also known as the equivariant Slodowy slice \cite{Losev1, Losev2}. This inherits the right-action moment map, and performing the second reduction then gives
%----%
\begin{equation}
\label{eq:second_KW_reduction}
\begin{split}
	\Zf_G\cong(G\times S_f)\overset{\chi}{\,/\!\!/\,}N_R &\cong (G\times S_f)\underset{\mu_R,\gf^*}{\times}S_f \\
	&\cong \{(g,s)\in G\times S_f \,|\, \mathrm{Ad}^*_g s \in S_f\}~.
\end{split}
\end{equation}
%----%
This gives a definition of $\Zf_G$ as the group scheme centralising points in $S_f$, so in particular as a group scheme over $S_f\cong \IC^{\rk\,\gf}$. $S_f$ can then be thought of as the base of an integrable system on $\Zf_G$ with action variables given by the (restriction to $S_f$ of the) generators of the Harish-Chandra centre.

The universal centraliser of type $G$ also arises in \emph{three-dimensional} physics as the Coulomb branch of pure $\NN=4$ gauge theory with gauge group $\check{G}$, the Langlands dual of $G$. (This is the three-dimensional mirror for the circle compactification of the class $\SS$ theory associated to the two-sphere.) The Coulomb branch of pure $\SU(2)$ gauge theory was identified with the Atiyah--Hitchin manifold in~\cite{Seiberg:1996nz}, and based in part upon this and further evidence from \cite{Argyres:1994xh,Martinec:1995by}, Teleman predicted in~\cite{Teleman:2014jaa} that the Coulomb branch of pure $\mathcal{N}=4$ gauge theory for any $\check{G}$ could be identified with the (spectrum of the) Borel--Moore equivariant homology ring of the affine Grassmannian $Gr_{\check{G}}$, which had been computed earlier by Bezrukavnikov--Finkelberg--Mirkovi\'c (BFM)~\cite{Bezrukavnikov:2003}. This definition, in terms of the affine Grassmannian, is confirmed by the construction of three-dimensional Coulomb branches by Braverman--Finkelberg--Nakajima (BFN)~\cite{Nakajima:2015txa,Braverman:2016wma, Braverman:2017ofm, Braverman:2016pwk}.\!\footnote{A physics-motivated computation of the Coulomb branches of $\NN=4$ gauge theories by abelianisation techniques appeared contemporaneously in~\cite{Bullimore:2015lsa}.}

For completeness, we recall this definition here. Let $\check{G}(\OO)= {\rm Hom}(\ID,G)$ and $\check{G}(\KK) = {\rm Hom}(\ID^\times , G)$ be the spaces of homorphisms from the formal disc, $\ID= \mathrm{Spec}~\IC[\![t]\!]$ and the formal punctured disc, $\ID^\times = \mathrm{Spec}~\IC(\!(t)\!)$, respectively, to $\check{G}$. The affine Grassmannian is the quotient,
%----%
\begin{equation}
\label{eq:affine_grassmannian}
    \mathrm{Gr}_{\check{G}} = \check{G}(\KK)/\check{G}(\OO)~.
\end{equation}
%----%
The universal centraliser is then identified with (the spectrum of) the $\check{G}(\OO)$-equivariant Borel--Moore homology,
%----%
\begin{equation}
\label{eq:universal_centraliser_from_affine_grassmannian}
    \Zf_G \cong \mathrm{Spec}~ H_{\bullet}^{\check{G}(\OO)}(\mathrm{Gr}_{\check{G}})~.
\end{equation}
%----%
From this definition, realising the symplectic structure on $\Zf_G$ is a more elaborate procedure than in the previous classical definitions.

%----------------------------------------------------------------------%
\subsection{\label{subsec:eg_rank_one}Examples in rank one}
%----------------------------------------------------------------------%

To make these spaces a bit more tangible, we consider in detail the two rank-one examples: $G={\rm SL}_2$ and $G={\rm PGL}_2$.

By the Killing isomorphism, $S_f$ is identified with the elements of $\slf_2$ of the form
%----%
\begin{equation}
\label{eq:sl2_slodowy_slice}
	\begin{pmatrix}
	\,0\, & \,S\,\\
	\,1\, & \,0\,
	\end{pmatrix}
	~,\quad S\in\IC~.
\end{equation}
%----%
The most general element of ${\rm SL}_2$ that commutes with such a matrix is of the form
%----%
\begin{equation}
\label{eq:sl2_centraliser_matrix}
	\sqrt{-1}\begin{pmatrix}
	Y & S X\\
	X & Y
	\end{pmatrix}
	~,\quad X,Y,S\in \IC~,
\end{equation}
%----%
where the overall factor of $\sqrt{-1}$ is purely conventional. Therefore, $\Zf_{{\rm SL}_2}\subset\mathbb{A}^3$ can be realised as an affine variety with coordinate ring
%----%
\begin{equation}
\label{eq:sl2_centraliser_coordinate_ring}
	I_{{\rm SL}_2} = \frac{\IC[X,Y,S]}{\langle SX^2 - Y^2 - 1\rangle}~,
\end{equation}
%----%
which is the presentation appearing in~\cite{Bezrukavnikov:2003}.

The Poisson bracket can be computed by carrying out the two sided KW reduction of $T^*{{\rm SL}_2}$.\!\footnote{One can also obtain the Poisson brackets from the blow-up computation of~\cite{Bezrukavnikov:2003}.}
%----%
\begin{equation}
\label{eq:sl2_centraliser_poisson_brackets}
	\{S,X\} = Y~,\quad \{S,Y\}= SX~,\quad \{Y,X\}=\frac{1}{2}X^2~,
\end{equation}
%----%
with all other Poisson brackets vanishing. We observe that as a Poisson algebra, $I_{{\rm SL}_2}$ is generated by just $S$ and $X$. This point will arise again later when we discuss the general case for ${\rm SL}_N$.

Now for the adjoint case $\Zf_{{\rm PGL}_2}$, we have
%----%
\begin{equation}
\label{eq:pgl2_centraliser_quotient}
	\Zf_{{\rm PGL}_2} \cong Z{({\rm SL}_2})\backslash\Zf_{{\rm SL}_2}~,
\end{equation}
%----%
where the left action of $Z({{\rm SL}_2})= \IZ/2\IZ$ is given by left multiplication by $-\id$ on the fibres. On functions this action is given by
%----%
\begin{equation}
\label{eq:pgl2_center_action}
	X\mapsto -X~,\quad Y\mapsto - Y~.
\end{equation}
%----%
Thus,
%----%
\begin{equation}
\label{eq:pgl2_centraliser_relations}
	I_{{\rm PGL}_2} \cong \frac{\IC[S,X^2,XY,Y^2]}{\langle SX^2 - Y^2 -1\rangle}\cong \frac{ \IC[S,A,B,C]}{\langle B^2 - AC,SA - C-1\rangle}~.
\end{equation}
%----%
The Poisson structure can be obtained by pullback along the quotient map $\Zf_{{\rm SL}_2} \rightarrow Z({{\rm SL}_2})\backslash \Zf_{{\rm SL}_2} \cong \Zf_{{\rm PGL}_2}$.

%----------------------------------------------------------------------%
\subsection{\label{subsec:eq_rank_two}Example in rank two}
%----------------------------------------------------------------------%

It will be useful in what follows to also have an example in rank two at our disposal. Here we choose the universal centraliser for $G={\rm SL}_3$.

Once again by using the Killing isomorphism, the Slodowy slice $S_f$ can be identified as an affine subspaces in $\gf=\slf_3$ of the form
%----%
\begin{equation}
\label{eq:sl3_slodowy_slice}
	\begin{pmatrix}
	0 & \frac12 S & W\\
	1 & 0 &\frac12 S \\
	0 & 1 & 0
	\end{pmatrix}
	~,\quad S,W\in \IC~.
\end{equation}

Any element of ${\rm SL}_3$ that stabilises a point in the slice can be written as
%----%
\begin{equation}
\label{eq:sl3_centraliser_matrix}
	\begin{pmatrix}
	Z & WX + \frac12 SY & \frac14 S^2X+ W Y \\
	Y & \frac12 SX +  Z &  W X + \frac12 S Y \\
	X & Y & Z
	\end{pmatrix}~,
\end{equation}
%----%
with $X$, $Y$, and $Z$ satisfying the following determinant relation,
%----%
\begin{equation}
\label{eq:sl3_determinant_relation}
\begin{split}
	D_{{\rm SL}_3}=  &-\frac{1}{8} S^3 X^3 + W^2 X^3 + \frac{1}{2} S W X^2 Y + \frac{1}{2} S^2 X Y^2 + W Y^3 - \frac{1}{4} S^2 X^2 Z\\ & - 3 W X Y Z - S Y^2 Z + \frac{1}{2} S X Z^2 + Z^3 \overset{!}{=}1~.
\end{split}
\end{equation}
%----%
Thus, we can identify $\Zf_{{\rm SL}_3}$ as an affine variety in $\mathbb{A}^5$ with coordinate ring
%----%
\begin{equation}
\label{eq:sl3_centraliser_coordinate_ring}
	I_{{\rm SL}_3} \cong \frac{\IC[S,W,X,Y,Z]}{\langle D_{{\rm SL}_3}-1\rangle}~.
\end{equation}
%----%
The Poisson structure can again be computed from the two-sided KW reduction, giving
%----%
\begin{alignat}{3}
\label{eq:sl3_Poisson}
	\{S,X\} &= -Y~,&\quad \{S,Y\} &= -Z-\frac{1}{2} SX~,&\quad \{S,Z\} &= -SY - WX~,\\
	\{W,X\} &= -Z +\frac{1}{6} S X~,&\quad \{W,Y\} &= -WX -\frac{1}{3}S Y~,&\quad \{W,Z\} &= - WY +\frac{1}{6} S Z -\frac{1}{4} S^2 X~,\\
	\{Y,X\} &=  \frac{2}{3}X^2~,& \quad \{Z,X\} &= -\frac{5}{6} XY ~,&\quad \{Z,Y\} &= -\frac{1}{3} Y^2 + \frac{1}{6} XZ -\frac{1}{4} SX^2~.
\end{alignat}
%----%
Observe again that $I_{{\rm SL}_3}$ is generated, as a Poisson algebra, by the generators $S,W$ of the Harish-Chandra centre and the coordinate $X$ that is conjugate to the bottom left corner of the ${\rm SL}_3$ matrix. The algebraic generator $Z$ is a shift of the Poisson bracket $\{S,Y\}$, but of course one is free to use the generating set $\big(S,W,X,\{S,X\},\{S,\{S,Y\}\}\big)$ instead.

%----------------------------------------------------------------------%
\subsection{\label{subsec:sln_poisson_structure}Poisson algebraic structure for \texorpdfstring{${\rm SL}_N$}{SLN}}
%----------------------------------------------------------------------%

In this section, we will codify some patterns we have observed in the examples of ${\rm SL}_2$ and ${\rm SL}_3$. We generalise these to statements about ${\rm SL}_N$ and discuss their further generalisation to arbitrary simple algebraic Lie groups.

We realise ${\rm SL}_N$ as the closed subgroup of ${\rm GL}_N$ satisfying the unit determinant condition. Let $(g_{ij})_{i,j=1}^{N}$ be the linear coordinates on ${\rm GL}_N$ that are conjugate to the matrix entries. These restrict to linear coordinates functions on ${\rm SL}_N$.

In both of the examples ${\rm SL}_2$ and ${\rm SL}_3$, the ring $I_G$ was generated by the Slodowy slice generators and the matrix entries in the bottom row of the group element. This observation generalises as follows.

%-------------------------------------------%
\begin{prop}\label{prop:ring_generators_sln}
Let $s\in S_f$. An element, $g$, of $\,{\rm SL}_N$ that stabilises $s$ is completely fixed by specifying the matrix entries of the bottom row of $g$.
\end{prop}
%-------------------------------------------%

In other words, if $g_{ij}$ are the coordinate functions conjugate to the matrix entries, then $I_{{\rm SL}_N}$ is generated, as a ring, by the restrictions of the fundamental invariants to $S_f$, $(P_i)_{i=1}^{N}$, and by $(g_{Ni})_{i=1}^{N}$.

%-------------------------------------------%
\begin{proof}
Let $g\in {\rm SL}_N$ and $s\in S_f$, then the condition $\mathrm{Ad}^\ast_g s = s$ is equivalent to requiring $[g,s]=0$, where we think of $g$ as a matrix of $\mathfrak{gl}_N$.

The kernel of $[-,s]$ in $\mathfrak{gl}_N$ is $N$-dimensional since $S_f\subset \mathfrak{gl}_N$ is regular. The proposition will follow if we can show that there is a basis, $(v_i)_{i=1}^{N}$, of $\mathrm{ker}~[-,s]$ consisting of matrices $v_i$ such that $(v_i)_{Nj}=\delta_{ij}$. In words, $v_i$ is an element of $\mathrm{ker}~[-,s]$ where the $i$th entry along the bottom row is equal to one and the rest of the bottom row entries vanish.

We can explicitly construct such a basis. First, let $v_N = \id$,~\ie, the identity matrix. We also fix $v_{N-1}= s$. It is clear that both $v_{N-1}$ and $v_{N}$ commute with $s$ and satisfy the required conditions on the bottom row.

Now consider the matrix product, $s^k$, the entries $(s^k)_{Ni}$ are zero for $i<N-k$ and for $i=N-K+1$. The entry $(s^k)_{N,N-k}=1$ and the entries $s^k_{Ni}$ for $k>N-k$ are generically nonzero. By induction, we can take linear combinations,
%----%
\begin{equation}
\label{eq:sln_commuting_matrices_from_slodowy}
	s^k + \sum_{j=0}^{k-2}c_j s^j~,
\end{equation}
%----%
for $c_j\in \IC$ to set all the bottom entries to zero, except the $(N,N-k)$ entry. We set $v_{N-k}$ equal to this linear combination.

This provides a linearly independent set $(v_i)_{i=1}^N$ which all lie in the kernel of $[-,s]$, since $s$ commutes with all powers $s^k$. Therefore the $v_i$ form a basis of $\mathrm{ker}~[-,s]$ and we are done.
\end{proof}
%-------------------------------------------%

Later on this Proposition will yield a significant simplification in the procedure for computing the ring $I_G$ in terms of Kostant--Toda lattice generators. There is also another simplification to be had.

%-------------------------------------------%
\begin{thm}\label{thm:Poisson_generation}
The Poisson algebra $I_{{\rm SL}_N}$ is Poisson generated by $(P_i)_{i=1}^{\rk\,\gf}$ and the generator $X =g_{N1}$.
\end{thm}
%-------------------------------------------%
%-------------------------------------------%
\begin{proof}
The proof of this theorem relies upon machinery that will be introduced in later sections. As such we relegate it to Appendix~\ref{app:proof_Poisson}.
\end{proof}
%-------------------------------------------%

We expect a generalisation of these results to any simple algebraic Lie group $G$ to hold as well. To justify this, let us characterise this preferred element $X$ in more representation theoretic terms. By the Peter--Weyl decomposition, the coordinate ring of $G$ decomposes into $G-G$ bi-modules according to
%----%
\begin{equation}
\label{eq:peter_weyl_functions}
	\IC[G] \cong \bigoplus_{\lambda\in P^+} V_{\lambda}\otimes V_{\lambda}^*~,
\end{equation}
%----%
where $P^+$ is the set of integral dominant weights of $\gf$, $V_{\lambda}$ is the finite-dimensional representation of $\gf$ with highest weight $\lambda$ and $V_{\lambda}^*$ is the contragredient dual representation. The vector space of functions on $T^*G\cong \gf^*\times G$ therefore takes the form
%----%
\begin{equation}
\label{eq:peter_weyl_cotangent_functions}
	\IC[T^*G] = \mathrm{Sym}~\gf~\otimes \bigoplus_{\lambda\in P^+} V_\lambda\otimes V_{\lambda^*}~.
\end{equation}
%----%
Let $V_{F}$ be the lowest dimensional faithful representation of $G$. If $G$ is not of type $D_{2n}$ then $V_{F}$ is irreducible, so let us make this assumption for the moment. Since $V_F$ is faithful, every irreducible representation $V_\lambda$ will appear in the decomposition of the tensor powers of $V_F$. In light of the Peter-Weyl decomposition, this means that one can choose as generators of $\IC[G]$ the functions transforming in the representation $V_F\otimes V_F^*$ in \eqref{eq:peter_weyl_functions}. For example, if $G={\rm SL}_n$ then $V_F$ is the defining representation. One usually picks the linear coordinates $\{g_{ij}\}$---conjugate to the matrix entries---to be generators of $\IC[{\rm SL}_n]$, and indeed these constitute precisely the term $V_F\otimes V_F^*$ in \eqref{eq:peter_weyl_functions}.

Let $v_F\in V_F$ be the highest-weight vector, then $V_F$ is generated, as a $\gf$-module by $v_F$. Now, suppose $\widetilde{X}$ is the function in $\IC[G]\subset\IC[T^*G]$ transforming in the highest weight state $v_F\otimes v_F^*$. Then $\widetilde{X}$ is not in the ideal generated by $\overline{\mu}_L$ and $\overline{\mu}_R$, and moreover it is $N_L\times N_R$ invariant---since it is the highest-weight state. As a result, $\widetilde{X}$ will descend unmodified to the two-sided KW reduction of $T^*G$. For ${\rm SL}_N$, $\widetilde{X}$ is the linear coordinate function conjugate to the matrix entry $g_{N1}$, and after KW reduction it will become the function $X$ introduced in Theorem~\ref{thm:Poisson_generation}.

The bottom row $(g_{Ni})_{i=1}^N$ that algebraically generate (along with the fundamental invariants) $I_{\rm SL_N}$ are identified with the vector subspace $v_F\otimes V_{F}^*$ of $V_{F}\otimes V_{F}^*$. This provides a more abstract characterisation of the algebraic generators of $I_{\rm SL_N}$. For general $G$ these functions are $N_L$ invariant and can be made $N_R$ invariant by corrections that are linear in the generators of $v_f\otimes V_F^*$. We should therefore expect a correspondence between (a subset) of the algebraic generators of $I_G$ and the weight vectors of $v_F\otimes V_F^*$.

The Poisson algebra $\IC[T^*G]$ has two Poisson-commuting subalgebras, each isomorphic to $\mathrm{Sym}~\gf$, arising from the moment maps for the left and right $G$-actions. The subspace $V_F\otimes V_F^*\subset \IC[T^*G]$ is generated by the (Poisson) action of the left and right moment maps on $\widetilde{X}$, and since $V_F\otimes V_F^*$ generates $\IC[G]$, we can choose $\widetilde{X}$ and the generators of the left and right moment map as Poisson generators of $\IC[T^*G]$. After two-sided KW reduction, what remains of the moment maps are the generators of $\IC[\gf^*]^G\cong \IC[S_f]$. It seems reasonable, then, to conjecture that $X$ and $\IC[S_f]$, which are in an appropriate sense the reductions of the Poisson generators of $\IC[T^*G]$, will be Poisson generators of the reduced algebra. 

Let us now return to the case where $G$ is of type $D_{2N}$. If $G$ is not simply connected, then $Z(G)$ is cyclic and our previous description still holds and $V_F$ is irreducible. If $G$ is of type $D_{2N}$ and simply connected then $V_F$ is not irreducible---it is the Dirac spinor, $V_{\rm s}\oplus V_{\rm c}$, where $V_{\rm s}$ and $V_{\rm c}$ are the smallest spin representations of $\rm Spin(N)$. In this case, we must pick two such highest-weight states, $X_s$ and $X_c$ in representations $V_{\rm s}\otimes V_{\rm s}^*$ and $V_{\rm c}\otimes V_{\rm c}^*$ respectively. Similarly, the $N_L$ invariant subspace is now $(v_{\rm s}\otimes V_{\rm s}^*)\oplus (v_{\rm c} \otimes V_{\rm c}^*)$, where $v_s$ and $v_c$ are highest weight states of $V_{\rm s}$ and $V_{\rm c}$ respectively.

Putting this all together, we have the following conjectures on the algebraic and Poisson algebraic structures of $I_G$.

%-------------------------------------------%
\begin{conj}
If $G$ is not simply connected, or if $G$ is not of type $D_{2N}$, then $I_G$ is generated, as a ring, by  the fundamental invariants $(P_{i})_{i=1}^{\rk\,\gf}$ and generators $(g_i)_{i=1}^{\dim\,V_F}$ which are in one-to-one correspondence with the weight vectors of $ V_F^*$.

If $G$ is simply connected and of type $D_{2N}$, $I_G$ is generated, as a ring, by the fundamental invariants $(P_{i})_{i=1}^{\rk\,\gf}$ and two sets of generators $(g_{{\rm s}i})_{i=1}^{\dim\,V_{\rm s}}$, and $(g_{{\rm c}i})_{i=1
}^{\dim\, V_{\rm c}}$. The generators $g_{{\rm s}i}$ (repectively $g_{{\rm c}i}$) are in one-to-one correspondence with the weight vectors of $ V_{\rm s}^*$ (repectively $ V_{\rm c}^*$).
\end{conj}
%-------------------------------------------%

%-------------------------------------------%
\begin{conj}\label{conj:poisson_generation}
If $G$ is not simply connected, or if $G$ is not of type $D_{2N}$, then $I_G$ is Poisson generated by the fundamental invariants $\big(P_i\big)_{i=1}^{\rk\,\gf}$ and $X$.

If $G$ is simply connected and of type $D_{2N}$, \ie, if $G= \mathrm{Spin}(2N)$, then $I_G$ is Poisson generated by the fundamental invariants $\big(P_i\big)_{i=1}^{\rk\,\gf}$, $X_s$, and $X_c$.
\end{conj}
%-------------------------------------------%
%----------------------------------------------------------------------%

%----------------------------------------------------------------------%
\section{\label{sec:kostant_toda_lattice}The universal centraliser and the Kostant--Toda lattice}
%----------------------------------------------------------------------%

The geometric structure that apparently underlies the free field realisations in~\cite{Beem:2019tfp} is the existence of some particularly simple Zariski open subset(s) $U$ of the Higgs branch $X$ in terms of which the global functions on $X$ can be realised via the inclusion $\IC[X]\hookrightarrow \IC[U]$. These subsets are generally identified (including holomorphic symplectic structure) with cotangent bundles of the form $(T^\ast\IC^\times)^n\times(T^\ast\IC)^m$ with canonical Poisson structures,\!\footnote{As in \cite{Beem:2019tfp}, we take our standard Poisson structure on $T^\ast\IC^\times$ to be $\{p,e\}=e$ where $p$ is the fibre coordinate and $e$ is the base coordinate.} so these can be thought of as Darboux (Zariski) neighbourhoods in $X$. In practice, the expressions for the images of generators of $\IC[X]$ in terms of $\IC[U]$ can often be inferred using symmetry data, such as weights under the contracting $\IC^\ast$ action on $X$ and the action of additional global symmetries. Such ``finite-dimensional free field realisations'' form the foundation for a chiral analogue, where the associated VOA is realised in terms of $\beta\gamma$ systems and isotropic lattice bosons that naturally chiralise $\IC[U]$. (We will say more about these isotropic lattice VOAs in Section \ref{sec:chiralising_centraliser}.)

In previous examples in the literature, the open set $U$ has been defined with reference to a (partial) Higgsing of some global symmetry acting on $X$. However, it is not at all clear that there is anything fundamental about this particular origin for such Darboux neighborhoods, and a more general story may certainly exist. Indeed, for the case of the universal centraliser, there are no global symmetries to speak of, so a new starting point is needed.

Fortunately, it has been known since work of Kostant (see~\cite{Kostant:1979195}) that the integrable system of the Kostant--Toda lattice of type $\gf$ and the universal centraliser $\Zf_G$ (for $G$ of adjoint type) are intimately connected, with the former having precisely the Darboux structure needed. In particular, the level sets of the Kostant--Toda lattice (as an integrable system) can be identified with (open) subsets of stabilisers of elements in $\gf^\ast_{\rm reg}$. Recently, Crooks has provided a complex-analytic embedding of an open dense subset of the Kostant--Toda lattice inside the universal centraliser~\cite{Crooks:2020}, still with the caveat that $G$ be of adjoint type.

In this section, we will adapt this construction to the algebraic setting and extend the domain of the morphism to the full Kostant--Toda lattice. Generalising further to allow for any simple $G$ (including, \eg, both adjoint and simply-connected forms of the group) will require us to replace the Kostant--Toda lattice with a certain covering space that we define. Our main result is the construction of an open immersion from the (cover of the) Kostant--Toda lattice to the universal centraliser of $G$, and thus a finite-dimensional free field realisation of the universal centraliser.

%----------------------------------------------------%
\subsection{\label{subsec:KT_Lattice}The Kostant--Toda lattice}
%----------------------------------------------------%

Consider a dynamical system of $N$ particles on a line, with Hamiltonian
%----%
\begin{equation}
\label{eq:KT_Hamiltonian}
	\mathcal{H}_{KT} = \sum_{i=1}^{N} \frac{p_i^{2}}{2} +\sum_{i=1}^{l}  e^{\psi_i(\mathbf{q})}~,
\end{equation}
%----%
where $\psi_i(\mathbf{q})$ are linear functions of the canonical position coordinates $\mathbf{q}=(q_1,q_2,\dots,q_N)$. There is a natural inner product on the vector space spanned by the $\psi_i$~\cite{Kostant:1978}. When this inner product space agrees with the root system of some real Lie algebra $\gf_{\IR}$ endowed with the Killing form, the $\psi_i$ define a set of roots for $\gf_{\IR}$, and in this case the representation theory of $\gf_{\IR}$ can be leveraged to integrate Hamilton's equations. Indeed, the system is integrable in the classical sense.

The complexified phase space $\mathrm{KT}_\gf$ of the Kostant--Toda Lattice of type $\gf_{\IR}$ can be identified (adopting the notation of the previous section) with an affine subspace of $\gf^\ast\colonequals\gf_\IC^\ast$,
%----%
\begin{equation}
\label{eq:KT_phase_space}
	\mathrm{KT}_\gf = \chi + \mathfrak{t}^* +\sum_{\alpha\in\Delta} \gf^*_{\alpha}\backslash\{0\}~.
\end{equation}
%----%
Let us denote by $(\alpha_i)_{i=1}^{\rk~\gf}$ the simple roots of $\gf$, by $e_{\alpha_i}$ the corresponding simple root elements of $\mathfrak{n}_+$, and by $h_{i}\in\mathfrak{t}$ the corresponding Cartan elements. Let $h_i^*$ and $e_{\alpha_i}^*$ be the Killing duals to $h_i$ and $e_{\alpha_i}$, respectively. An element $y\in \KT_{\gf}$ is of the form
%----%
\begin{equation}
\label{eq:point_in_KT}
	y = \chi + \sum_{i=1}^{\rk~\gf}b_i(y)h_i^* + \sum_{i=1}^{\rk~\gf} \gamma_i(y) e_{\alpha_i}^*~,
\end{equation}
%----%
so $b_i: \KT_{\gf}\rightarrow \IC$ and $\gamma_i:\KT_{\gf} \rightarrow \IC^\times$ furnish a basis of linear coordinates. The coordinate ring is then simply
%----%
\begin{equation}
\label{eq:KT_coordinate_ring}
	\IC[\KT_\gf]\cong\IC[b_i,\gamma^{\pm1}_{i}\,|\,i=1,\dots,\rk~\gf]~.
\end{equation}
%----%
The Poisson structure descends from the Kirillov--Kostant--Souriau bracket on $\IC[\gf^*]\cong\mathrm{Sym}\,\gf$ (see, \eg, \cite{Kostant:1979195}), and we have
%----%
\begin{equation}
\label{eq:KT_Poisson}
	\{b_i,\gamma_j^{\pm}\} = \delta_{ij} \gamma_{i}^{\pm}~,
\end{equation}
%----%
with all other brackets vanishing.\!\footnote{Note that the $b_i$, $\gamma_i$ are \textit{linear} duals of $h_i^*$ and $e_{\alpha_i}^*$, respectively. That is $b_i(h^\ast_j)=\delta_{ij}$ and $\gamma_i(e^\ast_{\alpha_j}) = \delta_{ij}$, so in particular $b_i\neq h_i$ under the canonical identification $\gf=(\gf^\ast)^\ast$.} Thus we identify $\KT_{\gf}\cong T^*(\IC^\times)^{\rk~\gf}$ as symplectic varieties. The structure of a complex integrable system on $\KT_\gf$ is then manifested through the action variables $P_i|_{\KT_\gf}$ for $i=1,2,\dots,\rk~\gf$, which are the restrictions of the generators of the Harish-Chandra centre, $\IC[\gf^*]^G$ to $\KT_{\gf}$.

In the construction of~\cite{Ginzburg:2018}---and its chiralisation in~\cite{Arakawa:2018egx}---the structure of Higgs branches of class $\SS$ as $S_f$-schemes plays an integral role. Here, the structure morphism of $\KT_\gf$ as an $S_f$-scheme is the composition
%----%
\begin{equation}
\label{eq:KT_S_scheme}
\begin{split}
	\KT_\gf&\xrightarrow{P} \gf^*/\!\!/G\xrightarrow{\sim} S_f~,\\
	y &\mapsto (P_1(y),P_2(y),\dots P_{\rk\,\gf}(y))~,
\end{split}
\end{equation}
%----%
where the second morphism is Kostant's isomorphism $\gf^*/\!\!/G\cong S_f$. This structure morphism is thus exactly the map to the base of the aforementioned complex integrable system, up to the identification of the base with $S_f$.

We have deliberately labelled the Kostant--Toda lattice with the Lie algebra and not the group, whereas the universal centraliser depends on the choice of group $G$. In~\cite{Kostant:1979195,Crooks:2020}, the embedding from this lattice targets the universal centraliser of the \textit{adjoint} group. However, for the purposes of class $\SS$ constructions we expect to be mainly interested in the simply connected group $G$ with $\mathrm{Lie}~G = \gf$. To generalise the embedding appropriately, we then need to replace $\KT_\gf$ with a suitable cover. 

The root lattice of the Lie \textit{group} $G$ always contains the root lattice of the Lie algebra $\gf$ as a sublattice. Let $\Delta_G$ denote the simple roots of the Lie algebra $\gf$ embedded into the root lattice of the Lie group $G$. Such a simple root, $\alpha\in\Delta_G$, is then a character $\alpha:T\rightarrow \IC^\times$, defined by its action on the simple root generator $e_\alpha\in\gf$ as
%----%
\begin{equation}
\label{eq:group_roots}
	\mathrm{Ad}_t e_\alpha = \alpha(t)\, e_\alpha~,\qquad \forall t\in T~.
\end{equation}
%----%
Let $Z_G: T\rightarrow \big(\IC^\times\big)^{\rk\,\gf}$ be defined by
%----%
\begin{equation}
\label{eq:group_center_acts_on_torus}
\begin{split}
	t\mapsto (\alpha(t))_{\alpha\in\Delta_G}~,\qquad \forall t\in T~,\\
\end{split}
\end{equation}
%----%
This map is an isomorphism for adjoint $G$, but if $G$ has a nontrivial centre, $Z(G)$, then it is only surjective. For any $G$, $Z_G$ is \'etale and the fibre above a point has cardinality $|Z(G)|$.

Now consider the map $\KT_\gf\twoheadrightarrow\big(\IC^\times\big)^{\rk\,\gf}$ that projects to the base of the cotangent bundle in $\KT_\gf\cong T^*(\IC^\times)^{\rk~\gf}$. We define $\KT_G$ to be the pullback of this projection and the map $T$ given above,
%----%
\begin{equation}
\label{eq:pullback_square}
	\begin{tikzcd}
	\KT_G \ar[d]\ar[r]& \KT_\gf \ar[d,twoheadrightarrow] \\
	T \ar[r,"Z_G"] & \big(\IC^{\times}\big)^{\rk\,\gf}
	\end{tikzcd}~.
\end{equation}
%----%
The structure sheaf of $\KT_G$ has the structure of an $\OO_{(\IC^\times)^{\rk\,\gf}}$-module given by
%----%
\begin{equation}
	\OO_{\KT_G} = \OO_{\KT_\gf}\otimes_{\OO_{(\IC^\times)^{\rk\,\gf}}} \OO_{T}~,
\end{equation}
%----%
with the $\OO_{(\IC^\times)^{\rk\,\gf}}$-module structure on $\OO_{T}$ coming from restriction of scalars along $Z_G^\#:\OO_{(\IC^\times)^{\rk\,\gf}}\rightarrow \OO_T$. Notice that we have the identification
%----%
\begin{equation}
\label{eq:KTG_is_cotangent}
\begin{split}
	\KT_G &= \{(y,t)\,|\, y\in \KT_\gf~,~t\in T~,~(\gamma_1(y),\gamma_2(y),\dots,\gamma_{\rk\,\gf}(y))= Z_G(t)\}~,\\
	&\cong \{(b_1(y),b_2(y),\dots,b_{\rk\,\gf}(y),t)\in (\IC\times\IC^\times)^{\rk\,\gf}\,\}~,
\end{split}
\end{equation}
%----%
and so $\KT_G\cong \KT_\gf$ as varieties. (In fact, by a linear redefinition of the coordinate functions $b$, this isomorphism can be realised as one of symplectic varieties as well.)

However, the canonical projection $\pi_G:\KT_G\rightarrow \KT_\gf$, coming from the pullback \eqref{eq:pullback_square} is not an isomorphism. Indeed, by standard properties of the pullback $\pi_G$ is \'etale, since $Z_G$ is \'etale. We should therefore treat $\KT_G$ as a covering space of $\KT_\gf$ by way of $\pi_G$.

From its definition \eqref{eq:pullback_square}, we also have a canonical projection
%----%
\begin{equation}\label{eq:piT}
	\pi_T:\KT_G\rightarrow T~.
\end{equation}
%----%
Since $\pi_T$ is a base change of the surjective morphism $\KT_\gf\twoheadrightarrow (\IC^\times)^{\rk\,\gf}$ (and self-evidently from \eqref{eq:KTG_is_cotangent}), it too is surjective.

By construction, the fibre of $\pi_G$ above a point in $\KT_\gf$ will contain $|Z(G)|$ many points. These fibres admit a $Z(G)$ action which makes $\pi_G:\KT_G\rightarrow \KT_\gf$ into a $Z(G)$-torsor. In particular, let $y\in \KT_\gf$. The points in the fibre $\pi_G^{-1}(y)$ are mapped bijectively by $\pi_T$ to a $Z(G)$ orbit in $T$. Let $\tilde{y}\in \pi_G^{-1}(y)$ and let $z\in Z(G)$. The $Z(G)$ action is defined uniquely by specifying
%----%
\begin{equation}\label{eq:ZGactionKTG}
	z \cdot \tilde{y} \coloneqq	 \tilde{y}'~,\quad\text{such that}\quad \pi_T(\tilde{y}') = z\cdot \pi_T(\tilde{y})~,\quad \pi_G(\tilde{y}) = \pi_G(\tilde{y}')~.
\end{equation}
%----%
By construction, this gives the following:
%----%
\begin{lem}\label{lem:piT_equivariant}
The map $\pi_T:\KT_G\rightarrow T$ is $Z(G)$-equivariant.
\end{lem}
%----%
In light of this, the projection $\pi_G:\KT_G\rightarrow \KT_\gf$ is just the natural projection to $Z(G)$ orbits, $\KT_G\rightarrow \KT_G/\!\!/Z(G)$, where $\KT_\gf$ is identified with the quotient space $\KT_G/\!\!/Z(G)$. Letting $(y,t)\in\KT_G$, with $y\in \KT_\gf$ and $t\in T$
%----%
\begin{equation}
	\pi_G(y,t) = \chi + \sum^{\rk\,\gf}_{i}b_i(y)h^*_i + \sum_{\alpha\in\Delta} \alpha(t)e_{\alpha}^*~.
\end{equation}
%----%
Since the map $\pi_G$ is \'etale, we can pullback the symplectic form on $\KT_\gf$ to a symplectic form on $\KT_G$, endowing its coordinate ring in particular with a Poisson bracket. Furthermore, $\KT_G$ can be made into an $S_f$-scheme with structure morphism
%----%
\begin{equation}
	\KT_G\xrightarrow{\pi_G}\KT_\gf\xrightarrow{P} \gf^*/\!\!/G\xrightarrow{\sim}S_f~.
\end{equation}
%----%

We are primarily interested in the case when $G$ is simply-connected; in this case the functions on $\KT_G$ have a more concrete presentation. Let $\big(\alpha_i\big)_{i=1}^{\rk\,\gf}$ be the simple roots of $G$ and $a_{ij}$ be the matrix elements of the Cartan matrix of $\gf$. Let $\widetilde{a}_{ij}$ be the matrix elements of the inverse Cartan matrix. Then $\OO_{\KT_\gf}(\KT_\gf)\cong \IC[b_{\alpha_i},\gamma_{\alpha_i}\,|\, i=1,\dots,\rk\,\gf]$ and
%----%
\begin{equation}
	\OO_{\KT_G}(\KT_G) \cong \IC[b_{\alpha_i}, \widetilde{\gamma}_{\alpha_i} | i=1,\dots,\rk\,\gf]~,
\end{equation}
%----%
where
%----%
\begin{equation}
	\widetilde{\gamma}_{\alpha_i} = \prod_{j=1}^{\rk\,\gf} \gamma^{\widetilde{a}_{ij}}_{\alpha_j}~.
\end{equation}
%----%
The map $\pi_G$, at the level of rings is,
%----%
\begin{equation}
	b_{\alpha_i} \mapsto b_{\alpha_i}~, \quad \gamma_{\alpha_i} \mapsto \prod_{i=1}^{\rk\,\gf} \widetilde{\gamma}_{j}^{a_{ij}}~,\qquad i=1,\dots,\rk~\gf~.
\end{equation}
%----%
The Poisson structure inherited by $\KT_G$ is just the natural extension of the Poisson bracket on $\KT_\gf$ to the $\widetilde{\gamma}_{\alpha_i}$ by imposing the chain rule.

When $G$ is simply connected, we can also give a more concrete presentation of $\pi_T$. Let $(t_i)_{i=1}^{\rk\,\gf}$ be the generators of the root lattice of $G$. Since the root lattice of $G$ is isomorphic to the weight lattice of the Lie algebra, $\gf$, these generators can be identified with the fundamental weights, $(\omega_i)^{\rk\,\gf}_{i=1}$, of $\gf$. Let $V_{i}$ be the finite dimensional representation of $\gf$ with highest weight $\omega_i$, and let $v_i$ be the highest weight vector. Then, if $h\in T$
%----%
\begin{equation}
    h v_i = t_i(h) v_i~.
\end{equation}
%----%

The $t_i$ are good coordinates on the torus and the map $\pi_T$ is induced by the ring homomorphism,
%----%
\begin{equation}
\label{eq:piT-concrete}
\begin{split}
	\pi_T^\#:\IC[t_1^{\pm 1},\dots,t^{\pm1}_{\rk\,\gf}]&\rightarrow \IC[b_{\alpha_i},\widetilde{\gamma}^{\pm 1}_{\alpha_j}\,|\,i,j=1,\dots\rk\,\gf]\\
	t_i &\mapsto \widetilde{\gamma}_{\alpha_i}=\prod_{j=1}^{\rk~\gf} \gamma_{\alpha_j}^{\widetilde{a}_{ij}}
\end{split}
\end{equation}
%----%
Suppose $G$ is neither simply-connected nor of adjoint type. Then $G$ is a quotient of the simply-connected group, $G_{sc}$, by some central subgroup $Z'\subset Z(G_{sc})$ with $Z'\neq Z(G_{sc})$. Coordinates on $G$ are then realised by passing to $Z'$ invariants of the coordinate ring of the simply-connected group.

%----------------------------------------------------%
\subsection{\label{subsec:sl2_KT_example}The embedding in rank one}
%----------------------------------------------------%

To make this abstract construction a bit more concrete, we consider in detail the two rank-one examples: $G={\rm SL}_2$ and $G={\rm PGL}_2$. We will then see directly how these Kostant--Toda lattices embed into the corresponding universal centralisers.

The $\slf_2$ Kostant--Toda lattice, thought of as a subset of $\slf_2$ using the Killing form, consists of all elements of the form
%----%
\begin{equation}
	y =
	\begin{pmatrix}
	0 & 0 \\
	1 & 0
	\end{pmatrix}
	+
	\begin{pmatrix}
	b(y) & 0 \\
	0 & -b(y)
	\end{pmatrix}
	+
	\begin{pmatrix}
	0 & \gamma(y) \\
	0 & 0
	\end{pmatrix}
	~,
\end{equation}
%----%
with $b(y) \in \IC$ and $\gamma(y)\in \IC^\times$. The coordinate ring is $\IC[b,\gamma^{\pm}]$, with non-vanishing Poisson brackets
%----%
\begin{equation}
	\{b,\gamma^{\pm}\} = \pm\gamma^{\pm}~.
\end{equation}
%----%

For $G={\rm PGL}_2$ the maximal torus is isomorphic to $\IC^\times$ and the simple root, $\alpha:\IC^\times\rightarrow \IC^\times$, is just the identity morphism. Following the construction above, this means that the structure sheaf $\OO_{\KT_{{\rm PGL}_2}}$ is isomorphic to $\OO_{\KT_{\slf_2}}$ as $\OO_{\IC^\times}$-modules. The global sections are just $\IC[b,\gamma^{\pm 1}]=\OO_{\KT_{\slf_2}}(\KT_{\slf_2})$ and the map $\pi_{PGL_2}:\KT_{\rm PGL_2}\rightarrow \KT_{\slf_2}$ is an isomorphism. Here we see directly that with $\rm PGL_2$ being of adjoint type, the covering space $\KT_{\rm PGL_2}$ is exactly identified with $\KT_{\slf_2}$. The structure morphism to $S_f$ is given by $b^2+\gamma$.

Alternatively, for $G={\rm SL}_2$ the simple root $\alpha$ is given by $\alpha:\IC^\times\xrightarrow{z^2} \IC^\times$.  The structure sheaf of $\KT_G$ is $\OO_{\KT_G} =\OO_{T}\otimes_{\OO_{\IC^\times}} \OO_{\KT_\gf}$, with global sections
%----%
\begin{equation}
	\OO_{\KT_G}(\KT_G) \cong\frac{ \IC[h^{\pm 1}]\otimes\IC[b,\gamma^{\pm 1}]}{\langle h^2 - \gamma\rangle}\cong \IC[b,\gamma^{\pm\frac{1}{2}}].
\end{equation}
%----%
which has the obvious module structure over $\IC[b,\gamma^{\pm 1}]$. The map $\pi_{{\rm SL}_2}$ on rings is just the natural inclusion $\pi_{{\rm SL}_2}^\#: \IC[b,\gamma^{\pm 1}]\hookrightarrow \IC[b,\gamma^{\pm\frac{1}{2}}]$. The Poisson structure, inherited from $\KT_\gf$ is
%----%
\begin{equation}
	\{b,\gamma^{\pm\frac{1}{2}}\} = \pm\frac{1}{2}\gamma^{\pm\frac{1}{2}}~.
\end{equation}
%----%
Clearly, with a rescaling of $b$ we have $\KT_{{\rm SL}_2}\cong T^*\IC^\times\cong \KT_{{\rm PGL}_2}$ including Poisson structures. However, unlike the adjoint case, here our construction gives $\KT_{{\rm SL}_2}$ as a $2$-cover of $\KT_{\slf_2}$ along $\pi_{{\rm SL}_2}$. The structure morphism for $\KT_{{\rm SL}_2}$ is $b^2+ \gamma^{\frac{1}{2}}\gamma^{\frac{1}{2}}$.

Now we turn to the universal centralisers for these two groups. From Section~\ref{subsec:eg_rank_one} we know that $\Zf_{{\rm SL}_2}\subset\mathbb{A}^3$ can be realised as an affine variety with coordinate ring
%----%
\begin{equation}
	I_{{\rm SL}_2} = \frac{\IC[X,Y,S]}{\langle SX^2 - Y^2 - 1\rangle}~.
\end{equation}
%----%
Similarly, the adjoint version, $\Zf_{\rm PGL_2}$, has coordinate ring
%----%
\begin{equation}
\label{eq:pgl2_relations}
	I_{\rm PGL_2} \cong \frac{\IC[S,X^2,XY,Y^2]}{\langle SX^2 - Y^2 -1\rangle}\cong \frac{ \IC[S,A,B,C]}{\langle B^2 - AC,SA - C-1\rangle}~.
\end{equation}
%----%

Let $U_{X}\subset \Zf_{{\rm SL}_2}$ be the elementary open set defined by localising on the ideal $\langle X \rangle$. In this open subset, the relation $SX^2-Y^2-1$ can be solved for $S$ according to
%----%
\begin{equation}
	S = \frac{1}{X^2} + \frac{Y^2}{X^2}~.
\end{equation}
%----%
So $\OO_{\Zf_{{\rm SL}_2}}(U_X)\cong \IC[Y,X^{\pm 1}]$.

Constructing an open immersion $\KT_{{\rm S}_2}\rightarrow \Zf_{{\rm SL}_2}$ is straightforward in this case. The open subset $U_X$ can be identified with $\KT_G$ as Poisson schemes over $S_f$ via the map
%----%
\begin{equation}
\label{eq:sl2_finite_free_field_XY}
	X \mapsto \gamma^{-\frac{1}{2}}~,\quad Y\mapsto  -b\gamma^{-\frac{1}{2}}~.
\end{equation}
%----%
The Slodowy slice coordinate function is then realised as follows,
%----%
\begin{equation}
\label{eq:sl2_finite_free_field_S}
	S \mapsto b^2 + \gamma^{\frac{1}{2}} \gamma^{\frac{1}{2}} = b^2 + \gamma~.
\end{equation}
%----%
This, of course, agrees with the structure morphism of $\KT_{{\rm SL}_2}$, so this morphism is an open immersion of (Poisson) $S_f$-schemes.

Now let $U_A$ be the elementary open set obtained by localising $I_{\rm PGL_2}$ at the ideal $\langle A\rangle $. In this open set, the relations \eqref{eq:pgl2_relations} can be solved according to
%----%
\begin{equation}
	C = \frac{B^2}{A}~,\quad S = \frac{1}{A} + \frac{B^2}{A^2}~.
\end{equation}
%----%
Therefore, $U_A \cong \mathrm{Spec}~\IC[B,A^{\pm1}]\cong T^{*}\IC^{\times}$ as schemes. We are free to pick generators $\frac{B}{A}$ and $A^{\pm 1}$ for the functions on $U_A$. The Poisson bracket is then determined by
%----%
\begin{equation}
	\left\{\frac{B}{A},A^{\pm}\right\} = \pm A^{\pm}~,
\end{equation}
%----%
so we can identify $U_A$ with $T^*\IC^\times$ as Poisson schemes. Now to embed $\KT_{\rm PGL_2}$, we make the following identifications
%----%
\begin{equation}
	A \mapsto \gamma^{-1}~, \quad \frac{B}{A} \mapsto b~.
\end{equation}
%----%
This extends to
%----%
\begin{equation}
	C \mapsto  b^2 \gamma^{-1}~, \quad S \mapsto b^2 + \gamma
\end{equation}
%----%
These identifications induce an open immersion of $S_f$ schemes $\KT_{\rm PGL_2}\rightarrow \Zf_{\rm PGL_2}$.

Since the open subsets $U_X$ and $U_A$ are both isomorphic to $T^*\IC^{\times}$ as Poisson varieties, one might at first inquire as to whether there is any reason to introduce a covering space after all. Indeed, one may construct the open immersion $\KT_{\rm PGL_2} \rightarrow \Zf_{{\rm SL}_2}$, with image $U_X$, by making the identificationes
%----%
\begin{equation}
	X \mapsto \gamma^{-1}~,\quad Y \mapsto -\frac{1}{2}b\gamma^{-1}~,\quad S \mapsto \frac{1}{4}b^2 + \gamma^2~.
\end{equation}
%----%
This map is an open immersion of Poisson schemes, however it fails to be a morphism of $S_f$-schemes.
Namely, the image of $S$ is not $b^2+\gamma$. Insisting that the open immersion must be a morphism of $S_f$-schemes is natural---it preserves some Lie theoretic data, namely the $\gamma$ can be thought of as linear function on $\gf^*$. Furthermore, this ensures that the inclusion $I_{{\rm SL}_2}\hookrightarrow \IC[\KT_{{\rm SL}_2}]$ is a map of $\IC[\gf^*]^G$-modules; this structure appears to be of general importance in the construction of Higgs branches for theories of class $\SS$.

%----------------------------------------------------%
\subsection{\label{subsec:morphism_tech}The embedding in the general case and some technicalities}
%----------------------------------------------------%

For higher rank cases, the relationship between $\KT_G$ and $\Zf_G$ is more intricate. In this section we provide an abstract description of an open immersion $\KT_G\xrightarrow{\varphi}\Zf_G$. For $G$ of adjoint type, this map was described in less concrete terms in \cite{Kostant:1979195}, and our analysis follows and extends that of \cite{Crooks:2020}.

We decompose $\varphi$ into morphisms $\beta:\KT_G\rightarrow S_f$ to the base and $\lambda:\KT_G\rightarrow G$ to the fibre,~\ie, for $y\in\KT_G$, $\lambda(y)\in G_{\beta(y)}$. Let $\mathfrak{b}^*$ be the image of the Borel subalgebra $\mathfrak{b}=\mathfrak{t}+\mathfrak{n}_+$ under the Killing isomorphism. Kostant~\cite{Kostant:1978} tells us that there is an isomorphism of affine varieties $\psi:N\times S_f\xrightarrow{\sim}\chi+\mathfrak{b}^*$ realised as
%----%
\begin{equation}
	(n,s) \mapsto \mathrm{Ad}^*_n s~.
\end{equation}
%----%
Additionally, we have canonical projections $\pi_{S_f}: N\times S_f\rightarrow S_f$ and $\pi_{N}: N\times S_f \rightarrow N$. The Kostant--Toda lattice is a subset of $\chi+\mathfrak{b}^*$, so we can define a map
%----%
\begin{equation}
\begin{split}
	\beta:\mathrm{KT}_\gf &\rightarrow S_f~,\\
	y &\mapsto \pi_{S_f}(\psi^{-1}(y))~.
\end{split}
\end{equation}
%----%
Recall the generators $P_i$ of $\IC[\gf^*]^G$ whose restriction to $\KT_{\gf}$ define the map $P:\KT_\gf\rightarrow \gf^*/\!\!/G$. By~\cite[Theorem 2.6]{Kostant:1979195} this is surjective. The following gives an important compatibility result between these various maps.
%----%
\begin{prop}\label{prop:beta_dominant} (\cite[Proposition 23]{Crooks:2020})
Consider the triangle
%----%
\begin{equation}
	\begin{tikzcd}
	\KT_{\gf} \arrow[rr,"\beta"]\arrow[dr,twoheadrightarrow, "P",swap] & & S_f\arrow[dl,"\sim"] \\
	&  \mathfrak{g}^*/\!\!/G &
	\end{tikzcd}~,
\end{equation}
%----%
where the diagonal right arrow is the Kostant isomorphism. This diagram commutes, so in particular $\beta$ is surjective.
\end{prop}

The structure morphism for $\KT_\gf$ as an $S_f$-scheme as introduced previously is the composition of the lower route around this diagram, so can be identified with the map $\beta$.

The morphism $\beta:\KT_\gf\rightarrow S_f$ can be lifted to a morphism $\KT_G\rightarrow S_f$ via the composition $\KT_G\xrightarrow{\pi_G}\KT_\gf\xrightarrow{\beta}S_f$, which is also surjective. In an abuse of notation, we denote this composition as $\beta:\KT_G\rightarrow S_f$. This $\beta$ again agrees with the structure morphism of $\KT_G$ as an $S_f$-scheme.

For the time being, let us assume $G$ is of adjoint type (subsequently we will describe the generalisation of this story to non-adjoint $G$, and the more general $\KT_G$ will make an appearance). Then the map $\beta$ gives us a foliation of $\KT_\gf$ over the Slodowy slice, and the fibre above a point $s\in S_f$ can be shown to be isomorphic to a Zariski open subset of the $G$-stabiliser of $s$ according to \cite[Theorem 2.6]{Kostant:1979195}. These isomorphisms of fibres, taken collectively over $S_f$, leads to an isomorphism of $S_f$-schemes between $\KT_\gf$ and a Zariski open subset of $\Zf_G$ that can be made completely explicit. We now review and elaborate upon this construction following \cite{Crooks:2020} with an eye towards the generalisation to general $G$.

As we said before, the phase space $\KT_\gf$ has a foliation
%----%
\begin{equation}
	\KT_\gf = \bigsqcup_{p\in\gf^*/\!\!/G} P^{-1}(p)~.
\end{equation}
%----%
Each leaf of this foliation can be related to the centraliser of a particular element in $\chi+\mathfrak{t}^*$ as follows. Consider the diagram,
%----%
\begin{equation}
	\begin{tikzcd}
	\KT_\gf \ar[r,twoheadrightarrow,"P"] & \gf^*/\!\!/G  & \chi + \mathfrak{t}^* \ar[l,twoheadrightarrow,"\sigma_{t}",swap]
	\end{tikzcd}
	~,
\end{equation}
%----%
where $\sigma_t$ is the composition $\chi+\mathfrak{t}^* \twoheadrightarrow \mathfrak{t}^*/\!\!/W\xrightarrow{\sim} \gf^*/\!\!/G$. Let $y\in \KT_{\gf}$ and $p=P(y)$, then the fibre $\sigma_t^{-1}(p)$ is nonempty. Let $\theta_p$ be the point in this fibre such that $\theta_p-\chi$ lies in the fundamental Weyl chamber.\!\footnote{%
%----%
Note that $y\mapsto\theta_p$ does not define an algebraic morphism between $\KT_\gf$ and $\chi+\mathfrak{t}^*$, since the fundamental domain is not algebraic. Nevertheless, the construction will telescope and $\theta_p$ itself will not appear in the final definition of the open immersion.}
%----%
Then $\theta_p$ and $y$ are $N$-conjugate, since $P(y)=P(\theta_p)$ and both elements are in $\chi+\mathfrak{b}^*$. Furthermore, as $\beta(y)$ and $y$ are also $N$-conjugate, $\theta_p$ must be $N$-conjugate to $\beta(y)$. Let $\delta_p\in N$ be the unique group element such that $\mathrm{Ad}_{\delta (p)}^* \theta_p = \beta(y)$. This element is unique since $\theta_p\in\chi+\mathfrak{b}^*\cong N\times S_f$, and can be written as follows,
%----%
\begin{equation}
\label{eq:N_group_element}
	\delta_p = \pi_N(\psi^{-1}(\theta_p))^{-1}~.
\end{equation}
%----%
Now let $w_0$ be the longest element of the Weyl group,~\ie, the \emph{Cartan element} of $W(G)$. There is a unique uplift of $w_0$ to a $\widetilde{w}_0\in N_G(T)$, the normaliser of $T$, such that $Ad_{\widetilde{w}_0}e_{\alpha} = e_{w_0 \alpha}$ for any simple root $\alpha$ and any associated generator $e_\alpha \in \nf_+$. Consider the Zariski open subset $G^*\coloneqq N \widetilde{w}_0 TN$, which is just the big Bruhat cell in $G$. For $x\in\gf^\ast_{\rm reg}$, let $G^*_x$ be the Zariski open subset of $G_x$ defined by the intersection $G^*_x\coloneqq G^*\cap G_x$.

%----%
\begin{thm}\label{thm:Kostant_level_set} \cite[Theorem 2.6]{Kostant:1979195}
Let $p\in \gf^*/\!\!/G$ and let $\theta_p$ be as above. The morphism
%----%
\begin{equation}\nonumber
\begin{split}
	K_p: G^*_{\theta_p}\cong P^{-1}(p)\\
	n^\prime\widetilde{w}_0h n\mapsto \mathrm{Ad}^*_{n}\theta_p~,
\end{split}
\end{equation}
%----%
where $n^\prime,n\in N$ and $h\in T$, is an isomorphism of algebraic varieties.
\end{thm}
%----%
This theorem relates the fibres of $\KT_\gf$ over $\gf^*/\!\!/G$ to a Zariski open subset of a stabiliser subgroup of $G$; we next want to relate this stabiliser subgroup to a stabiliser subgroup of the corresponding point $\beta(y)$ in the Slodowy slice.

For a given $p\in \gf^*/\!\!/G$, Kostant's isomorphism gives a unique $s_p\in S_f$, and once again the previously defined $\delta_p$ is precisely such that $\mathrm{Ad}^*_{\delta_p} \theta_p = s_p$. Thus the stabilisers of $s_p$ and $p$ are related by
%----%
\begin{equation}
	G^*_{s_p} = \delta_pG^*_{\theta_p}\delta_p^{-1}~.
\end{equation}
%----%

We can therefore define a Zariski open set, $\Zf_G^*\coloneqq (G^*\times S_f)\cap \Zf_G$, which has foliation
%----%
\begin{equation}
	\Zf_G^* = \bigsqcup_{s\in S_f} G^*_s~,
\end{equation}
%----%
with each fibre isomorphic to $G^*_{\theta_s}$ under conjugation by $\delta_s$. By Theorem~\ref{thm:Kostant_level_set} we can therefore construct an open immersion
%----%
\begin{equation}
	\varphi:\KT_\gf \rightarrow \Zf_G~,
\end{equation}
%----%
by sending each leaf of the foliation $\KT_\gf \cong \bigsqcup_{p\in\gf^*/\!\!/G}P^{-1}(p)$ to the stabiliser $G^*_{\theta_p}$ and then conjugating by $\delta_p$. Each step is an isomorphism, so $\varphi$ is an open immersion with image $\Zf_G^*$.

Defining the morphism $\varphi$ requires an explicit realisation of the inverse $K_p^{-1}: P^{-1}(p) \rightarrow G_{\theta_p}^*$. To construct the inverse to $K_p$, we define a few more morphisms. Via the Kostant isomorphism, we have a map $\rho_N:\KT_\gf\rightarrow N$ given by the composition
%----%
\begin{equation}
	\KT_\gf\xrightarrow{\psi^{-1}}N\times S_f \xrightarrow{\pi_N} N~.
\end{equation}
%----%
Let $\widetilde{\rho}_N$ be the further composition
%----%
\begin{equation}
	\KT_\gf\xrightarrow{\mathrm{Ad}^*_{\widetilde{w}_0\pi_T(-)^{-1}}}\KT_\gf\xrightarrow{\psi^{-1}}N\times S_f\xrightarrow{\pi_N} N~,
\end{equation}
%----%
In other words,
%----%
\begin{equation}
\begin{split}
	\widetilde{\rho}_N:\KT_G&\rightarrow N \\
	y &\mapsto \pi_N(\psi^{-1}(\mathrm{Ad}_{\widetilde{w}_0\pi_T(y)^{-1}}^*y)) ~.
\end{split}
\end{equation}
%----%
Here, $\pi_T: \KT_\gf\rightarrow T$ arises from the identification $\pi_G: \KT_{G}\xrightarrow{\sim}\KT_\gf$, for adjoint $G$. The inverse $K_p^{-1}:P^{-1}(p)\rightarrow G_{\theta_p}^*$ is given by
%----%
\begin{equation}
\begin{split}
	K_p^{-1}(y) &= \bigg((\widetilde{\rho}_N(y)\delta_p)^{-1}\widetilde{w}_0(\pi_T(y))^{-1} \rho_N(y) \delta_p\bigg)~,\\
	&= \delta_p^{-1} \widetilde{\rho}_N(y)^{-1} \widetilde{w}_0(\pi_T(y))^{-1} \rho_N \delta_p~,
\end{split}
\end{equation}
%----%
for any $y\in P^{-1}(p)$. As this expression is a bit opaque, let us try and explain informally why it stabilises $\theta_p$. A more detailed analysis of its origin can be found in the proof of~\cite[Lemma 19]{Crooks:2020}.

Given $y\in P^{-1}(p)$ we want an element of $G^*$ that stabilises $\theta_p$. Every such element has the form $n^\prime\widetilde{w}_0hn$ for some $n,n^\prime\in N$ and $h\in T$. In terms of the ingredients appearing in $K_p^{-1}$, we have
%----%
\begin{equation}
\begin{split}
	n^\prime &= \big(\widetilde{\rho}_N(y)\delta_p\big)^{-1}~,\\
	h &= \big(\pi_T(y)\big)^{-1}~,\\
	n &= \rho_N(y) \delta_p~.
\end{split}
\end{equation}
%----%
The composite expression for $K_p^{-1}(y)$ acts on $\theta_p$ in the following way. First $\delta_p$ takes $\theta_p$ to $s_p$, then $\rho_N(y)$ takes $s_p$ to $y$---this follows from the definitions of $\rho_N(y)$ and $\delta_p$. Note that this is precisely the image of the map $K_p$.

Next the element $y$ gets mapped to another $y'\in P^{-1}(p)$ by $\widetilde{w}_0h$, where the form of $h$ is dictated by the requirement that this indeed lands in $\KT_\gf$. Finally, $\widetilde{\rho}_N(y)^{-1}$ is the unique element of $N$ that takes this $y'$ back to $s_p$, and $\delta_p^{-1}$ will conjugate $s_p$ back to $\theta_p$---stabilising $\theta_p$ as a result.

Now, for any $y\in \KT_\gf$, the morphism $\varphi$ is given by
%----%
\begin{equation}
\label{eq:adjoint_map_def}
\begin{split}
	\varphi(y) &= \big(\delta_p K_p^{-1}(y) \delta_p^{-1}, \beta(y)\big)~,\\
	&= \big(\widetilde{\rho}_N(y) \widetilde{w}_0 (\pi_T(y))^{-1} \rho_N(y), \beta(y)\big)~,
	\end{split}
\end{equation}
%----%
where $p = P(y)$. Note that the definition has telescoped, so that the map, $\lambda:\KT_\gf\rightarrow G^*$, which identifies leaves of $\KT_\gf$ and $\Zf_G$, does not depend explicitly on $p$, $\theta_p$ or $\delta_p$. Indeed, the map \eqref{eq:adjoint_map_def} is manifestly algebraic.

Now let us return to the case where $G$ is not necessarily of adjoint type. Then $p\in \gf^*/\!\!/G$ and $\theta_p$ can be defined as before, while the longest element of the Weyl group, $w_0\in W$, no longer has a unique uplift to $N_G(T)$---there is now a $Z(G)$-orbit worth of uplifts. We fix $\widetilde{w}_0\in N_G(T)$ to be the unique uplift of $w_0$ that lies in the connected component of the identity. The Zariski open $G^*\subset G$ is defined, as before, as the big Bruhat cell $N\widetilde{w}_0TN$ for this choice of $\widetilde{w}_0$, and if $x\in\gf^\ast_{\rm reg}$, $G^*_x \coloneqq G^*\cap G_x$.

The map $K_p:G^*_{\theta_p}\rightarrow P^{-1}(p)$ is no longer an isomorphism; the fibres of the map have cardinality $|Z(G)|$. Therefore, it cannot be inverted and the previous definition of $\lambda$ is not fit for purpose. To rectify the wituation we will pass to the cover $\KT_G$. We define the morphism $\overline{P}:\KT_G\rightarrow \gf^*/\!\!/G$ as the composite
%----%
\begin{equation}
	\KT_G\xrightarrow{\pi_G}\KT_\gf\xrightarrow{P}\gf^*/\!\!/G~,
\end{equation}
%----%
and for $p\in\gf^*/\!\!/G$, $\overline{P}^{-1}(p)$ is a cover of $P^{-1}(p)$ and we will want to define a lift $\overline{K}_p:G^*_{\theta_p}\rightarrow \overline{P}^{-1}(p)$ that is an isomorphism. For the purposes of constructing $\varphi$, we are actually interested in the putative inverse $\overline{K}_p^{-1}: \overline{P}^{-1}(p) \rightarrow G^*_{\theta_p}$, so in what follows we will construct $\overline{K}_p^{-1}$ and show that it is an isomorphism.

The fibres $\overline{P}^{-1}(p)$ have $|Z(G)|$ many disconnected components, all of which are isomorphic to $P^{-1}(p)$ as varieties. The strategy will be to identify each component of $\overline{P}^{-1}(p)$ with a component of $G^*_{\theta_p}$ in a $Z(G)$-equivariant way.

To define our inverse morphism, we need to lift $\rho_N$ and $\widetilde{\rho}_N$ to the cover, which can be done by precomposition. For convenience, we continue to use $\rho_N$ to denote the composition $\KT_G\xrightarrow{\pi_G}\KT_\gf\xrightarrow{\rho_N}N$, and $\widetilde{\rho}_N$ for the composition
%----%
\begin{equation}
	\KT_G\xrightarrow{\pi_G}\KT_\gf\xrightarrow{\mathrm{Ad}^*_{\widetilde{w}_0\pi_T(-)^{-1}}}\KT_\gf\xrightarrow{\psi^{-1}}N\times S_f\xrightarrow{\pi_N} N~.
\end{equation}
%----%
In other words,
%----%
\begin{equation}
\begin{split}
	\widetilde{\rho}_N:\KT_G&\rightarrow N \\
	y &\mapsto \pi_N(\psi^{-1}(\mathrm{Ad}_{\widetilde{w}_0\pi_T(y)^{-1}}^*\pi_G(y))) ~.
\end{split}
\end{equation}
%----%
Note that now the morphisms $\rho_N$ and $\widetilde{\rho}_N$ are both invariant on the $Z(G)$ orbits of $\KT_G$. For $\rho_N$ this is clear from the definition. For $\widetilde{\rho_N}$, note that the twisting by $\mathrm{Ad}^*_{\widetilde{w}_0(\pi_T(y))^{-1}}$ that enters the definition is insensitive to the centre, since the coadjoint action is blind to it.

We now define $\overline{K}_p^{-1}:\overline{P}^{-1}(p)\rightarrow G_{\theta_p}^*$, by the same expression as before,
%----%
\begin{equation}
\begin{split}
	\overline{K}_p^{-1}(y) &= \bigg((\widetilde{\rho}_N(y)\delta_p)^{-1}\widetilde{w}_0(\pi_T(y))^{-1} \rho_N(y) \delta_p\bigg)~,\\
	&= \delta_p^{-1} \widetilde{\rho}_N(y)^{-1} \widetilde{w}_0(\pi_T(y))^{-1} \rho_N \delta_p~,
\end{split}
\end{equation}
%----%
where now the definition applies for any $y\in \overline{P}^{-1}(p)$. For adjoint $G$, of course, this recovers the inverse $\overline{K}^{-1}_p$. Notice that the map $\pi_T$ \emph{is sensitive} to which connected component of $\KT_G$ $y$ lies in, and moreover, it is the only part of the map which is sensitive to this. Recalling the equivariance properties of Lemma~\ref{lem:piT_equivariant} we can anticipate that this map will target every connected component of $G^*_{\theta_p}$. We have the following.

\begin{thm}~\label{thm:Kx_surjective}
Let $p\in\gf^*/\!\!/G$, then $\overline{K}_p^{-1}:\overline{P}^{-1}(p)\rightarrow G^*_{\theta_p}$ is an isomorphism.
\end{thm}
%----%
\begin{proof}
The following is based on the proof of \cite[Lemma 19]{Crooks:2020}, modified to account for the subtleties of $G$ not being of adjoint type.

First, we establish surjectvity. Let $n^\prime\widetilde{w}_0hn\in G_{\theta_p}^*$ with $n^\prime,n\in N$ and $h\in T$. Then $y=Ad^*_n\theta_p$ is some element in $ P^{-1}(p)$ by Theorem \ref{thm:Kostant_level_set}. Since $n^\prime\widetilde{w}_0hn$ stabilises $\theta_p$, we have
%----%
\begin{equation}
	Ad^*_{n^\prime\widetilde{w}_0hn} \theta_p = Ad^*_{n^\prime\widetilde{w}_0h} y\overset{!}{=} \theta_p~.
\end{equation}
%----%
Concretely $\theta_p = \chi + \tau^*$ for some $\tau^*$ in the fundamental Weyl chamber and $y = \chi  + \tau' + \sum_{\alpha\in \Delta} \gamma_\alpha e_\alpha^*$. Therefore
%----%
\begin{equation}
\begin{split}
	\mathrm{Ad}^*_{n^\prime\widetilde{w}_0h} y &= \mathrm{Ad}^*_{n^\prime\widetilde{w}_0h}\big( \chi+\tau'+ \sum_{\alpha\in\Delta}\gamma_\alpha(y) e^*_\alpha \big)~,\\
	&= c_+ + \sum_{\alpha\in\Delta} \alpha(h) \gamma_\alpha(y) e_{w_0\alpha}^*~,
\end{split}
\end{equation}
%----%
with $c_+\in\mathfrak{b}^*$. Since this must be equal to $\theta_p = \chi + \tau$, we must have that
%----%
\begin{equation}
	\chi = \sum_{\alpha\in\Delta}\alpha(h)\gamma_\alpha(y) e^*_{w_0 \alpha}~,
\end{equation}
%----%
which implies
%----%
\begin{equation}
	\alpha(h) = \gamma_\alpha^{-1}(y)~.
\end{equation}
%----%
This means that $h$ can be written as $\big(\pi_T(\tilde{y})\big)^{-1}$ for some $\tilde{y}\in \pi^{-1}_G(y)$. Moreover, this choice of $\tilde{y}$ must be unique, since $\pi_T$ is $Z(G)$ equivariant from Lemma~\ref{lem:piT_equivariant}. Since $\theta_p$ and $y$ are in the same $N$-orbit of a point on the Slodowy slice, we must have that $n = \rho_N(\tilde{y})\delta_p$ for every $\tilde{y}\in\pi_G^{-1}(y)$.

Next, to fix $n^\prime$ we can follow, with slight modifications for non-adjoint $G$, the proof of~\cite[Lemma 19]{Crooks:2020} to establish that
%----%
\begin{equation}
	n^\prime = \big (\widetilde{\rho}_N(\tilde{y})\delta_p\big)^{-1}~.
\end{equation}
%----%
Thus, every element of $G^*_{\theta_p}$ looks like $\overline{K}_p^{-1}(\tilde{y})$ for some $\tilde{y}\in\overline{P}^{-1}(p)$, and $K^{-1}_p$ is surjective.

Now to establish injectivity, let $x,y\in \overline{P}^{-1}(p)$ such that $\overline{K}^{-1}_p(x)=\overline{K}^{-1}_p(y)=n^\prime h n\in G^*_{\theta_p}$. We see that $Ad^*_{n}\theta_p = \pi_G(x)=\pi_G(y)$ and so $x,y$ are both in the fibre of $\pi_G$ above a common point on $\KT_\gf$. But by Lemma~\ref{lem:piT_equivariant}, this means $\pi_T(x)$ and $\pi_T(y)$ will differ by an action of $Z(G)$, and so will $\overline{K}^{-1}_p(x)$ and $\overline{K}^{-1}_p(y)$. Thus, we conclude that $x\overset{!}{=}y$ if they map to the same point under $\overline{K}_p^{-1}$.
\end{proof}

To summarise, the map $\lambda:\KT_G\rightarrow G^*$ is defined by
%----%
\begin{equation}
	\lambda(y)\coloneqq \delta_p \overline{K}^{-1}_p(y)\delta^{-1}_p~,
\end{equation}
%----%
for $y\in \KT_G$ and $p=\overline{P}(y)$. Unravelling the definitions again gives the more compact and algebraic form,
%----%
\begin{equation}
\label{eq:lambda_definition}
\begin{split}
	\lambda:~ &\KT_G \rightarrow G^*\\
	& y \mapsto \widetilde{\rho}_N(y)^{-1}\widetilde{w_0}\big(\pi_T(y)\big)^{-1}\rho_N(y)~.
\end{split}
\end{equation}
%----%
For any $y\in \KT_G$, we have that $\lambda(y)\in G^*_{\beta(y)}$, and as an immediate corollary of Theorem~\ref{thm:Kx_surjective}, we have the following:
%----%
\begin{corr}\label{corr:lambda_iso}
Let $p\in \gf^*/\!\!/G$ and let $s_p\in S_f$ be the corresponding Slodowy slice element under Kostant's isomorphism. Then, the restriction of $\lambda$ to $\overline{P}^{-1}(p)$ gives an isomorphism
%----%
\begin{equation}
	\lambda: \overline{P}^{-1}(p) \xrightarrow{\sim} G_{s_p}^*~.
\end{equation}
%----%
\end{corr}
%----%
\begin{proof}
Follows immediately from Theorem~\ref{thm:Kx_surjective} and the definition of $\lambda$.
\end{proof}
%----%
Thus, we have the morphism
%----%
\begin{equation}
\label{eq:def_varphi}
\begin{split}
	\varphi: \KT_G &\rightarrow \Zf_G~,\\
	x &\mapsto (\lambda(x), \beta(x))~.
\end{split}
\end{equation}
%----%
Let $\Zf^*_G\subset G^*\times S_f$ be the closed subscheme,
%----%
\begin{equation}
	\Zf^*_G = \{ (g,s)\in G^*\times S_f | \mathrm{Ad}^*_gs = s\}~.
\end{equation}
%----%
Note that $\Zf_G^* =( G^*\times S_f )\cap \Zf_G$ inside $G\times S_f$, and so is Zariski open in $\Zf_G$. The open subset $\Zf^*_G$ can alternatively be characterised as the Hamiltonian reduction
%----%
\begin{equation}
	\Zf^*_G \cong \big(\mu_{diag}^{-1}(0)\cap (G^*\times\gf^\ast_{\rm reg})\big)/\!\!/G_{diag}~.
\end{equation}
%----%

%----%
\begin{prop}\label{prop:phi_open_immersion}
The morphism $\varphi:\KT_G\rightarrow \Zf_G$ is an open immersion of $S_f$-schemes, with image equal to $\Zf^*_G$
\end{prop}
%----%
\begin{proof}
We have a foliation of $\KT_G$,
%----%
\begin{equation}
	\KT_G = \bigsqcup_{p\in \gf^*/\!\!/G} \overline{P}^{-1}(p)
\end{equation}
%----%
From Proposition~\ref{prop:beta_dominant}, the morphism $\beta$ can be used to rewrite the above as a foliation over $S_f$,
%----%
\begin{equation}
	\KT_G = \bigsqcup_{s\in S_f} \beta^{-1}(s)~.
\end{equation}
%----%
The leaf $\beta^{-1}(s)$ can be identifed with the leaf $\overline{P}^{-1}(p)$ when $s=s_p$. By Corollary~\ref{corr:lambda_iso}, the morphism $\lambda$ gives an isomorphism between each leaf of this foliation and $G^*_s$. Thus,
%----%
\begin{equation}
	\KT_G\cong \bigsqcup_{s\in S_f} G^*_s = \Zf_G^*~.
\end{equation}
%----%

By construction, $\varphi$ sends the leaf of $\KT_G$ over $s\in S_f$ to the leaf of $\Zf_G^*$ over the same point. Therefore the diagram,
%----%
\begin{equation}
	\begin{tikzcd}
	\KT_G \ar[dr,"\beta",swap] \ar[rr,"\varphi"]& & \Zf_G \ar[dl]\\
	& S_f &
	\end{tikzcd}
\end{equation}
%----%
commutes and so $\varphi$ is a morphism of $S_f$-schemes.

\end{proof}
We would like $\varphi$ to not just be an open immersion, but also a symplectic morphism.

\begin{prop}~\label{prop:varphi-symplectic}
For any simple, algebraic Lie group $G$, $\varphi:\KT_G\rightarrow \Zf_G$ is symplectic.
\end{prop}
\begin{proof}
Let $G_{ad}$ be a simple algebraic Lie group of adjoint type. From~\cite{Bezrukavnikov:2003}, it is known that
%----%
\begin{equation}
	\varphi_{ad}:\KT_{G_{ad}} \rightarrow \Zf_{G_{ad}}
\end{equation}
%----%
is symplectic.

Now let $G$ be a simple algebraic Lie group such that $\mathrm{Lie}~G = \mathrm{Lie}~ G_{ad}$. The universal centralisers are related by
%----%
\begin{equation}
	\Zf_{G_{ad}} \cong Z(G)\backslash \Zf_{G}~,
\end{equation}
%----%
where the centre $Z(G)$ acts by left multiplication on the fibres of $\Zf_{G}$. Let $\pi_\Zf$ be the projection $\Zf_{G}\rightarrow Z(G)\backslash \Zf_{G} \cong \Zf_{G_{ad}}$, which is symplectic. 

Consider the diagram,
%----%
\begin{equation}
	\begin{tikzcd}
	\KT_{G} \ar[d,"\pi_G"]\ar[r,"\varphi"] & \Zf_G \ar[d,"\pi_\Zf"]\\
	\KT_{G_{ad}}\ar[r,"\varphi_{ad}"] & \Zf_{G_{ad}}
	\end{tikzcd}
	~.
\end{equation}
%----%
Since $\varphi$ is $Z(G)$ equivariant, the square commutes. Let $\omega_{\Zf_{ad}}$ be the symplectic form on $\Zf_{G_{ad}}$, then
%----%
\begin{equation}
	\varphi^*_{ad}(\omega_{\Zf_{ad}}) = \omega_{\KT_\gf}~,
\end{equation}
%----%
where $\omega_{\KT_\gf}$ is the symplectic form on $\KT_\gf$. Since the symplectic form on $\KT_G$, $\omega_{\KT_G}$, is defined by pullback along $\pi_G$, we have that
%----%
\begin{equation}
	(\pi_G^*\circ \varphi_{ad}^*)(\omega_{\Zf_{ad}}) = \omega_{\KT_G}~,
\end{equation}
%----%
and $\varphi_{ad}\circ\pi_G$ is symplectic. As the square commutes,
%----%
\begin{equation}
	(\varphi^*\circ\pi_{\Zf}^*)(\omega_{\Zf_{ad}}) = (\pi_G^*\circ \varphi_{ad}^*)(\omega_{\Zf_{ad}}) = \omega_{\KT_G}~.
\end{equation}
%----%
But $\pi_{\Zf}^*(\omega_{\Zf_{ad}}) = \omega_{\Zf_G}$, where $\omega_{\Zf_G}$ is the symplectic form on $\Zf_G$. Therefore,
%----%
\begin{equation}
	(\varphi^*\circ \pi_{\Zf}^*)(\omega_{\Zf_{ad}}) = \varphi^*(\omega_{\Zf_G}) = \omega_{\KT_G}~,
\end{equation}
%----%
and $\varphi$ is symplectic.
\end{proof}

%----------------------------------------------------%
\subsection{\label{subsec:algorithm}An algorithmic construction for \texorpdfstring{$\varphi$}{varphi}}
%----------------------------------------------------%

With the morphism $\varphi:\KT_G\rightarrow \Zf_G$ established, we turn to the more practical problem of finding explicit expressions for the components of this map. Here we will restrict our attention to the case $G={\rm SL}_N$, whose structure as an algebraic group is particularly transparent. In the case of other algebraic Lie groups, this level of explicit detail would entail more work, though the story should be essentially analogous.

The embedding $\varphi$ consists of two parts, $\lambda$ and $\beta$. Let us first discuss $\beta$. An element $y\in \KT_\gf$ has form
%----%
\begin{equation}
	y =
	\begin{pmatrix}
	b_1(y) & \gamma_1(y) & 0 & \dots\\
	1 &  b_2(y)-b_1(y) & \gamma_2(y) & \dots\\
	0 & 1 & b_3(y)-b_2(y) & \dots\\
	\vdots & \vdots & \vdots & \ddots
	\end{pmatrix}~,
\end{equation}
%----%
when considered as an element of $\gf$.\!\footnote{Throughout this section we make liberal use of the Killing isomorphism to represent elements of $\gf^*$ as elements in $\gf$ without further comment.} The fibre of the map $\pi_G$ at $y$ consists of elements $\{\tilde{y}_k,~k=0,\ldots,N-1\}$ that are permuted cyclically under the action of $Z(G)$.

For $x\in \slf_n$ realised as a matrix in $\mathrm{End}(\IC^n)$, the fundamental invariants (up to conventions regarding overall scaling) can be computed as
%----%
\begin{equation}
	P_j(x) = \frac{1}{j+1} \mathrm{Tr}\big( x^{j+1}\big)~,\qquad \text{for } n=1,2,\dots,\rk~\gf~.
\end{equation}
%----%
The components of $\beta(y)$ are given by $P_1(y),P_2(y),\ldots P_{\rk_\gf}(y)$ and for a given value of $N$ can be easily computed by the trace formula above. To compute $\beta(\tilde{y}_k)$ we simply project to $y$ by $\pi_G$ before applying the expression above.

Next let us consider the more elaborate morphism $\lambda$. The definition \eqref{eq:lambda_definition} of $\lambda$ applied to a point in $\KT_G$ takes the form
%----%
\begin{equation}
	\lambda(\tilde{y}_k)  = \widetilde{\rho}_N(\tilde{y}_k)^{-1} \widetilde{w}_0\big(\pi_T(\tilde{y}_k)\big)^{-1}\rho_N(\tilde{y}_k)~.
\end{equation}
%----%
From the expression in \eqref{eq:piT-concrete}, the map $\pi_T$ is given by
%----%
\begin{equation}
\label{eq:pi_t_as_matrix}
	\pi_T =
	\begin{pmatrix}
	\widetilde{\gamma}_1 & 0 & 0 & 0 & \dots\\
	0 & \frac{\widetilde{\gamma}_2}{\widetilde{\gamma}_1} & 0 & 0 & \dots \\
	0 & 0 & \ddots & 0 & \dots \\
	\dots & 0 & 0 & \frac{\widetilde{\gamma}_{N}}{\widetilde{\gamma}_{N-1}} & 0 \\
	\dots & 0 & 0 & 0 &  \widetilde{\gamma}_{N}^{-1}
	\end{pmatrix}~.
\end{equation}
%----%
Like $\beta$, the maps $\rho_N$ and $\widetilde{\rho}_N$ only depend on the image of $\tilde{y}_k$ under $\pi_G$, whereas $\pi_T$ satisfies
%----%
\begin{equation}
	\pi_T(\tilde{y}_{k}) = \xi^{k}\pi_T(\tilde{y}_0)~.
\end{equation}
%----%

A key input into the maps $\rho_N$ and $\widetilde{\rho}_N$ is the inverse map $\psi^{-1}: \chi + \mathfrak{b}^* \rightarrow N\times S_f$ of the Kostant isomorphism. An element of $N\times S_f$ can be written as $(n,s)$, where we have
%----%
\begin{equation}
	n =
	\begin{pmatrix}
	1 & n_{1,1} & n_{1,2} & n_{1,3} & \dots\\
	0 & 1 & n_{2,1} & n_{2,2} & \dots\\
	0 & 0 & 1 & n_{3,1} & \dots\\
	\vdots &\vdots &\vdots &\vdots & \ddots
	\end{pmatrix}
	~,\qquad
	s =
	\begin{pmatrix}
	0 & p_1 & p_2 & p_3 & \dots\\
	1 & 0 & p_1 & p_2 & \dots\\
	0 & 1 & 0 & p_1 &\dots\\
	\vdots & \vdots & \vdots & \ddots
	\end{pmatrix}
	~.
\end{equation}
%----%
All entries are $\IC$-valued. Under $\psi$ the element $(n,s)$ is taken to the element of $\KT_\gf$ whose matrix entries are
%----%
\begin{equation}
\begin{split}
	Ad^\ast_{n}s=nsn^{-1} &=
	\begin{pmatrix}
	n_{1,1} & p_1 -(n_{1,1})^2 + n_{1,2} & \ast 	& \dots\\
	1 		& n_{2,1} - n_{1,1} 		 & p_1-(n_{2,1})^2 + n_{2,1}n_{1,1} + n_{2,2} - n_{2,1}			& \dots\\
	0 		& 1 						 & n_{3,1}-n_{2,1}	& \dots\\
	\vdots 	& \vdots 					 & 	\vdots			& \ddots
	\end{pmatrix}\\
	&\overset{!}{=}
	\begin{pmatrix}
	b_1 & \gamma_1 & 0 & \dots\\
	1 & b_2 -b_1 & \gamma_2 & \dots \\
	0 & 1 & b_3 - b_2 & \dots\\
	\vdots & \vdots & \vdots & \vdots
	\end{pmatrix}=y
	~.
\end{split}
\end{equation}
%----%
The inverse morphism $\psi^{-1}$ is determined by systematically solving this equation for the coordinates $n_{i,j}$ and $p_i$ on $N\times S_f$. In practice, equality of diagonals sets $n_{i,1}=b_i$, and qquating the first superdiagonals then fixes the values of $n_{i,2}$ and $p_1$. The next superdiagonal fixes $n_{i,3}$ and $p_2$, and so on continuing to higher superdiagonals. At each step, the new unfixed coordinates appear only linearly. The value of $\rho_N(\tilde{y}_k)$ is then just $n$, re-expressed in terms of the $b$ and $\tilde\gamma$ variables according to this procedure.

The map $\widetilde{\rho}_N$ can be determined in the same manner. Note that the conjugate $\mathrm{Ad}_{\widetilde{w}_0(\pi_T(\tilde{y}_0))^{-1}} y$ will be the reflection of $y$ along the antidiagonal. To see this, observe that the element $\pi_T(\tilde{y}_0)^{-1}$ leaves the $b_i$ invariant and contracts each $\gamma_i$ entry to $1$ while dilating $\chi$ to $\sum_{i=1}^{\rk\,\gf} \gamma_ie_{-\alpha_i}$ (as the negative simple roots transform with opposite weights). In effect, this reflects $y$ across the diagonal. On the other hand, the element $\widetilde{w}_0$, acts by sequential reflection across the diagonal and the antidiagonal, so in combination the effect is that $\mathrm{Ad}_{\widetilde{w}_0(\pi_T(\tilde{y}_0))^{-1}}$ is the antidiagonal reflection. Consequently, in practice the expression for $\widetilde{\rho}_N(y)$ can simply be obtained from the $n_{i,j}$ above by making the replacement $(b_i,\gamma_i)\leftrightarrow(-b_{N-1-i},\gamma_{N-1-i})$.

For the purpose of realising the universal centraliser as a symplectic variety, this procedure can be expedited as a result of Proposition~\ref{prop:ring_generators_sln} and Theorem~\ref{thm:Poisson_generation}. These imply that rather than constructing $\lambda$ in its entirety, it is sufficient to determine just the lower left-hand corner entry.

By definition, both $\rho_N(\tilde{y}_k)$ and $\widetilde{\rho}_N(\tilde{y}_k)$ are upper unitriangular matrices, while $\widetilde{w}_0\big(\pi_T(\tilde{y})_k\big)^{-1}$ is supported along the antidiagonal as in \eqref{eq:pi_t_as_matrix}. The bottom row of the stabilising group element is therefore completely fixed by the first row of the map $\rho_N$, and what's more, the lower-left corner entry is just equal to the reciprocal of the upper-left corner entry of $\pi_T$, so in particular we have
%----%
\begin{equation}
	X \coloneqq \widetilde{\gamma}^{-1}_1= \prod_{j=1}^{\rk~\gf} \gamma_j^{-\widetilde{a}_{1j}}~.
\end{equation}
%----%
From this and the expression for the $P_j$ functions, Theorem~\ref{thm:Poisson_generation} allows us to generate all of $I_{{\rm SL}_N}$.

For $G\neq \rm SL_N$ we have a \textit{conjectural} expression for $X$, coming from the discussion in Section~\ref{subsec:sln_poisson_structure}. Suppose that $G$ is not $\mathrm{Spin}(4N)$, then the lowest dimensional faithful representation, $V_F$ is irreducible. Let $(a_1,a_2,\dots_,a_{\rk\,\gf})$ be the components of the highest weight of $V_F$ \textit{in the simple root basis}. We conjecture that
%----%
\begin{equation}
\label{eq:finite_X_def}
	X = \prod_{j=1}^{\rk\,\gf} \gamma_j^{-a_{j}}~.
\end{equation}
%----%
For $G={\rm SL}_N$ this specialises to $X= \prod_{j=1}^{\rk\,\gf} \gamma_j^{-\widetilde{a}_{1j}}$, where $\widetilde{a}_{ij}$ are again the entries of the inverse Cartan matrix. Similarly, if $G$ is ${\rm Spin}(4N)$ we define $X_{\rm s}$ and $X_{\rm c}$ analogously to \eqref{eq:finite_X_def} using the highest weights of $V_{\rm s}$ and $V_{\rm c}$ respectively. 

%----------------------------------------------------%
\subsection{\label{subsec:KT_SL3}The embedding for \texorpdfstring{${\rm SL}_3$}{SL3}}
%----------------------------------------------------%

The ring of functions on the Kostant--Toda lattice for $\slf_3$ is given by
%----%
\begin{equation}
\label{eq:ktg_sl3_coordinate_ring}
	\IC[\KT_{\gf}] = \IC[b_1,b_2, \gamma_1^{\pm 1},\gamma_2^{\pm 1}]~.
\end{equation}
%----%
Following the general prescription, for $G={\rm SL}_3$ the covering space $\KT_G$ comes with coordinate ring
%----%
\begin{equation}
\label{eq:ktG_SL3_coordingate_ring}
	\IC[\KT_G] = \IC[b_1,b_2,\widetilde{\gamma}^{\pm1}_1,\widetilde{\gamma}^{\pm1}_2]~,
\end{equation}
%----%
where $\widetilde{\gamma}_1 = \gamma^{-\frac{2}{3}}_1\gamma_2^{-\frac{1}{3}}$ and $\widetilde{\gamma}_2 = \gamma^{-\frac{2}{3}}_2\gamma_1^{-\frac{1}{2}}$, and the Poisson bracket follows accordingly. We can make explicit the maps involved in the definition of $\varphi$. For $\pi_T:\KT_{G} \rightarrow T$ we have (at the level of rings of functions)
%----%
\begin{equation}
\label{eq:piT_map_SL3}
	t_1 \mapsto \widetilde{\gamma}_1~,\quad t_2 \mapsto \widetilde{\gamma}_2~.
\end{equation}
%----%
The longest element of the Weyl group of $A_2$ can be lifted to $N_{{\rm SL}_3}(T)$ as
%----%
\begin{equation}
\label{eq:sl3_longest_weyl_element}
	\widetilde{w}_0 =
	\begin{pmatrix}
		0 & 0 & -1 \\
		0 & -1 & 0 \\
		-1 & 0 & 0
	\end{pmatrix}~.
\end{equation}
%----%
The coordinate ring of the Slodowy slice, $S_f$, is embedded (up to normalisation) as
%----%
\begin{equation}
\label{eq:sl3_slodowy_functions}
\begin{split}
	S &= \frac12\Tr(y^2)=b_1^2 + b_2^2 -b_1b_2 +\gamma_1 + \gamma_2~,\\
	W &= \frac13\Tr(y^3)=b_1^2b_2-b_1b_2^2 +\gamma_1b_2 - \gamma_2 b_1~.
\end{split}
\end{equation}
%----%
The image of $\lambda$ is realised as follows. Solving for the matrix $\rho_N(y)$ as described above, one finds
%----%
\begin{equation}
\label{eq:sl3_rho_N}
	\rho_N(y)=
	\begin{pmatrix}
	~~1~~ & ~~b_1~~ & \frac12(b_1^2+b_1b_2-b_2^2+\gamma_1-\gamma_2) \\
	0 & 1 & b_2 \\
	0 & 0 & 1
	\end{pmatrix}
	~.
\end{equation}
%----%
The group element $\tilde w_0(\pi_T(y))^{-1}$ is given as follows,
%----%
\begin{equation}
\label{eq:intermediate_group_element}
	\tilde w_0(\pi_T(y))^{-1}=
	\begin{pmatrix}
		0 & 0 & -1 \\
		0 & -1 & 0 \\
		-1 & 0 & 0
	\end{pmatrix}
	\begin{pmatrix}
		\tilde\gamma_1 & 0 & 0 \\
		0 & \frac{\tilde\gamma_2}{\tilde\gamma_1} & 0 \\
		0 & 0 & \frac{1}{\tilde\gamma_2}
	\end{pmatrix}^{-1}
	=
	\begin{pmatrix}
		0 & 0 & -\gamma_1^{\frac13}\gamma_2^{\frac23} \\
		0 & -\frac{\gamma_2^{\frac13}}{\gamma_1^{\frac13}} & 0 \\
		-\frac{1}{\gamma_1^{\frac23}\gamma_2^{\frac13}} & 0 & 0
	\end{pmatrix}
	~.
\end{equation}
%----%
Solving for the matrix $\widetilde{\rho_N}(y)$ then produces the result
%----%
\begin{equation}
\label{eq:sl3_rho_N_tilde}
	\widetilde{\rho_N}(y)=
	\begin{pmatrix}
	~~1~~ & ~~-b_2~~ & \frac12(-b_1^2+b_1b_2+b_2^2-\gamma_1+\gamma_2) \\
	0 & 1 & -b_1 \\
	0 & 0 & 1
	\end{pmatrix}
	~.
\end{equation}
%----%
Putting everything together, this gives the following parameterisation of (an open subset of) stabiliser of the point $\beta(y)\in S_f$,
%----%
\begin{equation}
\label{eq:sl3_bottom_row_classical}
	\lambda(y)=-\frac{1}{\gamma_1^{\frac23}\gamma_2^{\frac13}}
	\begin{pmatrix}
	b_1^2+\gamma_1-\frac{S}{2} & W+\frac{b_1S}{2} & b_1 W+\frac{S^2}{4} \\
	b_1 & b_1^2+\gamma_1 & W+\frac{b_1S}{2} \\
	1 ~&~ b_1 ~&~ b_1^2+\gamma_1-\frac{S}{2}
	\end{pmatrix}
	~,
\end{equation}
%----%
where to simplify the expression we have used $S$ and $W$ as defined in \eqref{eq:sl3_slodowy_functions}. We define the additional algebraic generators of the coordinate ring as the entries on the bottom row of $\lambda(y)$,
%----%
\begin{equation}
\label{eq:sl3_matrix_generators_classical}
\begin{split}
	X &\coloneqq  -\gamma_1^{-\frac{2}{3}}\gamma_2^{-\frac{1}{3}}~,\\
	Y &\coloneqq-\-\gamma_1^{-\frac{2}{3}}\gamma_2^{-\frac{1}{3}} b_1~,\\
	Z &\coloneqq -\frac{1}{2}\gamma_1^{-\frac{2}{3}}\gamma_2^{-\frac{1}{3}} (\gamma_1-\gamma_2 + b_1^2-b_2^2+b_1b_2)~.
\end{split}
\end{equation}
%----%
By Proposition~\ref{prop:ring_generators_sln}, only the bottom row is needed to uniquely specify $\lambda(y)$, and indeed, the remaining rows can be expressed in terms of the $S,W,X,Y,Z$ (and could be determined directly by solving the $[\lambda(y),\beta(y)]=0$). Altogether, we have for the embedding $\varphi(y)$,
%----%
\begin{equation}
	\varphi(y) = \left(
	\begin{pmatrix}
	Z ~&~ WX + \frac12 SY ~&~ \frac14 S^2X+ W Y \\
	Y & \frac12 SX +  Z &  W X + \frac12 S Y \\
	X & Y & Z
	\end{pmatrix}
	~,~
	\begin{pmatrix}
	0 & \frac12 S &  W \\
	1 & 0 & \frac12 S \\
	0 & 1 & 0
	\end{pmatrix}
	\right)~.
\end{equation}
%----%
Note that this agrees with the presentation of Section~\ref{subsec:eq_rank_two}. The determinantal relation
%----%
\begin{equation}
\begin{split}
	D_{{\rm SL}_3}=  &-\frac{1}{8} S^3 X^3 + W^2 X^3 + \frac{1}{2} S W X^2 Y + \frac{1}{2} S^2 X Y^2 + W Y^3 - \frac{1}{4} S^2 X^2 Z\\ & - 3 W X Y Z - S Y^2 Z + \frac{1}{2} S X Z^2 + Z^3 \overset{!}{=}1~,
\end{split}
\end{equation}
%----%
is satisfied identically by the expressions for $S,W,X,Y,Z$ in terms of $\KT_G$ variables. The Poisson structure is inherited from the structure on $\KT_G$, and one can verify that the Poisson brackets also agree with~\eqref{eq:sl3_Poisson}.

%----------------------------------------------------------------------%
\section{\label{sec:chiralising_centraliser}Chiralising the universal centraliser}
%----------------------------------------------------------------------%

In~\cite{Arakawa:2018egx}, the (critical level) chiral universal centraliser ($\mathbf{I}_G$, denoted in \emph{loc.} {\it cit.} by $\mathbf{I}_G$), was defined to be the two-sided quantum Drinfel'd--Sokolov reduction of critical level chiral differential operators (CDOs) on $G$ ($\D_{G,\kappa_c}$).\!\footnote{The chiral universal centraliser $\mathbf{I}_{{\rm SL}_2}$ at \emph{non-critical level} appeared previously in~\cite{Frenkel:200657} as the regular representation of the Virasoro algebra.} The operation of Drinfel'd--Sokolov reduction is a chiral analogue of Kostant--Whittaker reduction, and $\D_{G,\kappa_c}$ represents a \textit{chiral quantisation} (defined in the following section) of $T^*G$, so $\mathbf{I}_G$ is indeed a natural chiral analogue of $\Zf_G$. By virtue of its definition in terms of DS reduction, it should be identified with the VOA assigned to the sphere of type $G$ in class $\SS$ \cite{Beem:2014rza}.

In the previous section, we presented a method for realising $I_G=\IC[\Zf_G]$ by way of the Kostant--Toda lattice $\KT_G$. The open immersion $\varphi:\KT_G\rightarrow \Zf_G$ gives a realisation of $I_G$ in terms of functions on an algebraically trivial space, which we think of as a (finite-dimensional, classical) ``free field realisation'' of $I_G$. In this section we describe a chiral analogue of this construction, and propose a strategy for producing (chiral, quantum) free field realisations for $\mathbf{I}_G$ for general $G$. We make this completely explicit for the cases $G={\rm SL}_2$ and $G={\rm SL}_3$.

In this section $G$ will be assumed to be simply connected unless otherwise stated.

%----------------------------------------------------%
\subsection{\label{subsec:chiral_KT_lattice}Chiralising the Kostant--Toda lattice}
%----------------------------------------------------%

The Kostant--Toda lattice is isomorphic as a symplectic variety to $(T^\ast\IC^\times)^{\rk\,\gf}$. The standard chiral quantisation(s) of this space are given by CDOs on $(\IC^{\times})^{\rk\,\gf}$. In~\cite{Beem:2019tfp}, these CDOs were presented in terms of a lattice boson system. We will review this construction here in connection with the Kostant--Toda lattice and its refinement.

Let chiral bosons $\phi_i$, $\delta_j$ for $i,j=1,2,\dots,\rk\,\gf$ be normalised according to
%----%
\begin{equation}
\label{eq:chiral_boson_opes}
	\phi_i(z)\phi_j(w) \sim -\delta_{ij}\log(z-w)~,\quad \delta_i(z)\delta_j(w) \sim \delta_{ij}\log(z-w)~.
\end{equation}
%----%
Let $\Pi_\gf$ be the isotropic lattice vertex algebra generated by $\partial\phi_i$, $\partial\delta_i$, and vertex operators, $e^{n(\phi_i+\delta_i)}$, for $n\in\IZ$. Since the $\phi_i+\delta_i$ operator has regular self-OPE, the exponential vertex operators form a commutative vertex algebra. This is a ``half-lattice vertex algebra'' in the sense of \cite{half_lattice}.

In fact, $\Pi_\gf$ is strongly generated by the operators
%----%
\begin{equation}
\label{eq:half_lattice_strong_generators}
	\gamma_i^{\pm 1} = e^{\pm(\phi_i+\delta_i)}~,\quad b_i = \frac{1}{2}\partial(\delta_i - \phi_i)~,
\end{equation}
%----%
whose only nonvanishing singular OPEs are given by
%----%
\begin{equation}
\label{eq:bgamma_ope}
	b_i(z)\gamma_j^{\pm 1}(w)\sim \frac{\pm\delta_{ij}\gamma_j^{\pm1}(w)}{z-w}~,
\end{equation}
%----%
and which also satisfy the suggested relation $(\gamma_i^{-1}\gamma_i^{+1})=1$. This presentation is the one relevant to description of this vertex algebra as CDOs on the multiplicative group $(\IC^{\times})^{\rk\,\gf}$, which itself can be thought of as the localisation of $\rk\,\gf$ many $\beta\gamma$ systems (\ie, CDOs on $\IC^{\rk\,\gf}$~\cite{Malikov:1999}). The lattice representation then arises via Friedan--Martinec--Shenker bosonisation of the $\beta\gamma$ system.

This presentation in terms of lattice bosons makes the following generalisation natural. We pass to a refined lattice with rational spacings, specifically with denominator $|Z(G)|$ for $G$ simply connected, and define the vertex subalgebra $\Pi_G$ to be strongly generated by
%----%
\begin{equation}
\label{eq:refined_generators}
	\widetilde{\gamma}^{\pm 1}_i = \prod_{i=1}^{\rk{\gf}} e^{\pm\tilde{a}_{ij}(\phi_j+\delta_j)}~,\quad b_i = \frac{1}{2} \partial(\delta_i - \phi_i) ~,
\end{equation}
%----%
where, as previously, the $\tilde{a}_{ij}$ are matrix elements of the inverse Cartan matrix of $\gf$. Note that the $b_i$ are unchanged. For $\gf=\slf_2$, this is precisely the lattice $\Pi_{\frac{1}{2}}$ appearing in the free field constructions of~\cite{Beem:2019tfp} (and previously in \cite{Adamovic2019}). The vertex algebra $\Pi_G$ admits a natural action of $Z(G)$ according to the prescription in~\eqref{eq:ZGactionKTG}, extended to be $\IC[\partial]$-linear.

Now suppose $G^{\prime}$ is a non-simply connected algebraic Lie group with the same Lie algebra as $G$. Then $G^{\prime} \cong Z^{\prime}(G)\backslash G$ for some central subgroup $Z^{\prime}(G)\subset Z(G)$. Let $\gamma_i'$ be the generators of the ring $\IC[\widetilde{\gamma}_i]^{Z^{\prime}(G)}$. Then we define a lattice vertex subalgebra $\Pi_{G^{\prime}}\subset\Pi_G$ that is strongly generated by the $b_i$ and the $Z(G)$ invariants, $(\gamma_i')^{\pm1}$.

There is an obvious inclusion of vertex algebras $\Pi_\gf\hookrightarrow \Pi_G$, given by
%----%
\begin{equation}
\label{eq:Z_action_on_lattice}
	\gamma_i^{\pm1} \mapsto \prod_{j=1}^{\rk\,\gf} (\widetilde{\gamma}_j^{\pm1})^{a_{ij}}~,\qquad b_i \mapsto b_i~.
\end{equation}
%----%
The image of $\Pi_\gf$ inside $\Pi_G$ is precisely the $Z(G)$ invariant subalgebra of $\Pi_G$---indeed, for the case that $G^{\prime}$ is the adjoint group, we have $\Pi_{G^{\prime}}=\Pi_\gf$.

Our vertex algebras $\Pi_G$ and the Kostant--Toda lattices are related in the sense that the former is a chiral quantisation of the latter. The relationship between vertex algebras and geometric spaces were explored by Arakawa--Moreau \cite{Arakawa:2012c2, Arakawa2017introduction, Arakawa_Moreau:2018, Arakawa_Moreau:2017157, Arakawa_Moreau:2018542}. We shall recall some of the relevant definitions now. For a more pedagogical introduction, we recommend the monograph \cite{Arakawa_Moreau:jet}.

For a scheme $X$ over $\IC$, we denote its arc space by $J_\infty X\coloneqq\mathrm{Hom}(\ID,X)$, where $\ID=\mathrm{Spec}~\IC[\![t]\!]$ is the formal disc. For $X$ Poisson, functions on the arc space $\OO(J_\infty X)$ has a canonical vertex Poisson algebra structure~\cite{Arakawa:2012c2}.

Every vertex algebra has a canonical decreasing filtration, $F^\bullet$, known as the Li filtration~\cite{Li:2005}. Specifically, for $V$ a vertex algebra with vacuum $|0\rangle$, we have
%----%
\begin{equation}
\label{eq:Li_filtration}
	F^iV = \bigg\{v_{1,(-n_1-1)}v_{2,(-n_2-1)}\dots v_{k,(-n_k-1)}|0\rangle\,\bigg|\, k\in \IN\,,~~ v_{1},v_2,\dots,v_k\in V\,,~~ \sum_{j=1}^{k} n_j \geqslant i\,\bigg\}~.
\end{equation}
%----%
$V$ is \textit{separated} if $\bigcap_{i=1} F^i V=0$ and if $V$ is a separated vertex algebra, then the vector space $R_V\coloneqq V/F^1(V)$ is a Poisson algebra, with multiplication arising from the normally ordered product and the Poisson bracket coming from simple poles in the OPE. This Poisson algebra is called the \textit{Zhu's $C_2$ algebra} of $V$, and $X_V\coloneqq \mathrm{Specm}\,R_V$ is the \emph{associated variety} of $V$. $\overline{X}_V\coloneqq\mathrm{Spec}\,R_V$ is then the \emph{associated scheme} of $V$.

Now let $X$ be an affine, Poisson scheme; a vertex algebra $V$ is said to be a \textit{chiral quantisation} of $X$~\cite{Arakawa2017introduction} if, as Poisson varieties,
%----%
\begin{equation}
\label{eq:chiral_quantisation}
	X_V \cong X_{red}~.
\end{equation}
%----%
It is a \textit{strict chiral quantisation} of $X$ if furthermore,
%----%
\begin{equation}
\label{eq:strict_chiral_quantisation}
	\overline{X}_V\cong X\qquad \mathrm{and}\qquad\mathrm{Spec}(\mathrm{gr}_F\,V) \cong J_\infty X~.
\end{equation}
%----%
A vertex algebra satisfying $\mathrm{Spec}~(\mathrm{gr}_F V) \cong J_\infty \mathrm{Spec}~\overline{X}_V$ is also known as \textit{classically free}.

The vertex algebra $\Pi_\gf$ is a strict chiral quantisation of $T^\ast\big(\IC^{\times}\big)^{\rk\,\gf}$, and given its definition in terms of strong generators, it is straightforward to see that $R_{\Pi_G}\cong \IC[b_i,\widetilde{\gamma}_j\,|\,i,j=1,\dots,\rk\,\gf]$ with the expected Poisson brackets. Therefore, $\overline{X}_{\Pi_G}$ is reduced and can be identified with $\KT_G$ as a symplectic variety. Moreover, $\Pi_G$ is classically free, so $\Pi_G$ is a strict chiral quantisation of $\KT_G$ for any $G$. Moreover, the inclusion $\Pi_\gf\hookrightarrow \Pi_G$ as the $Z(G)$ invariant subalgebra descends, at the level of associated schemes, to the covering map $\pi_G^\#: \OO(\KT_\gf)\hookrightarrow \OO(\KT_G)$.

Finally, let us briefly remark on the conformal structure of these vertex algebras. Chiral bosons admit a well-known family of conformal vectors, which in our case amount to a $\IC^{2\rk\,\gf}$ parameter space of conformal vectors of the form
%----%
\begin{equation}
\label{eq:lattconf}
	T = \sum_{i} \frac{1}{2}\bigg(  \partial \delta_i \partial \delta_i - \partial \phi_i\partial \phi_i + \alpha_i\partial^2 \phi_i + \beta_i \partial^2 \delta\bigg)~,
\end{equation}
%----%
where $\alpha_i,\beta_i\in\IC$ are background charges for $\phi_i$ and $\delta_i$, respectively. Tuning the values of $\alpha_i$ and $\beta_i$ changes the conformal grading on $\Pi_G$, and in particular changes the weight of the $\gamma_i$ operators. To match with the $\IC^\ast$ grading on $\KT_G$, we require the weights $\Delta_{\gamma_i} = 2$ and $\Delta_{b_i} = 1$, which is accomplished by setting $\alpha_i+\beta_i=-2$. With no further input, this leaves $\rk\,\gf$ many unconstrained background charges, one combination of which could further be constrained by demanding that the central charge match the expected value for the chiral universal centraliser,
%----%
\begin{equation}
\label{eq:sphere_central_charge}
	c_{\gf}=\rk\,\gf + 48 (\rho_\gf,\rho_\gf^\vee)~,
\end{equation}
%----%
where $\rho_\gf$ is the half-sum of positive roots and $\rho_\gf^\vee$ is the half-sum of positive coroots. For simply laced $\gf$, we can apply the strange formula of Freudenthal and de Vries to simplify this expression to
%----%
\begin{equation}
\label{eq:simply_laced_central_charge}
	c_{\gf}  =\rk\,\gf + 2 h_\gf^\vee\,\dim\,\gf~,
\end{equation}
%----%
where $h_\gf^\vee$ is the dual Coxeter number.

%----------------------------------------------------%
\subsection{\label{subsec:chiral_morphism}A chiral embedding and free field realisation}
%----------------------------------------------------%

The aim in this section is to identify a candidate for a free field realisation of the chiral universal centraliser by constructing a chiral version of the morphism $\varphi:\KT_G\rightarrow \Zf_G$. Recall that $\varphi$ decomposes into a map $\beta$ to the Slodowy slice $S_f$ and a map $\lambda$ to the centraliser in $G$ of $\beta(y)$. We will treat the chiral version in a similar matter.

The chiral analogue of $\IC[S_f]$ is the Feigin--Frenkel centre $\zf(\gf)$ \cite{Frenkel:2007}. Indeed, as $\mathbf{I}_G$ is defined as the two-sided Drinfel'd---Sokolov reduction of $\D_{G,\kappa_c}$, it has a subalgebra isomorphic to $\zf(\gf)$ (\cf\ \cite{Arakawa:2018egx,Beem:2022mde} for discussion of Feigin--Frenkel centres in chiral algebras of class $\SS$), so there is a map $\zf(\gf)\rightarrow \mathbf{I}_G$, the image of which should appear in our free field realisation as a Feigin--Frenkel center subalgebra in $\Pi_G$.

To construct this Feigin--Frenkel subalgebra, we specify an Ansatz by asserting a chiral version of the map $\beta$. This Ansatz will be a quantum corrected version of the expressions for the generators of $\IC[S_f]$, which arise from the ring map $\beta^\#: \IC[S_f]\rightarrow \IC[\KT_G]$.

We make the following observation: the morphism $\beta$ comes from the restriction of $\IC[\gf^*]^G$ to $\KT_\gf$, so in particular the functions on $\mathrm{im}\,\beta$ extend smoothly to the closure $\overline{\KT}_\gf= \chi+\mathfrak{t}^* + \bigoplus_{\alpha\in\Delta}\gf_\alpha$. Therefore, the expressions for the generators will only feature positively graded generators of $\IC[\KT_\gf]$. Furthermore, the classical map $\beta$ does not require passage to $\KT_G$, so the expressions for $\beta^\#(P_i)$ only feature positive integer powers of the $\gamma_i$ and $b_i$.

We can construct a chiral analogue of $\beta$ in the following manner. Let $V^{0}(\mathfrak{b})$ be the universal affine vertex algebra associated to the Borel subalgebra $\mathfrak{b}$ at level zero. This algebra contains a descending chain of vertex ideals, which are chiral analogues of the lower central series of $\mathfrak{n}$, given by
%----%
\begin{equation}
\label{eq:descending_vertex_chain}
	V^{0}(\mathfrak{b})\supset V(\nf) \supset V([\nf,\nf]) \supset V([\nf,[\nf,\nf]])\supset\dots\supset 0~,
\end{equation}
%----%
where the chain terminates in finitely many steps since $\nf$ is solvable. The quotient $V^{0}(\mathfrak{b})/V([\nf,\nf])$ is a chiral quantisation of $\overline{\KT}_\gf\cong (\mathfrak{b}/[\nf,\nf])^*$, and there is a corresponding inclusion $V^0(\mathfrak{b})/V([\nf,\nf])\hookrightarrow \Pi_\gf$, where we can view $\Pi_\gf$ as a localisation of $V^0(\mathfrak{b})/V([\nf,\nf])$.

Starting with $V^\kappa(\gf)$, the universal affine vertex algebra of $\gf$ at level $\kappa$, principal Drinfel'd--Sokolov reduction returns the (universal) principal $\WW$-algebra $\WW_{prin}(\gf,\kappa)$. This $\WW$-algebra can be realised as a subalgebra of the semi-infinite BRST complex that computes the reduction, with explicit expressions for the generators of this embedded subalgebra computed by the tic-tac-toe algorithm of \cite{deBoer:1993iz}. In particular, the generators of $\WW_{prin}(\gf,\kappa)$ can be expressed entirely in terms of the so-called \emph{hatted current subalgebra}, which itself is isomorphic to $V^{\kappa-\kappa_c}(\mathfrak{b})$---note the shift by the critical level $\kappa_c$. It was also observed in~\textit{loc.\ cit}.\  that passing to the quotient $V^{\kappa-\kappa_c}(\mathfrak{b})/V(\nf)\cong V^{\kappa-\kappa_c}(\mathfrak{t})$ produces the usual quantum Miura transformation, expressing generators of $\WW_{prin}(\gf,\kappa)$ in terms of $\rk\,\gf$ many Abelian currents.

At the critical level, the principal $\WW$-algebra is isomorphic to the Feigin--Frenkel centre, \ie, $\WW_{prin}(f,\kappa_c) \cong \zf(\gf)$. One therefore has an embedding $\zf(\gf)\hookrightarrow V^0(\mathfrak{b})$. Since $V([\nf,\nf])$ is a vertex ideal, this map can be extended to the quotient giving a map $\zf(\gf)\hookrightarrow V^0(\mathfrak{b})\twoheadrightarrow V^0(\mathfrak{b})/V([\nf,\nf])$, which gives rise to a commutative vertex subalgebra---that is isomorphic to the Feigin--Frenkel centre---contained in $V^0(\mathfrak{b})/V([\nf,\nf])$. Our candidate chiralisation of $\beta$ is then given by the further composition
%----%
\begin{equation}
\label{eq:chiral_beta_map}
	\beta^{ch}: \zf(\gf)\hookrightarrow V^0(\mathfrak{b})/V([\nf,\nf])\hookrightarrow \Pi_\gf~.
\end{equation}
%----%
Concretely, $\beta^{ch}(P_i)= (P_i^{t})|_{\mathfrak{b}/[\nf,\nf]}$, where $P_i^{t}$ are the tic-tac-toed expressions for the generators of $\WW_{prin}(\gf,\kappa_c)$, as detailed in \cite{deBoer:1993iz} and $|_{\mathfrak{b}/[\nf,\nf]}$ denotes the projection $V^0(\mathfrak{b})\twoheadrightarrow V^0(\mathfrak{b})/V([\nf,\nf])$. In a sense, $\beta^{ch}$ is a deformed version of the Miura transformation. Of course, this construction also applies away from the critical level, giving rise to an embedding $\beta^{ch}:\WW_{prin}(\gf,\kappa)\hookrightarrow V^{\kappa-\kappa_c}(\mathfrak{b})$.

%----%
\begin{rem}\label{rem:deformed_Miura}
%----%

The above construction is actually just one in a series of embedded realisations of principal $\WW$-algebras. Let $\mathfrak{n}_i$ be the lower central series of $\nf$, \ie,
%----%
\begin{equation}
\label{eq:lower_central_series}
	\nf_j = \underbrace{[\nf,[\nf,[\nf,[\dots,\nf]]\dots]]}_{j}~,
\end{equation}
%----%
for $j=1,\dots,d_\gf$, where $\nf_{d_\gf}$ is the last non-zero term. There is a system of embeddings $\beta^{ch}_{j}:\WW_{prin}(\gf,\kappa)\hookrightarrow V^{\kappa-\kappa_c}(\mathfrak{b})/ V(\nf_k)$ for $j=1,\dots,d_\gf+1$, which maps generators according to
%----%
\begin{equation}
\label{eq:generalised_Miura_maps}
	\beta^{ch}_j(P_i) = P_i^{t}|_{\mathfrak{b}/\nf_k}~,
\end{equation}
%----%
where $|_{\mathfrak{b}/\nf_j}$ is the projection $V^0(\mathfrak{b})\twoheadrightarrow V^0(\mathfrak{b})/V(\nf_j)$. The usual Miura transform is given by $\beta^{ch}_1$, our embedding corresponds to $\beta^{ch}_2$, and $\beta^{ch}_{d_\gf+1}$ is the usual tic-tac-toe embedding from \cite{deBoer:1993iz}.

It would be interesting to know whether these generalised Miura maps have a more geometric interpretation, \textit{viz.}, the relation of the Miura map to opers in the critical case.

%----%
\end{rem}
%----%

In practice, an alternative strategy for calculating these generators would be to write down the most general correction to the classical expression, $\beta^\#(P_i)$, and then fix free coefficients by imposing that all OPEs amongst the generators are regular. In principle, one might worry that the ``most general correction'' may have an infinite number of unfixed parameters, because the vertex algebra $\mathbf{I}_{G}$ is not conical. However, based on our discussion of $\beta^{ch}$ we should only look for corrections which contain positive powers of the $(\gamma_i)_{i=1}^{\rk\,\gf}$---of which there are only finitely many. In the examples in the following sections, the expressions that one arrives at by imposing these regularity constraints are precisely the ones obtained by actually computing $\beta^{ch}$ with the tic-tac-toe algorithm.

With a realisation of the Feigin--Frenkel centre established, we turn to the additional generators of $\mathbf{I}_G$. We shall proceed by defining a free-field expression for $X$ that arises from the finite setting in \eqref{eq:finite_X_def}. Suppose that $G$ is not simply-connected of type $D_{2N}$, then $V_F$ is irreducible. Let $(a_1,a_2,\dots_,a_{\rk\,\gf})$ be the components of the highest weight of $V_F$ in the simple root basis. We then define
%----%
\begin{equation}
\label{eq:chiral_X_def}
	X \coloneqq \prod_{j=1}^{\rk\,\gf} \gamma_j^{-a_{j}}~.
\end{equation}
%----%
Similarly, if $G$ is ${\rm Spin}(2N)$ we define $X_{\rm s}$ and $X_{\rm c}$ analogously to \eqref{eq:chiral_X_def} using the highest weights of $V_{\rm s}$ and $V_{\rm c}$ respectively. We now formulate the chiral version of Conjecture \ref{conj:poisson_generation}:
%----%
\begin{conj}
\label{eq:chiral_generator_conjecture}
	Suppose $G$ is not $\mathrm{Spin}(2N)$. Then the vertex algebra $\mathbf{I}_G$ is isomorphic to the vertex subalgebra of $\Pi_{G}$ that is weakly generated by the generators $(\beta^{ch}(P_i))_{i=1}^{\rk\,\gf})$ of the Feigin--Frenkel centre and by the generator $X$.

	If $G$ is $\mathrm{Spin}(2N)$, then $\mathbf{I}_G$ is isomorphic to the vertex subalgebra of $\Pi_G$ that is weakly generated by the generators $(\beta^{ch}(P_i))_{i=1}^{\rk\,\gf})$ of the Feigin--Frenkel centre and by the generators $X_{\rm s}$ and $X_{\rm c}$.
\end{conj}
%----%

Again inspired by the finite dimensional setting we have beliefs on the set of strong generators of $\mathbf{I}_G$.

%----%
\begin{conj}
Suppose $G$ is not $\mathrm{Spin}(2N)$. Then $\mathbf{I}_G$ is strongly generated by the generators $(P_i)_{i=1}^{\rk\,\gf}$ of the Feigin--Frenkel centre and generators $(\hat{g}_{i})_{i=1}^{\dim V_F}$ in one-to-one correspondence with the weight vectors of $ V_F^*$.

If $G$ is $\mathrm{Spin}(2N)$, then $\mathbf{I}_G$ is strongly generated by by the generators $(P_i)_{i=1}^{\rk\,\gf}$ of the Feigin--Frenkel centre and generators, $(\hat{g}_{{\rm s}i})_{i=1}^{\dim\,V_{\rm s}}$ and $(\hat{g}_{{\rm c}i})_{i=1}^{\dim\, V_{\rm c}}$ in one-to-one correspondence with the weight vectors of $ V_{\rm s}^*$ and $ V_{\rm c}^*$ respectively.
\end{conj}
%----%
Unlike the finite dimensional setting, we will not attempt a proof of these conjectures. This should potentially be dealt with in terms of the (micro-)localisation of $\mathbf{I}_G$ over its associated variety. Such a localisation theory is still in its early stages (see, \eg,~\cite{Arakawa:2015loc,Kuwabara:2021hyp}) but we hope that a geometric strategy will prove fruitful. For our purposes here, we will be content to consider examples for the low rank cases of ${\rm SL}_2$ and ${\rm SL}_3$, and leave more extensive analysis of examples to future investigation.

It is known~\cite{Arakawa:2018egx} that $\mathbf{I}_G$ is conformal---with critical central charge given by
%----%
\begin{equation}\label{eq:univ_cent_charge}
	c_{\mathbf{I}_G} = 2\rk\,\gf + (\rho,\rho^\vee)	~,
\end{equation}
%%%%%
where $\rho$ is the half sum of positive roots of $\gf$---\ie, the Weyl vector---and $\rho^\vee$ is the half sum of positive coroots. After imposing the $\rk\,\gf$ many constraints on~\eqref{eq:lattconf}, which fixed $\Delta_{\gamma_i} = 2$ we are left with $\rk\,\gf$ many unfixed background charges. These are fixed by requiring that the quadratic Feigin--Frenkel generator, $P_1$ is quasiprimary. The residue $T_{(2)}P_1$ is linear in the background charges and setting it to zero produces $2\rk\,\gf$ many linear equations for the background charges. However, these equations are not linearly independent. The expression for $P_1$ is unchanged under the symmetry $\partial \phi\mapsto -\partial \delta$ and $\partial \delta\mapsto \partial \phi$. Thus, the swaps $\alpha_i\mapsto - \beta_i$ and $\beta_i\mapsto - \alpha_i$ leaves the system of linear equations invariant. This symmetry cuts the system down to $\rk\,\gf$ many linear equations for $\rk\,\gf$ many unknowns. In the ${\rm SL}_N$ case it is schematically clear that these equations are independent, leading to a unique assignment of background charges. For other choices of $G$ we expect the same argument to hold, but we have not attempted a careful analysis. This assignment of background charges, in the ${\rm SL}_N$ case, fixes the central charge to agree with~\eqref{eq:univ_cent_charge} without additional input, and the general $G$ case should behave similarly.

The expression~\eqref{eq:lattconf}, with the correctly chosen values of $\alpha_i,\beta_i$, gives a free field realisation of the conformal vector of $\mathbf{I}_G$. On general grounds, we do not expect this to be a strong generator of $\mathbf{I}_G$, but rather it should arise as a composite of the other strong generators. We will verify this explicitly in the case $G=SL_2$ below.

%----------------------------------------------------%
\subsection{\label{subsec:sl2_ffr}Free field realisation for \texorpdfstring{$\mathbf{I}_{{\rm SL}_2}$}{I(SL2)}}
%----------------------------------------------------%

Let us work through the ${\rm SL}_2$ case as an illustrative example. The chiralisation of $\KT_G$ is $\Pi_{{\rm SL}_2}$, which is strongly generated by
%----%
\begin{equation}
	\gamma^{\pm\frac{1}{2}}\equiv e^{\pm \frac{1}{2}(\phi+\delta)}~,\quad b \equiv \frac{1}{2}\partial(\delta-\phi)~.
\end{equation}
%----%
We first find the expression for $S=\beta^{ch}(P_1)$, the generator of the Feigin--Frenkel centre. Our prescription above yields (\cf\ \cite[Eqn. (3.4)]{deBoer:1993iz}, with some change of conventions),
%----%
\begin{equation}
	S = bb+\gamma+\partial b~.
\end{equation}
%----%
We note here that rather than using this deformed Miura map, we could adopt a more bootstrap-inspired approach. The non-chiral image $\beta(P_1)$ is generated by $S = b^2 + \gamma$, and we should allow for quantum corrections as the self-OPE of $bb+\gamma$ is not regular. The most general corrected Ansatz is
%----%
\begin{equation}\label{eq:sl2-ff-generator}
	S = bb+\gamma + \alpha\partial b~,
\end{equation}
%----%
for some $\alpha\in \IC$. The normalisation is such that the Fourier mode $S_{(1)}$ acts via the generalised eigenvalue $\frac{1}{4}\lambda(2\lambda+1)$ on the Weyl module $\IV_{\lambda}$. Imposing regularity of the $S\times S$ OPE fixes $\alpha=1$ and so fixes the corrected $S$ to be precisely of the form returned by the deformed Miura map.

We expect only one additional strong generator, corresponding to the function $Y$ on $\Zf_{{\rm SL}_2}$. This operator appears in the simple pole in the $S\times X$ OPE, which is the chiral version of the extraction of the non-chiral $Y$ from the Poisson bracket $\{S,X\}$. Indeed, defining $Y$ as such, the OPEs close and we find a strongly generated vertex algebra with the following generators,
%----%
\begin{equation}
\begin{split}
	S &= bb+\gamma+\partial b~,\\
	X &= \gamma^{-\frac{1}{2}}~,\\
	Y &= -b\gamma^{-\frac{1}{2}}~.
\end{split}
\end{equation}
%----%
Comparing with \eqref{eq:sl2_finite_free_field_XY}, the expression for $Y$ (and $X$) in terms of $\Pi_G$ fields matches precisely the expression in terms of $\KT_G$ previously with no quantum corrections, but this is probably not particularly meaningful. Indeed, $Y$ defined here is not quasiprimary, but rather $Y+\frac34 X'$ is instead, and if we worked in terms of that quasiprimary instead there would indeed be quantum corrections.

The OPEs between these strong generators are given by
%----%
\begin{equation}\label{eq:univ_cen_sl2_OPE}
\begin{split}
	S(z)X(w) &\sim \frac{\frac{3}{4}X(w)}{(z-w)^2} + \frac{Y(w)}{z-w}~,\\
	S(z)Y(w) &\sim \frac{\frac{3}{4}Y(w)}{(z-w)^2} + \frac{(SX)(w)}{z-w}~,\\
	Y(z)X(w) &\sim \frac{\frac{1}{2}(XX)(w)}{z-w}~,\\
	Y(z)Y(w) &\sim \frac{-\frac{1}{4}(XX)(w)}{(z-w)^2}-\frac{\frac{1}{4}(X\partial X)(w)}{z-w}~.
\end{split}
\end{equation}
%----%
The determinantal relation $D_{{\rm SL}_2} = SX^2 - Y^2 \overset{!}{=} 1$ is chiralised here to the relation
%----%
\begin{equation}\label{eq:det_chiral_sl2}
	\hat{D}_{{\rm SL}_2} = SXX  - YY -\tfrac{3}{2} \partial Y X+ \tfrac{3}{2}Y\partial X- \tfrac{3}{4}\partial X \partial X - \tfrac{1}{8}X\partial^2 X\overset{!}{=} \id~.
\end{equation}
%----%
The generators $S$, $X$, and $Y$, together with the OPEs and relation above define a genuine strongly generated vertex algebra, which is a vertex subalgebra of $\Pi_{{\rm SL}_2}$.

As expected, the free-field conformal vector (with background charges $\alpha = \frac32$ and $\beta=-\frac12$) can be expressed as the following composite operator in terms of the strong generators,
%----%
\begin{equation}
\begin{split}
	T_0 = 2S(\partial Y X - Y\partial X) +3S\partial X \partial X+\tfrac{1}{2} Y\partial^3 X
	- \partial S YX +\tfrac{5}{2}\partial S \partial X X - 2 \partial Y \partial Y&\\
	+ \tfrac{3}{2}\partial Y \partial^2 X - \tfrac{3}{4}\partial^2 X \partial^2 X -\tfrac{3}{2} \partial^2 Y \partial X
	-\tfrac{7}{12} \partial^3 X \partial X -\tfrac{1}{2}\partial^3 Y X -\tfrac{1}{24}\partial^4 X X&~.
\end{split}
\end{equation}
%----%
The central charge is $c=26$, which agrees with~\eqref{eq:univ_cent_charge}.

The equivariant affine $\WW$-algebra for $\gf=\slf_2$ was presented in terms of strong generators and their OPEs in~\cite{Arakawa:2018egx}. By performing the Drinfel'd--Sokolov reduction, which is still tractable in this case, one can show that the the resulting vertex algebra is isomorphic to our realisation. We will not review this computation here.

%----------------------------------------------------%
\subsection{\label{subsec:interlude}Interlude---a more universal centraliser}
%----------------------------------------------------%

It is interesting to observe that the singular OPEs in \eqref{eq:univ_cen_sl2_OPE} define, abstractly, a strongly generated vertex algebra that we denote $\widetilde{\mathbf{I}}_{{\rm SL}_2}$ in which the vertex ideal generated by $\hat{D}_{{\rm SL}_2}-\id$ is not quotiented out. That is to say, the OPEs satisfy the Jacobi identities without removing any null states. The (now nonvanishing) operator $\hat{D}_{{\rm SL}_2}$ is not only null, but is actually central in $\widetilde{\mathbf{I}}_{{\rm SL}_2}$. Consequently, there is a one parameter family of analogous chiral universal centralisers given by the quotient of $\widetilde{\mathbf{I}}_{{\rm SL}_2}$ by the ideal $\hat{D}_{{\rm SL}_2} - \xi \id$ for $\xi\in \IC$. For $\xi\neq 0$, these quotients are all isomorphic to $\mathbf{I}_{{\rm SL}_2}$, with the isomorphism realised by rescaling the strong generators $X$ and $Y$. However, the centreless specialisation $\hat{D}_{{\rm SL}_2}=0$ is not isomorphic to $\mathbf{I}_{{\rm SL}_2}$. This centreless specialisation seems to be something of an oddity; in particular it is not a chiral quantisation of $\Zf_{{\rm SL}_2}$ since the determinant is set to zero.

The conformal vector $T_0$ does not lift to one in $\widetilde{\mathbf{I}}_{{\rm SL}_2}$. For example, the cubic pole in the $T_0 \times T_0$ OPE is proportional to $\partial\hat{D}_{{\rm SL}_2}$ (and the other poles are also polluted by the appearence of $\partial\hat{D}_{{\rm SL}_2}$ and its derivatives). The Virasoro OPE only holds in the quotient. Indeed, $\widetilde{\mathbf{I}}_{{\rm SL}_2}$ does not have a conformal vector. Any candidate conformal vector must be built out of the $S,X,Y$, and the OPEs amongst these generators are homogeneous in the number of $X$ and $Y$ fields appearing (\ie, there is an outer $U(1)$ symmetry of this vertex algebra under which $X$ and $Y$ have charge one and $S$ has charge zero). Charge conservation then prevents any candidate conformal vector from realising the Virasoro OPE.

This outer $U(1)$ action leads to a family of gradings on $\widetilde{\mathbf{I}}_{{\rm SL}_2}$---the OPEs are compatible with any grading that assign $|S| = 2$, $|X| = \Delta$ and $|Y|=\Delta+1$ for any $\Delta\in \IZ$. The simple quotient $\hat{D}_{{\rm SL}_2}- \xi \id$, for $\xi\neq 0$, violates this grading for all choices except $\Delta=-1$. The centreless specialisation, $\xi=0$, however, is compatible with this grading for any $\Delta\in\IZ$.

From a geometric point of view, this one parameter family arises from a closed embedding of ${\rm SL}_N$ into the ambient space of ${\rm GL}_N$. The cotangent bundle $T^\ast{\rm GL}_N\cong \mathfrak{gl}_N^*\times {\rm GL}_N$ has a splitting
%----%
\begin{equation}\label{eq:gln_split}
	T^*{\rm GL}_N \cong \mathfrak{gl}_N^*\times {\rm GL}_N\cong (\slf_N^*\times \mathfrak{gl}^*_1)\times {\rm GL}_1 \times {\rm SL}_N\cong (T^*\IC^\times)\times (T^*{\rm SL}_N)~,
\end{equation}
%----%
arising from the splitting $\mathfrak{gl}_N\cong \mathfrak{gl}_1 \oplus \slf_N$ and ${\rm GL}_N\cong {\rm SL}_N \rtimes {\rm GL}_1$. Consider the Poisson subalgebra of $\IC[T^*{\rm GL}_N]$ generated by the $\IC[\slf_N^*]\otimes \IC[{\rm GL}_N]$ factors in the splitting. We denote its spectrum by $\widetilde{T}^*{\rm SL}_N$. As varieties, we have that $\widetilde{T}^*{\rm SL}_N \cong {{\rm GL}_1} \times  T^*{\rm SL}_N\cong \mathfrak{sl}_N^*\times {\rm GL}_N$.

This thickened space inherits Hamiltonian left and right action of ${\rm SL}_N$, and we can perform the corresponding two-sided Kostant--Whittaker reduction of this ambient space. The $ {\rm GL}_1$ factor of the splitting is $N_L\times N_R$ invariant, and lies in the (intersection) $\overline{\mu}_L^{-1}(\chi)\cap\overline{\mu}_R^{-1}(\chi)$. Therefore, the two-sided reduction will be isomorphic to $\IC^\times\times \Zf_{{\rm SL}_N}$, and we denote this space by $\widetilde{\Zf}_{{\rm SL}_N}$. It is naturally a scheme over $S_f$ by extending the structure morphism of $\Zf_{{\rm SL}_N}$. The cotangent bundle $T^*{\rm SL}_N$ is a closed, Poisson subvariety of $\widetilde{T}^* {\rm SL}_N$, so by functoriality of reduction, $\Zf_{{\rm SL}_N}$ is a closed, Poisson, subvariety of $\widetilde{\Zf}_{{\rm SL}_N}$.

There exist a family of $\IC^\times$ actions on $T^*{\rm GL}_N$, which acts as uniform rescaling of the vector space $\mathfrak{gl}_N^*$ and as left multiplication by a constant diagonal matrix on ${\rm GL}_N$. Under these $\IC^\times$ actions, the linear functions on $\mathfrak{gl}_N^*$ can be assigned unit weight, in which case the Poisson bracket has degree $-1$, while the weights of the matrix entries, $g_{ij}$, which generate $\IC[{\rm GL}_N]$ are unfixed. Any assignment $|g_{ij}|= \Delta$ for all $i,j=1,\dots,N$ is compatible with the Poisson structure. The fundamental invariants of $\slf_N$, $(P_{i})_{i=1}^{N}$ will have weight $|P_i| = i+1$ for $i=1,\dots, N$,~\ie, the weight is equal to their degree. The thickened space $\widetilde{T}^*{\rm SL}_N$ inherits this grading.

Similarly, the fibres of $\widetilde{\Zf}_{{\rm SL}_N}$ will inherit the grading on $\widetilde{T}^*{\rm SL}_N$. Under Kostant--Whittaker reduction the $\IC^\times$ weights are shifted by the weight associated to the principal $\slf_2$-embedding. Consider, specifically, the bottow row $g_{Nj}$ whose weights under the $\IC^\times$ action must satisfy $|g_{Nj}| = |g_{N1}| + (j-1)$.  Denoting $g_{N1}$ by $X$, as is our convention, and denoting its weight by $|X|= \Delta$, any grading where $|g_{Nj}| = \Delta + (j-1)$ is compatible with the Poisson bracket.

Picking out a closed subvariety by imposing the determinantal relation $D_{{\rm GL}_N}\overset{!}{=}1$ is incompatible with this grading---unless $\Delta = 1-N$. Note that any such relation $D_{{\rm GL}_N}\overset{!}{=} \xi$, for $\xi\in \IC^\times$ forces $\Delta = 1-N$. The closed subvariety satisfying $D_{{\rm GL}_N}=1$ gives a closed embedding $\Zf_{{\rm SL}_N}\hookrightarrow \widetilde{\Zf}_{{\rm SL}_N}$. In fact, any subvariety $D_{{\rm GL}_N} = \xi$ gives a closed embedding of $\Zf_{\rm SL_N}$ inside $\widetilde{\Zf}_{{\rm SL}_N}$ as schemes over $S_f$.

The thickened space, $\widetilde{T}^*{\rm SL}_N$ has a closure isomorphic to $\IC \times {\rm SL}_N$. Performing Kostant--Whittaker reduction on this space gives a closure of $\widetilde{\Zf}_{{\rm SL}_N}$. In this closure, we have an additional subvariety with vanishing ideal $D_{{\rm GL}_N}$. This subvariety of determinant zero matrices is not isomorphic to $\Zf_{{\rm SL}_N}$ and seems fairly exotic.

The geometric construction above reflects many of the properties we have seen in $\widetilde{\mathbf I}_{{\rm SL}_2}$, namely the auxiliary grading and one-parameter family of quotients. Let us present a schematic chiral version of this construction that should explain the existence of this more general vertex algebra.

Let $\D_{{\rm GL}_N,-N}$ be CDOs on ${\rm GL}_N$ at level $-N$; these admit a splitting  $\D_{{\rm GL}_N,-N} \cong \D(\IC^\times)\otimes \D_{{\rm SL}_N,-N}$. This is the chiral version of the splitting \eqref{eq:gln_split}. Functions on $T^*\IC^\times$ are generated by $D_{{\rm GL}_N}^{\pm}$ and the fibre coordinate $p$, with Poisson bracket $\{p, D_{{\rm GL}_N}^\pm\} = \pm D_{{\rm GL}_N}^\pm$. Abusing notation, the CDOs on $\IC^\times$ are strongly generated by $D_{{\rm GL}_N}^\pm$ and $p$ with OPE
%----%
\begin{equation}
	p(z)D_{{\rm GL}_N}(w)^\pm \sim \frac{D_{{\rm GL}_N}^\pm(w)}{z-w}~.
\end{equation}
%----%
Let $\overline{\D}_{SL_N,-N}\subset \D_{{\rm GL}_N,-N}$ be the vertex subalgebra corresponding to $\mathrm{ker}~p_{(0)}$. It is strongly generated by the strong generators of $\D_{{\rm SL}_N,-N}$ and by $D_{{\rm GL}_N}$. To recover $\D_{{\rm SL}_N,-N}$ we quotient out by the vertex ideal generated by $D_{{\rm GL}_N} -\id$. Under Drinfel'd--Sokolov reduction, the $\D_{{\rm SL}_N,-N}$ subalgebra reduces to $\mathbf{I}_{{\rm SL}_N}$. Meanwhile, the generator $D_{{\rm GL}_N}$ gets corrected to $\hat{D}_{{\rm GL}_N}$, which is the composite expression for the determinant inside $\mathbf{I}_{{\rm SL}_N}$.

Thus, the Drinfel'd--Sokolov reduction of $\overline{D}_{{\rm SL}_N,-N}$ yields $\mathbf{I}_{{\rm SL}_N}\otimes\langle \hat{D}_{{\rm GL}_N}\rangle$, where $\langle \hat{D}_{{\rm GL}_N}\rangle$ is a commutative vertex algebra generated by $\hat{D}_{{\rm GL}_N}$ and its inverse. The reduction has a one-parameter family of quotients by setting $\hat{D}_{{\rm GL}_N}=\xi\id$ for $\xi\in\IC^\times$. Each of these quotients is isomorphic to $\mathbf{I}_{{\rm SL}_N}$, and we can identify this reduction as the more universal centraliser $\widetilde{\mathbf I}_{{\rm SL}_N}$.

%----------------------------------------------------%
\subsection{\label{subsec:sl3_ffr}Free field realisation for \texorpdfstring{$\mathbf{I}_{{\rm SL}_3}$}{I(SL2)}}
%----------------------------------------------------%

We return to another example, this time in rakn two. Once again, our starting point is to find expressions for the Feigin--Frenkel generators $S$ and $W$. Like in the ${\rm SL}_2$ case, we can utilise our deformed Miura approach to produce the right expressions. For $G={\rm SL}_3$, the restrictions of the generators of the Harish-Chandra centre to $S_f$ are given, up to a convention for overall scaling, by
%----%
\begin{equation}
\label{eq:SW_classical}
\begin{split}
	S &=\frac{1}{2}( b_1^2 +b_2^2 -b_1b_2 + \gamma_1 +\gamma_2)~,\\
	W &=\frac{1}{3}( b_1b_2(b_1-b_2)+\gamma_1b_2-\gamma_2b_1~)~,
\end{split}
\end{equation}
%----%
and by performing the tic-tac-toe procedure as in~\cite[Eqn. (3.10)]{deBoer:1993iz}, we arrive at chiral versions given by
%----%
\begin{equation}
\begin{split}
	S &= \frac{1}{2}(b_1^2 + b_2^2 - b_1b_2 +\gamma_1 + \gamma_2) + \frac{1}{2} \partial( b_1+ b_2)~,\\
	W &= \frac{1}{3}\left(b_1b_2(b_1-b_2) + b_2\gamma_1  -b_1\gamma_2\right) + \frac{1}{6}\partial(\gamma_1-\gamma_2)  +\frac{1}{3} (b_1\partial b_1 -b_2 \partial b_2)\\
	&\phantom{=} +\frac{1}{6} (b_2\partial b_1  - b_1\partial b_2) +\frac{1}{6}\partial^2(b_1 - b_2)~.
\end{split}
\end{equation}
%----%
Alternatively, we can start with an Ansatz based on the finite-dimensional expression and add quantum corrections while demanding regularity of the OPEs. The appropriate Ans\"atze are
%----%
\begin{equation}
\label{eq:SW_chiral_ansatz}
\begin{split}
	S &= \frac{1}{2}(b_1^2 + b_2^2 - b_1b_2 +\gamma_1 + \gamma_2) + \alpha_1 (\partial b_1+\partial b_2)~,\\
	W &= \frac{1}{3}(b_1b_2(b_1-b_2) + b_2\gamma_1  -b_1\gamma_2) +  \alpha_2\partial(\gamma_1-\gamma_2) + \alpha_3 (b_1\partial b_1 -b_2 \partial b_2)\\
	&\phantom{=} +\alpha_4 (b_2\partial b_1  - b_1\partial b_2) +\alpha_5\partial^2 (b_1 - b_2)~,
\end{split}
\end{equation}
%----%
where $\alpha_i\in\IC$. The corrections are constrained to satisfy the requirement that $S$ be invariant under the outer-automorphism $b_1 \mapsto b_2$, $e_1\mapsto e_2$ while $W$ is negated. Requiring regularity of the $S\times S$ OPE fixes the parameter $\alpha_1=1/2$, and regularity of $W\times W$ and $W\times S$ fixes the remaining coefficients as follows,
%----%
\begin{equation}
	\alpha_2 =\frac{1}{6} ~,\qquad \alpha_3 = \frac{1}{3}~,\qquad \alpha_4 = \frac{1}{6}~,\qquad \alpha_5 = \frac{1}{6}~,
\end{equation}
%----%
which agrees exactly with the expressions arising from $\beta^{ch}$.

To complete the construction, we start with the strong generator $X=-\gamma^{-2/3}_1\gamma^{-1/3}_2$, and identify additional strong generators $Y$ and $Z$ by analysing the simple poles in the $S\times X$ and $S\times Y$ OPEs, respectively (here we again define $Y$ and $Z$ to not be quasi-primaries, which could be rectified by adding descendants of other generators). The complete set of strong generators then takes the form
%----%
\begin{equation}
\label{eq:sl3_chiral_strong_generators}
\begin{split}
	S &= \frac{1}{2}(b_1^2 + b_2^2 - b_1b_2 +\gamma_1 + \gamma_2) + \frac{1}{2} \partial( b_1+ b_2)~,\\
	W &= \frac{1}{3}\left(b_1b_2(b_1-b_2) + b_2\gamma_1  -b_1\gamma_2\right) + \frac{1}{6}\partial(\gamma_1-\gamma_2)  +\frac{1}{3} (b_1\partial b_1 -b_2 \partial b_2)\\
		&\phantom{=} +\frac{1}{6} (b_2\partial b_1  - b_1\partial b_2) +\frac{1}{6}\partial^2(b_1 - b_2)~,\\
	X &=-\gamma^{-2/3}_1\gamma^{-1/3}_2~,\\
	Y &= -\frac{1}{2}b_1(\gamma^{-2/3}_1\gamma^{-1/3}_2)~,\\
	Z &= -\frac{1}{2}\big( \gamma_1 - \gamma_2 + b_1b_1- b_2b_2+b_1b_2-\partial b_2\big)\gamma_1^{-2/3}\gamma_2^{-1/3}~.
\end{split}
\end{equation}
%----%
The OPEs between the Feigin--Frenkel generators and the others are given by
%----%
\begin{equation}
\begin{split}
	S(z)X(w) &\sim \frac{\frac{4}{3}X}{(z-w)^2}+\frac{-Y}{z-w}~,\\
	S(z)Y(w) &\sim \frac{\frac{4}{3}Y}{(z-w)^2}+\frac{-Z-\frac{1}{2}(SX)}{z-w}~,\\
	S(z)Z(w) &\sim \frac{\frac{4}{3}Z}{(z-w)^2} +\frac{-\frac{1}{2}(SY)-(WX)+\frac{1}{6}(\partial SX)}{z-w}~,\\
	W(z)X(w) &\sim \frac{-\frac{20}{27}X}{(z-w)^3} + \frac{\frac{5}{6}Y}{(z-w)^2}+\frac{-Z+\frac{1}{6}(SX)}{z-w}~,\\
	W(z)Y(w) &\sim \frac{-\frac{20}{27}Y}{(z-w)^3} +\frac{\frac{5}{6}Z+\frac{5}{12}(SX)}{(z-w)^2}+ \frac{-(WX)-\frac{1}{3}(SY)+\frac{1}{6}(\partial SX)}{z-w}~,\\
	W(z)Z(w) &\sim \frac{-\frac{20}{27}Z}{(z-w)^3} +\frac{\frac{5}{6}(WX)+\frac{5}{12}(SY)-\frac{5}{36}(\partial SX)}{(z-w)^2}\\
	&+ \frac{-(WY) + \frac{1}{6}(SZ)+\frac{1}{6}(\partial SY)-\frac{1}{4}(SSX)}{z-w}~,
\end{split}
\end{equation}
%----%
where we have suppressed the dependence of the residues on the coordinate $w$ for clarity. Similarly, for the OPEs amongst $X$, $Y$, and $Z$, we have
%----%
\begingroup
\allowdisplaybreaks
\begin{align}
	Y(z)X(w) &\sim \frac{-\frac{2}{3}(XX)}{z-w}~,\\
	Y(z)Y(w) &\sim  \frac{-\frac{4}{9}(XX)}{(z-w)^2} + \frac{-\frac{4}{9}(X\partial X)}{z-w}~,\\
	Z(z)X(w) &\sim \frac{\frac{1}{9}(XX)}{(z-w)^2} + \frac{-\frac{5}{6}(YX) - \frac{4}{9} (X\partial X)}{z-w}~,\\
	Z(z)Y(w) &\sim \frac{\frac{2}{27}(XX)}{(z-w)^3} + \frac{-\frac{4}{9}(YX) - \frac{4}{18}(X\partial X)}{(z-w)^2} \\
	& + \frac{-\frac{1}{3}(YY) +\frac{1}{6}(ZX)-\frac{1}{4}(SXX) - \frac{4}{9}(\partial Y X) +\frac{5}{36} (Y\partial X) -\frac{4}{27} (X\partial^2 X) +\frac{5}{54} (\partial X \partial X)}{z-w}~,\\
	Z(z)Z(w) &\sim  \frac{\frac{1}{81}(XX)}{(z-w)^4}  +\frac{\frac{1}{81}(X\partial X)}{(z-w)^3} \\
	& +\frac{-\frac{13}{36}(YY) -\frac{1}{9}(ZX) -\frac{1}{6}(SXX) -\frac{1}{9}(\partial Y X) +\frac{1}{54}(Y \partial X) -\frac{8}{81} (X\partial^2 X)  + \frac{1}{81} (\partial X \partial X)}{(z-w)^2}\\
	&-\frac{\frac{1}{18}(\partial (Z X)
	 + \partial(\partial Y X) -\frac{1}{6} \partial (Y \partial X) + \frac{3}{2} \partial (S X X) +\frac{2}{3} X\partial^3 X +\frac{13}{2} Y\partial Y)}{z-w}~.
\end{align}
\endgroup
%----%
The determinantal null ideal is chiralised to
%----%
\begin{equation}
\begin{split}
	& D_{{\rm SL}_3} = -\tfrac{1}{8} SSS XXX + WW XXX + \tfrac{1}{2}  W SY XX + \tfrac{1}{2} SSYYX + W YYY - \tfrac{1}{4} SSZX  - 3 W ZYX \\
	&- SZ YY + \tfrac{1}{2} SZZ X + ZZZ+\tfrac{9}{8}SSYX'X + \tfrac{83}{144}SSX'X'X - \tfrac{11}{8} SSY'XX - \tfrac{1}{32}SSX''XX\\
	&- \tfrac{107}{144} SYYX''- \tfrac{1091}{432}SYX''X' - \tfrac{29}{108} SYX'''X - \tfrac{23}{12}SZYX' - \tfrac{71}{72} SZX'X' + \tfrac{5}{12}SZY'X\\
	&+ \tfrac{19}{18} SZX''X - \tfrac{2}{3}SY'YY - \tfrac{19}{36} SY'YX' + \tfrac{257}{216}SY'X'X' - \tfrac{41}{24} SY'Y'X+ \tfrac{511}{216}SY'X''X + \tfrac{8}{3} SZ'YX \\
	&+ \tfrac{101}{36}SZ'X'X - \tfrac{649}{432} SX''X'X' + \tfrac{5765}{3456}SX''X''X+ \tfrac{19}{16} SY''YX + \tfrac{11}{6}SY''X'X - \tfrac{161}{144} SZ''XX \\
	& + \tfrac{1651}{2592}SX'''X'X + \tfrac{131}{648} SY'''XX + \tfrac{14053}{62208}SX^{(4)}XX- \tfrac{1}{12} WSX'XX + \tfrac{7}{2}WYYX'+ \tfrac{64}{9} WYX'X'  \\
	&+ \tfrac{7}{36}WYX''X - \tfrac{13}{2} WZX'X - \tfrac{1}{2}WS'XXX+ \tfrac{199}{27} WX'X'X' - 7WY'YX - \tfrac{263}{18} WY'X'X \\
	&+ 5WZ'XX - \tfrac{137}{72} WX''X'X + \tfrac{163}{72}WY''XX - \tfrac{19}{24} WX'''XX + \tfrac{7}{576}YYX^{(4)} - \tfrac{52799}{46656}YX'''X''\\
	& - \tfrac{61225}{186624}YX^{(4)}X' - \tfrac{98117}{933120}YX^{(5)}X - \tfrac{71}{216}ZYX'''+ 4ZZY' + \tfrac{71}{36} ZZX'' + \tfrac{52}{9}ZY'Y' + \tfrac{667}{108}ZY'X'' \\
	&+ \tfrac{3557}{2592}ZX''X'' - \tfrac{25}{72}ZY''Y + \tfrac{59}{432}ZY''X' + \tfrac{1493}{2592}ZX'''X' + \tfrac{203}{216}ZY'''X + \tfrac{12823}{31104}ZX^{(4)}X  \\
	&- \tfrac{5}{8}S'SYXX - \tfrac{31}{24} S'SX'XX + \tfrac{1}{2}S'YYY + \tfrac{11}{8} S'YYX' + \tfrac{37}{18}S'YX'X' + \tfrac{137}{144} S'YX''X - \tfrac{3}{4}S'ZYX \\
	&+ \tfrac{25}{12} S'ZX'X - \tfrac{3}{16}S'S'XXX + \tfrac{1205}{864} S'X'X'X' - \tfrac{9}{4}S'Y'YX + \tfrac{37}{9} S'Y'X'X - \tfrac{5}{24}S'Z'XX \\
	&+ \tfrac{1121}{144} S'X''X'X + \tfrac{23}{288}S'Y''XX + \tfrac{4261}{5184}S'X'''XX - \tfrac{1}{4}W'SXXX + W'YX'X - \tfrac{3}{2} W'ZXX \\
	&+ \tfrac{41}{12}W'X'X'X - \tfrac{7}{2} W'Y'XX - \tfrac{3}{4}W'X''XX - \tfrac{547}{1296}Y'YX''' + \tfrac{80}{27}Y'Y'Y' + \tfrac{737}{144} Y'Y'X'' \\
	&+ \tfrac{63505}{31104}Y'X''X'' + \tfrac{60653}{46656}Y'X'''X' + \tfrac{128695}{186624}Y'X^{(4)}X - \tfrac{323}{108}Z'YX'' - 4Z'ZY - \tfrac{46}{9} Z'ZX'\\
	& - \tfrac{40}{9}Z'Y'Y - \tfrac{140}{27}Z'Y'X' + \tfrac{40}{9} Z'Z'X - \tfrac{6881}{864}Z'X''X' + \tfrac{197}{48}Z'Y''X - \tfrac{9575}{7776}Z'X'''X + \tfrac{1}{6}S''SXXX\\
	&- \tfrac{3}{16} S''YYX  - \tfrac{2}{3}S''YX'X + \tfrac{49}{48} S''ZXX - \tfrac{605}{1728}S''X'X'X + \tfrac{39}{16} S''Y'XX + \tfrac{6841}{3456}S''X''XX\\
	&+ \tfrac{1}{4} W''YXX + \tfrac{11}{12}W''X'XX - \tfrac{8747}{15552}X''X''X'' - \tfrac{323}{432}Y''YX''
	+ \tfrac{7}{36} Y''Y'Y + \tfrac{4571}{2592}Y''Y'X' \\
	&- \tfrac{913}{324}Y''X''X' + \tfrac{1631}{1296}Y''Y''X - \tfrac{52315}{93312}Y''X'''X+ \tfrac{14}{9} Z''YY + \tfrac{817}{216}Z''YX' - \tfrac{77}{36}Z''ZX\\
	&+ \tfrac{10307}{2592}Z''X'X' - \tfrac{97}{18}Z''Y'X - \tfrac{2255}{5184}Z''X''X - \tfrac{1}{72} S'''YXX + \tfrac{395}{1296}S'''X'XX + \tfrac{1}{72} W'''XXX \\
	&- \tfrac{63841}{23328}X'''X''X'- \tfrac{90493}{209952}X'''X'''X - \tfrac{13}{108}Y'''YY - \tfrac{3317}{3888}Y'''YX' - \tfrac{2327}{1728}Y'''X'X' + \tfrac{6929}{3888}Y'''Y'X \\
	&+ \tfrac{1999}{5832}Y'''X''X - \tfrac{1}{24}Z'''YX + \tfrac{1841}{3888}Z'''X'X + \tfrac{463}{10368}S^{(4)}XXX - \tfrac{72935}{139968}X^{(4)}X'X' \\
	&+ \tfrac{9155}{93312}X^{(4)}X''X - \tfrac{865}{15552}Y^{(4)}YX + \tfrac{9443}{186624}Y^{(4)}X'X + \tfrac{1159}{15552} Z^{(4)}XX + \tfrac{2243}{349920}X^{(5)}X'X \\
	&+ \tfrac{244}{3645}Y^{(5)}XX + \tfrac{2237}{51840}X^{(6)}XX~.
\end{split}
\end{equation}
%----%
There is, currently, no expression in the literature to which we can compare our presentation of $\mathbf{I}_{SL_3}$. Nevertheless, our Ansatz has the properties one would expect from computing a schematic two-sided Drinfel'd---Sokolov reduction of $\D_{SL_3,-3}$.

%----------------------------------------------------------------------%
\section{\label{sec:FF_genus_zero}Class \texorpdfstring{$\SS$}{S} free field realisations at genus zero}
%----------------------------------------------------------------------%

We now move on to the construction of free field realisations of chiral algebras of class $\SS$ in the $\mathfrak{a}_1$ case. In~\cite{Beem:2019tfp}, the authors described (amongst other things) a free field realisation of the vertex algebra associated to the four-punctured sphere, $\V_{\af_1,4}\cong L_{-2}(\mathfrak{d}_4)$, which was inspired by the geometry of the associated variety/Higgs branch (in this case, the closure of the minimal nilpotent orbit in $\mathfrak{d}_4$). This amounted to a free field realisation of $\V_{\af_1,4}$ inside $\V_{{\af_1},3}\otimes \Pi_{{\rm SL}_2}$, with $\V_{{\af_1},3}$ being just a system of symplectic bosons.

More generally, in this section we find evidence for a general inductive free field realisation where the number of punctures can be increased at the expense of introducing extra lattice degrees of freedom,\!\footnote{For a selection of values of $s$, such relations were observed previously by Carlo Meneghelli \cite{CMUnpublished}.}
%----%
\begin{equation}
	\V_{\af_1,s} \subset \V_{\af_1,s-1}\otimes \Pi_{{\rm SL}_2}~.
\end{equation}
%----%
At the level of Higgs branches/associated varieties, this relation arises from the existence of a Zariski open $U_n \subset X_n$ such that passing to an \'etale $2$-cover $\widetilde{U}_n$ gives
%----%
\begin{equation}
	\IC[\widetilde{U}_n] \cong \IC[X_{s-1}] \otimes \IC[T^*\IC^\times]~.
\end{equation}
%----%
The realisations $\V_{{\rm SL}_2,s}\hookrightarrow \V_{{\rm SL}_2,s-1}\otimes \Pi_{{\rm SL}_2}$ should then be thought of as being a consequence of localising along $U_n\hookrightarrow X_{s}$ and passing to an appropriate cover. Here we will not endeavour to make these ideas more precise in the abstract, but instead proceed by example. Furthermore, we will skip directly to the chiral case, from which the corresponding relations of finite-dimensional symplectic varieties can be recovered by passing to the (reduced) $C_2$-algebra $(R_V)_{red}$.

In the rest of this section, we will construct genus zero chiral algebras of class $\SS$ of type $\af_1$ using these free field methods. In principle such a program could exploit the fact that in type $\af_1$ the trinion theory is a theory of free hypermultiplets, with associated vertex operator algebra given by symplectic bosons, and so use $\V_{{\rm SL}_2,3}$ as a starting point as in \cite{Beem:2019tfp}. Such an approach obscures some symmetry of the problem---indeed, by generalised $S$-duality all (full) punctures should be treated equally in the associated vertex algebra, while starting with the trinion VOA privileges the corresponding three ``initial'' punctures. We instead take the chiral universal centraliser as our starting point, or more precisely its free field realisation from Section \ref{subsec:sl2_ffr}. We will therefore present these vertex algebras in free field form in a manner that renders generalised $S$-duality completely manifest.

%----------------------------------------------------%
\subsection{\label{subsec:genus_zero_A1}General structure of genus zero free field realisations}
%----------------------------------------------------%

We will realise $\V_{\af_1,s}$ inside $\bigotimes_{i=0}^s \Pi^{i}_{{\rm SL}_2}$, where the superscript $i$ labels the puncture associated to that lattice factor when $i\neq0$, and the zeroth copy will be used to construct $\mathbf{I}_{{\rm SL}_2}$. We can think of $i$th copy for $i>0$ being used to construct $\V_{\af_1,i}$ in an inductive process, though at the end of the day all punctures will be treated symmetrically. For $i\neq0$ we denote the weight $\pm \frac{1}{2}$ generators by $e_{i}^{\pm\frac{1}{2}}$ and the weight one generator by $p_i$. For the generators of $\Pi^0_{{\rm SL}_2}$ that are used in the free field realisation of the sphere VOA, we continue to use $\{\gamma^{\pm\frac{1}{2}},b\}$ as in Section \ref{subsec:sl2_ffr}.

For each $i>0$, we can realise a critical-level affine $\slf_2$ subalgebra in $\Pi^i_{{\rm SL}_2}$, generated by the currents
%----%
\begin{equation}
\label{eq:sl2_ffr_specialised}
	e_i = \left(e_i^{\frac{1}{2}}\right)^2~, \quad
	h_i =-2\partial\phi_i~,\quad
	f_i = (\partial^2\delta_i-(\partial\delta_i)^2)e_i^{-1}~.
\end{equation}
%----%
These currents generate the (simple quotient of the) universal current algebra $V^{\kappa_c}(\slf_2)$,~\ie, the quotient of $V^{\kappa_c}(\gf)$ where all elements of the Feigin--Frenkel centre are set to zero. The universal $V^{\kappa_c}(\slf_2)$ is recovered by reintroducing a commutative vertex algebra with generator $S$ that also commutes all of $\Pi^i_{{\rm SL}_2}$. This plays the role of the strong generator of the Feigin--Frenkel centre, and is included in the definition of the lowering operator,
%----%
\begin{equation}
\label{eq:sl2_ffr_with_centre}
	f_i = (S+\partial^2\delta_i-(\partial\delta_i)^2)e_i^{-1}~.
\end{equation}
%----%
In chiral algebras of class $\SS$, the Feigin--Frenkel generators for the current algebras at each puncture are identified \cite{Arakawa:2018egx,Beem:2022mde}, and this is accomplished in the free field realisation by letting the same $S$ appear in each instance of $V^{\kappa_c}(\slf_2)$. Indeed, we further identify this with the Feigin--Frenkel generator of $\mathbf{I}_{{\rm SL}_2}\subset \Pi^0_{{\rm SL}_2}$, which gives a realisation of $V^{\kappa_c}(\slf_2)$ as a subalgebra of $\mathbf{I}_G\otimes \Pi^i_{{\rm SL}_2}$.

A conjectured accounting of strong generators for $\V_{{\rm SL}_2,s}$ has appeared previously in the literature (see~\cite[Conjectures 1 \& 2]{Beem:2014rza}). In particular, there are expected to be $s$-many sets of $\slf_2$ currents, the Virasoro stress tensor, and generators $X^{a_1a_2\dots a_s}$ which transform in the $s$-fold fundamental representation $\mathbf{2}^{\otimes s}$ of the zero-mode algebra of the $\slf_2$ currents. Our approach will be to identify free-field candidates for these generators, which can then be checked to strongly generate closed vertex subalgebras.

We start of by constructing the $\slf_2$ currents $(e_i,h_i,f_i)$ for each puncture $i=1,\dots,s$ as above. Next, we identify the highest-weight of the $s$-fundamental operator,
%----%
\begin{equation}
	X^{\underbrace{++\dots++}_{s}} \coloneqq \gamma^{-\frac{1}{2}}\prod_{i=1}^{s} e^{\frac{1}{2}}_i~,
\end{equation}
%----%
The remaining components $X^{a_1a_2\dots a_s}$ of the $s$-fundamental representation are realised by acting with the $f_{i,(0)}$ zero modes. In general, we have
%----%
\begin{equation}
	X^{a_1a_2\dots a_s} = \left(\prod_{\{i\,:\, a_i=-\}} f_{i,(0)}\right) \, X^{++\dots+}~,
\end{equation}
%----%
For $s\neq 3$ there are expected to be null relations amongst the composites of the putative strong generators. For small $s$, we will explicitly identify these null ideals, but in general their structure can be quite intricate. However, the free field realisation constructs a quotiented vertex algebra, so some (conjecturally all) of the the null relations will be automatically satisfied in our realisations of the $\V_{{\rm SL}_2,s}$. Finding the exact form of the null relations in the free field realisation amounts to a book-keeping exercise.

The vertex algebras $\V_{{\rm SL}_2,s}$ are conformal for all $s\in \IN$ with central charge given by
%----%
\begin{equation}
	c_s = (s-(2s-4)h^\vee)\dim~\gf - (s-2)\rk~\gf = 26 -10s~.
\end{equation}
%----%
We previously identified the conformal vector of $\V_{{\rm SL}_2,0}\cong \mathbf{I}_{{\rm SL}_2}$ as a canonical free field conformal vector with appropriate background charges chosen. This continues to be the case after introducing punctures, and for $\V_{{\rm SL}_2,s}$ we define a conformal vector
%----%
\begin{equation}
	T_s = \sum_{i=0}^{s} \bigg( \partial\delta_i\partial\delta_i -\partial \phi_i \partial \phi_i -\alpha_i\partial^2 \phi_i -\beta_i \partial^2 \delta\bigg)~.
\end{equation}
%----%
Demanding that $\Delta_{e_i} = \Delta_p = 1$ and fixing the central charge forces the background charges to $\alpha_{i>0} =0$, and $\beta_{i>0} = 1$. The background charges $\alpha_0$ and $\beta_0$ have been set to $\frac32$ and $-\frac12$ respectively from our arguments in Section~\ref{subsec:chiral_morphism}.

This conformal vector will generally be included as an additional strong generator for $\V_{{\rm SL}_2,s}$. For the special cases $s=0,1,2,3,4$, the conformal vector is actually redundant, and can be written as a composite of the remaining strong generators, so in these cases there will be a null relation that allows $T_s$ to be identified with this composite. In our free field realisations, this relation will be automatically satisfied---this was shown above for the case of $s=0$.

To summarise, we present the following conjectures.
%----%
\begin{conj}
	The vertex algebra $\V_{{\rm SL}_2,s}$, for $s\geqslant 1$, is isomorphic to the subalgebra of $\Pi_{{\rm SL}_2}^{\otimes s+1}$ that is weakly generated by the moment maps $e_i,h_i,f_i$ for $i=1,\dots,s$, the highest weight state $X^{+++\dots+}$. This vertex subalgebra will be simple, with all null relations in its strongly-generated presentation satisfied identically.
\end{conj}
%----%
\begin{conj}
	The subalgebra of $\Pi^{\otimes s+1}_{{\rm SL}_2}$ that is weakly generated by $e_i,h_i,f_i$ for $i=1,\dots s$ and $X^{+++\dots +}$ is strongly generated by
	%----%
	\begin{itemize}
		\item $e_i,h_i,f_i$, for $i=1,\dots,s$.
		\item $X^{a_1a_2\dots a_s}$, for $a_i = +,-$.
		\item $T_s$, the conformal vector.
	\end{itemize}
	%----%
\end{conj}
%----%
Together, these conjectures imply an algorithmic approach to constructing the strong generators---matching the expectation of~\cite{Beem:2014rza}---in a free-field realisation of $\V_{{\rm SL}_2,s}$ inside $\Pi^{\otimes s+1}_{{\rm SL}_2}$. For $1\leqslant s\leqslant4$, we will explicitly verify these statements.

A corollary of these conjectures is that the free field realisations can be performed recursively. Given $\V_{{\rm SL}_2,s}$ as a strongly generated vertex algebra,~\ie\ not as a free field realisation, one can construct a realisation of $\V_{{\rm SL}_2,s+1}$, inside $\V_{{\rm SL}_2,s}\otimes \Pi_{{\rm SL}_2}$. The vertex algebra $\V_{{\rm SL}_2,s}$ comes with affine currents $e_i,h_i,f_i$ for $i=1,\ldots,s$; let $S$ denote the generator of their common Feigin--Frenkel centre. We can then construct $e_{s+1},h_{s+1},f_{s+1}$ inside $\V_{{\rm SL}_2,s}\otimes \Pi_{{\rm SL}_2}$ using that common $S$. Then taking $X^{a_1a_2\dots a_s}$ to be the state in $\V_{{\rm SL}_2,s}$ with weight $\Delta = \frac{1}{2}(s-2)$, which transforms in the $\mathbf{2}^{\otimes s}$ under the $\slf_2$ actions provided by $e_{i,(0)},h_{i,(0)},f_{i,(0)}$, we can define
%----%
\begin{equation}
	X^{a_1a_2\dots a_s+} \coloneqq X^{a_1a_2\dots a_s}e_{s+1}^{\frac{1}{2}}~, X^{a_1a_2\dots a_s -}\coloneqq f_{s+1,(0)}X^{a_1a_2\dots a_s +}~.
\end{equation}
%----%
Letting $T_s$ be the unique conformal vector of $\V_{{\rm SL}_2,s}$ with central charge $c_s$, then we can define the conformal vector
%----%
\begin{equation}
	T_{s+1} \coloneqq T_s + \partial\delta_{s+1}\partial\delta_{s+1} - \partial \phi_{s+1}\partial \phi_{s+1} - \partial ^{2} \delta_{s+1}~,
\end{equation}
%----%
with central charge $c_{s+1}$. The vertex algebra $\V_{{\rm SL}_2,s+1}$ is isomorphic to the subalgebra of $\V_{{\rm SL}_2,s}\otimes \Pi_{{\rm SL}_2}$ that is strongly generated by $e_i,h_i,f_i$, for $i=1,\dots s+1$; the $X^{a_1a_2\dots a_{s+1}}$; and the conformal vector $T_{s+1}$.

Note also that this construction makes manifest the action of the symmetric group $S_s$ on $\V_{{\rm SL}_2,s}$. The symmetric group acts by permuting the various affine current algebras as well as by the indices of $X^{a_1a_2\dots a_s}$. The presence of this automorphism is a consequence of generalised $S$-duality in four dimensions. The same symmetry is manifest in the construction of~\cite{Arakawa:2018egx}, which uses $s$-identical copies of $\W_{{\rm SL}_2}\equiv \V_{{\rm SL}_2,1}$ as the starting point of its construction of $\V_{{\rm SL}_2,s}$.

It is an interesting feature of these free field realisations that the prescription for iterative addition of punctures is a kind of inverse functor to DS reduction---much like in~\cite{Adamovic2019}. Given that such an inverse functor has already been provided by the Feigin--Frenkel gluing functor $\Hi{0}(\ZZ,\W_G\otimes -)$ in~\cite{Arakawa:2018egx}, it is a natural question whether the relation between the two inverse operations can be clarified, and in particular whether the free field approach can be understood as providing a solution to the semi-infinite cohomology problem arising in the Feigin--Frenkel gluing.

%----------------------------------------------------%
\subsection{Examples}
%----------------------------------------------------%

In the remainder of this section we present details of the implementation of this algorithm for some small values of $s$, leading to explicit constructions of the corresponding $\V_{{\rm SL}_2,s}$. We start with $s=1$ and work up to $s=6$ iteratively.

%----------------------------------------------------%
\subsubsection{\label{subsubsec:cap}The cap}
%----------------------------------------------------%

The cap vertex algebra was dubbed the \emph{equivariant affine $\WW$-algebra} in~\cite{Arakawa:2018egx}; here we consider the $\af_1$ case $\V_{{\rm SL}_2,1}\simeq\W_{{\rm SL}_2}$. A concrete expression for the strong generators and OPEs of this algebra can be found in the appendix of \textit{loc.\ cit.}, and our realisation will eventually reproduce those results up to field redefinitions. First, however, the general scheme presented above leads to an alternative presentation for $\W_{{\rm SL}_2}$.

Let $e_1$, $h_1$, $f_1$ be the strong generators of the $V^{\kappa_c}(\slf_2)$ subalgebra inside $\Pi^0_{{\rm SL}_2}\otimes \Pi^1_{{\rm SL}_2}$ defined as in \eqref{eq:sl2_ffr_specialised} and \eqref{eq:sl2_ffr_with_centre}. The highest weight state of the generator $X^a$ (transforming in the fundamental representation of $\slf_2$) is realised by adjoining $e_1^{\frac12}$ to the generator $X$ of $\mathbf{I}_{{\rm SL}_2}$, while the lowest weight state is recovered by acting with $f_1$. This gives the following strong generators,
%----%
\begin{equation}
\begin{split}
	X^+	&= \gamma^{-\frac{1}{2}} e_1^{\frac{1}{2}}~,\\
	X^- &\coloneqq f_{1,(0)} X^{+} = -(b+\partial\delta_1)\gamma^{-\frac12}e_1^{-\frac12}~.
\end{split}
\end{equation}
%----%
These generators have conformal weight $-\frac12$ and have regular OPEs between themselves. Together with the affine currents, they strongly generate a vertex algebra with a nontrivial center, the generator $C_{cap}$ of which takes the form
%----%
\begin{equation}
	C_{cap}=J^{(ab)}X_a X_b-2(\partial X^a)X_a~,
\end{equation}
%----%
where we have introduced a matrix $J_1^{ab}$ of $\slf_2$ currents according to
%----%
\begin{equation}
	J_1 \coloneqq
	\begin{pmatrix}
		e_1 & -\frac{1}{2}h_1\\
		-\frac{1}{2} h_1 & -f_1
	\end{pmatrix}~,
\end{equation}
%----%
and $X^+\equiv X^1$, $X^-\equiv X^2$, and where $\slf(2)$ indices are lowered according to $\ast_a=\epsilon_{ab}\ast^b$, with $\epsilon_{+-}=1$. The central quotient of this vertex algebra where we enforce $C_{cap}=-\id$ will be identified with $\W_{{\rm SL}_2}$.

This presentation does not match that of \cite{Arakawa:2018egx}, which is due to the amusing fact that this VOA can be equivalently defined to be strongly generated by a different subset of fields. Indeed, the naive expectation for the strong generators of the VOA realised by (one-sided) quantum Drinfel'd--Sokolov reduction from CDOs on ${\rm SL}_2$ is that there will be weight $-\frac12$ and a weight $+\frac12$ fundamentals of the remaining $\slf_2$, along with a Feigin--Frenkel generator. The extra weight $+\frac12$ generators can be written as $Y^a = J^{ab}X_b$, which in terms of free fields gives
%----%
\begin{equation}
\begin{split}
	Y^+	&= b\gamma^{-\frac{1}{2}}e_1^\frac{1}{2}~.\\
	Y^- &\coloneqq f_{1,(0)}Y^+ = b(b+\partial\delta_1)\gamma^{-\frac12}e_1^{-\frac12} + \gamma^\frac12e_1^{-\frac12}~.
\end{split}
\end{equation}
%----%
Meanwhile, the Feigin--Frenkel generator $S=J^{ab}J_{ab}$ is, of course, given as in $\mathbf{I}_{{\rm SL}_2}$ by
%----%
\begin{equation}
	S=b^2+\gamma+\partial b~.
\end{equation}
%----%
The $X^a$, $Y^a$, and $S$ also strongly generate $\W_{{\rm SL}_2}$, with their singular OPEs being given by
%----%
\begin{equation}
\begin{split}
	S(z)X^a(w) 		&=  \frac{\frac{3}{4}X^a}{(z-w)^2} + \frac{Y^a}{z-w}~,\\
	S(z)Y^{a}(w) 	&=  \frac{\frac{3}{4} Y^a}{(z-w)^2} + \frac{SX^a}{z-w}~,\\
	Y^a(z)X^b(w) 	&=  \frac{\frac{1}{2}(X^{a}X^{b})}{z-w}~,\\
	Y^a(z)Y^b(w) 	&= -\frac{\frac{1}{4}(X^{a}X^{b})}{(z-w)^2}+\frac{\epsilon^{ab}\frac{1}{2}\big( Y^cX_d -\frac12\partial X^c X_d \big)-\frac14 \partial X^a X^b}{z-w}~.
\end{split}
\end{equation}
%----%
The null relation in this case is a lift of~\eqref{eq:det_chiral_sl2} in $\mathbf{I}_{{\rm SL}_2}$, and is given by
%----%
\begin{equation}
	(X^aY_a) + \frac{1}{2}\big(\partial X^aX_a\big) = \id~,
\end{equation}
%----%
which can also be thought of as the chiralised version of the determinantal relation in the equivariant Slodowy slice ${\rm SL}_2\times S_f$.

In this presentation, the affine currents $e_1$, $h_1$, $f_1$ are realised as composites,
%----%
\begin{equation}
\begin{split}
	e_1 =& SX^+X^+ - Y^+Y^+ -\frac{3}{2}\big(X^+\partial Y^+- \partial X^+ Y^+) -\frac{3}{4} \partial X^+\partial X^+ - \frac{1}{8} X^+\partial^2 X^+~,\\
	h_1 =& -2SX^-X^+ +2 Y^-Y^+ + \frac{3}{2}(X^-\partial Y^+ + X^+\partial Y^- - \partial X^- Y^+ - \partial X^+ Y^-)\\
	&+\frac{3}{2} \partial X^+ \partial X^- + \frac{1}{4}X^+\partial^2 X^-~, \\
	f_1 =& -SX^-X^- +Y^-Y^- + \frac{3}{2}\big(X^-\partial Y^- - \partial X^- Y^-\big) + \frac{3}{4}\partial X^- \partial X^- +\frac{1}{8} X^-\partial^2 X^-~.
\end{split}
\end{equation}
%----%
Similarly, the stress tensor $T_1$ is composite and takes the form
%----%
\begin{equation}
\begin{split}
	T_1 &= -S \big(\partial X^a X_a\big) + \partial Y^a Y_a +\frac{3}{2} \big(\partial X^a\partial Y_a\big) +\frac38 \partial^2 X^a\partial X_a + \frac16 \partial^3X^aX_a~.
\end{split}
\end{equation}
%----%
(The conformal vector can also be expressed in terms of just the currents and the $X^a$, but the form is not particularly illuminating.) This second presentation of the cap agrees precisely with the expressions appearing in the appendix of~\cite{Arakawa:2018egx} after making the identifications $S\mapsto -S$, $X^i \mapsto - \begin{pmatrix} c \\ a\end{pmatrix}$ and $Y^i\mapsto \begin{pmatrix} d \\ b \end{pmatrix} $.

%----------------------------------------------------%
\subsubsection{\label{subsubsec:cylinder}The cylinder}
%----------------------------------------------------%

Given $\W_{{\rm SL}_2}$, we can construct the cylinder VOA as a subalgebra $\D_{{\rm SL}_2,-2}\subset\W_{{\rm SL}_2}\otimes \Pi_{{\rm SL}_2}$. This involves first defining a new $V^{\kappa_c}(\slf_2)$ subalgebra inside $\W_{{\rm SL}_2}\otimes \Pi_{{\rm SL}_2}$---now with generators $e_2,h_2,f_2$---which commutes with the previous affine current algebra. 
We further define generators
%----%
\begin{equation}
\begin{split}
	X^{a+} &\coloneqq X^ae_2^{\frac{1}{2}}~,\\
	X^{a-} &\coloneqq f_{2,(0)} X^{a + }~.
\end{split}
\end{equation}
%----%
Replacing the $X^a$ by their free field expressions from the previous subsection, we have
%----%
\begin{equation}
\begin{split}
	X^{++}	& = \gamma^{-\frac12}e_1^{\frac12}e_2^{\frac12}~, \\
	X^{-+}	& =-(b+\partial\delta_1)\gamma^{-\frac12}e_1^{-\frac12}e_2^{\frac12}~, \\
 	X^{--}	& = \big((b+\partial\delta_1)(b+\partial\delta_2)+ \big)\gamma^{-\frac12} e_1^{-\frac12}e_2^{-\frac12}~.
\end{split}
\end{equation}
%----%
The expressions for $X^{-+}$ and $X^{+-}$ are related by performing the permutation $1\leftrightarrow2$ on all subscripts. This arises from the manifest permutation symmetry of our construction.

The generators $X^{ab}$ live in the bifundamental representation of $(\slf_2)_1\times(\slf_2)_2$, with action given by the zero modes of $e_i$, $h_i$, $f_i$. The OPEs between the $X^{ab}$ are regular; the only non-trivial OPEs involve $e_i$, $h_i$, $f_i$.

The (central) null ideal in the cap is now lifted to
%----%
\begin{equation}
	\begin{vmatrix}
	X^{++} & X^{+-} \\
	X^{-+} & X^{--}
	\end{vmatrix}
	= X^{++}X^{--} - X^{+-}X^{-+} \overset{!}{=} \id~,
\end{equation}
%----%
and the commutative subalgebra generated by the $X^{ab}$ is identified with the commutative vertex subalgebra $\OO(J_{\infty}{\rm SL}_2)$,~\ie, the chiralisation of $\OO({\rm SL}_2)$ inside $\D_{{\rm SL}_2,-2}$. The $X^{ab}$ along with the currents, $e_1,h_1,f_1$, are strong generators of $\D_{{\rm SL}_2,-2}$.

The conformal vector is again composite, now being given by
%----%
\begin{equation}
	T_2 = J^{ab}\partial X_{a}\,\! ^cX_{cb} + \partial^2 X^{ab}X_{ab} +2 \partial X^{ab}\partial X_{ab}~.
\end{equation}
%----%
Replacing all strong generators with their free field expressions, this reduces exactly to the expression for the canonical $T_2$ in terms of lattice bosons with background charges.

%----------------------------------------------------%
\subsubsection{\label{subsubsec:three_punctures}Three punctures}
%----------------------------------------------------%

Now we reach the traditional starting point of the free field realisation of~\cite{Beem:2019tfp}, the three punctured sphere. The vertex algebra $\V_{{\rm SL}_2,3}$ is isomorphic to the symplectic bosons associated to the $\mathbf{2}\otimes\mathbf{2}\otimes\mathbf{2}$ representation of $(\slf_2)_1\times(\slf_2)_2\times(\slf_2)_3$. Repeating the same steps that are by now becoming familiar, we define a $V^{\kappa_c}(\slf_2)$ subalgebra inside $\D_{{\rm SL}_2,-2}\otimes \Pi_{{\rm SL}_2}$, generated by $e_3$, $h_3$, $f_3$. We further define trifundamental generators
%----%
\begin{equation}
	X^{ab+}\coloneqq X^{ab}e_3^{\frac12}~,\quad X^{ab-}\coloneqq f_{3,(0)}X^{ab+}~,
\end{equation}
%----%
In terms of the full free field realisation, these generators are given by:
%----%
\begin{equation}
\begin{split}
	 X^{+++} &= \gamma^{-\frac12}e_1^\frac12e_2^\frac12e_3^\frac12~,\\
	X^{-++} &=-(b+\partial\delta_1) \gamma^{-\frac12}e_1^{-\frac12}e_2^\frac12e_3^\frac12~,\\
	X^{--+} &= \big((b+\partial\delta_1)(b+\partial\delta_2) + \gamma\big)\gamma^{-\frac12}e_1^{-\frac12}e_2^{-\frac12}e_3^\frac12~,\\
	X^{---} &=- \big((b+\partial\delta_1)(b+\partial\delta_2)(b+\partial\delta_3)+ ( b+ \partial \delta_1+\partial \delta_2+\partial \delta_3)\gamma-\partial\gamma\big)\gamma^{-\frac12}e_1^{-\frac12}e_2^{-\frac12}e_3^{-\frac12}~,\\
\end{split}
\end{equation}
%----%
where the other generators are related to the ones displayed by permutations of the lattice variables. The OPEs can be computed direclty, and now take the (expected) exceptionally simple form,
%----%
\begin{equation}
	X^{abc}(z)X^{a'b'c'}(w) \sim \frac{\epsilon^{aa'}\epsilon^{bb'}\epsilon^{cc'}}{z-w}~.
\end{equation}
%----%
There are no non-trivial null relations amongst the $X^{abc}$, but the relations in the cylinder now lift to relations identifying the various current subalgebras as composites of the $X^{abc}$. In particular, we now have
%----%
\begin{equation}
\begin{split}
	J^{ab}_1 = \frac{1}{2} X^{ai_1j_1}X^b_{i_1j_1}~,\\
	J^{ab}_2 = \frac{1}{2} X^{i_1aj_1}X_{i_1\phantom{b}j_1}^{\phantom{i_1}b}~,\\
	J^{ab}_3 = \frac{1}{2} X^{i_1j_1a}X_{i_1j_1}^{\phantom{i_1j_1}b}~.
\end{split}
\end{equation}
%----%
The stress tensor $T_3$ is just the standard symplectic boson stress tensor, which again reduces to the canonical lattice conformal vector with background charges.

%----------------------------------------------------%
\subsubsection{\label{subsubsec:four_punctures}Four Punctures}
%----------------------------------------------------%

The free field realisation of $\V_{{\rm SL}_2,4}$ inside $\V_{{\rm SL}_2,3}\otimes \Pi_{{\rm SL}_2}$ can be found in~\cite{Beem:2019tfp}, but we will revisit the computation in the present context. We again realise a new $V^{\kappa_c}(\slf_2)$ subalgebra inside $\V_{{\rm SL}_2,3}\otimes \Pi_{{\rm SL}_2}$, and as above, adopt a basis for the currents $J^{ab}_4$.

The quadrifundamental generators take the usual form,
%----%
\begin{equation}
	X^{a_1a_2a_3+} \coloneqq X^{a_1a_2a_3}e_4^{\frac{1}{2}}~,\quad X^{a_1a_2a_3-}\coloneqq f_{4,(0)}X^{a_1a_2a_3+}~.
\end{equation}
%----%
which can be expressed entirely in terms of free fields as follows
%%%%%
\begin{equation}
\begin{split}
	X^{++++} &= \gamma^{-\frac12}e_1^\frac12e_2^\frac12e_3^\frac12e_4^\frac12~,\\
	X^{-+++} &=-(b+\partial\delta_1) \gamma^{-\frac12}e_1^{-\frac12}e_2^\frac12e_3^\frac12e_4^\frac12~,\\
	X^{--++} &= \big((b+\partial\delta_1)(b+\partial\delta_2) + \gamma\big)\gamma^{-\frac12}e_1^{-\frac12}e_2^{-\frac12}e_3^\frac12e_4^\frac12~,\\
	X^{---+} &=- \big((b+\partial\delta_1)(b+\partial\delta_2)(b+\partial\delta_3)+ ( b+ \partial \delta_1+\partial \delta_2+\partial \delta_3)\gamma - \partial\gamma\big)\gamma^{-\frac12}e_1^{-\frac12}e_2^{-\frac12}e_3^{-\frac12}e_4^\frac12~,\\
	X^{----} &= \bigg(\prod_{i=1}^4(b+\partial \delta_i) + \sum_{i<j}^{4}(b+\partial\delta_i)(b+\partial\delta_j)\gamma  -2(b+\partial\delta_1+\partial\delta_2+\partial\delta_3+\partial\delta_4)(b+\partial \gamma^{\frac12} \gamma^{-\frac12})\gamma \\
	& -2(bb + \partial b)\gamma +b\partial\gamma + \gamma^2+\partial^2\gamma\bigg)\gamma^{-\frac12}e_1^{-\frac12}e_2^{-\frac12}e_3^{-\frac12}e_4^{-\frac12}~.
\end{split}
\end{equation}
%%%%%
Again one finds free-field expressions for the unlisted generators by applying the appropriate permutation on subscripts. The OPE between the generators is
%----%
\begin{equation}
	X^{a_1a_2a_3a_4}(z) X^{b_1b_2b_3b_4}(w) \sim \frac{2\prod_{i=1}^{4} \epsilon^{a_ib_i}}{(z-w)^2} + \frac{\sum_{i=1}^4 \prod_{j\neq i} \epsilon^{a_jb_j} J^{(a_ib_i)}(w)}{z-w}~,
\end{equation}
%----%
and taken together with the three sets of $V^{\kappa_c}(\slf_2)$ currents these are strong generators of $L_{-2}(\mathfrak{d}_4)$, the simple quotient of $V^{-2}(\mathfrak{d}_4)$.
The null ideals of this vertex algebra and their relation to Joseph's minimal realisations \cite{Joseph:1974min} were investigated in \cite{Beem:2019tfp}.

%----------------------------------------------------%
\subsubsection{\label{subsubsec:five_punctures}Five Punctures}
%----------------------------------------------------%

Continuing further, we come to a realisation of $\V_{{\rm SL}_2,5}$. We define the usual $e_5,h_5,f_5$ currents and the generators $X^{a_1a_2a_3a_4a_5}$ as before. In terms of free-fields, the expressions are somewhat burdensome to write down. (Additionally, it is only the state $X^{-----}$ that is novel. All other states can be constructed (by appealing to symmetry) from the free-field realisations we have already constructed.)

The nontrivial OPEs are the OPEs of the pentafundamental generators $X^{a_1a_2a_3a_4a_5}$ amongst themselves, and at this stage the OPEs are very intricate and so to keep the expressions to a manageable length we will suppress all conformal descendants,~\ie\ we only show contributions from conformal primaries in the OPE. We then have
%----%
\begin{equation}
\begin{split}
	X^{a_1a_2a_3a_4a_5}(z)X^{b_1b_2b_3b_4b_5}(w) &\sim \frac{4\prod_{i=1}^5 \epsilon^{a_ib_i} }{(z-w)^3} +  \frac{2\sum_{i=1}^5 \prod_{j\neq i}^5\epsilon^{a_jb_j}J_i^{(a_ib_i)}(w)}{(z-w)^2}+ \\
	&\frac{1}{z-w}\bigg(\sum_{i<j}^5\prod_{k\neq i,j}^5\epsilon^{a_kb_k} (J_i^{(a_ib_i)}J_j^{(a_jb_j)})(w) - 2\prod_{i=1}^5 \epsilon^{a_ib_i} (S+\tfrac{1}{4} T_5)(w)\bigg)~.
\end{split}
\end{equation}
%----%
Here, we have the appearance of a new primary $S+\frac{1}{4}T_5$. In general, one can show that the state $S-\frac{6}{c_s}$ is a primary for $T_s$, with the appropriate value of the central charge.

Finding the exact expression for the generators of the null ideal is now somewhat involved, though largely a book-keeping exercise using the free field expressions. We will not reproduce them here or in the next example.

%----------------------------------------------------%
\subsubsection{\label{subsubsec:six_punctures}Six Punctures}
%----------------------------------------------------%

We conclude with the realisation of $\V_{\af_1,6}$, which has strong generators $(e_i,h_i,f_i)$ for $i=1,2,\dots,6$, the hexafundamental $X^{a_1a_2a_3a_4a_5a_6}$, and the stress tensor $T_6$. The nontrivial OPE is the $X\times X$ self-OPE of the hexafundamental, which takes the form
%----%
\begin{equation}
\begin{split}
	X^{a_1a_2a_3a_4a_5a_6}(z)X^{b_1b_2b_3b_4b_5b_6}(w) &\sim \frac{8\prod_{i=1}^6\epsilon^{a_ib_i}}{(z-w)^4} +4 \frac{\sum_i \prod_{j\neq i} \epsilon^{a_jb_j} J^{(a_ib_i)(w)}}{(z-w)^3} +\\
	&\frac{1}{(z-w)^2}\bigg(2\sum_{i<j}\prod_{k\neq i,j}^6 \epsilon^{a_kb_k} (J^{(a_ib_i)}J^{(a_jb_j)})(w) -6 \prod_{i=1}^6 \epsilon^{a_ib_i} (S+\tfrac{3}{17}T_6)\bigg) +\\
	&\frac{1}{(z-w)}\bigg(\sum_{i<j<k} \prod_{l\neq i,j,k} \epsilon^{a_lb_l}(J^{(a_ib_i)}J^{(a_jb_j)}J^{(a_kb_k)})\\
	&-2\sum_{i}^6\prod_{k\neq i}\epsilon^{a_kb_k}\bigg((S+\tfrac{3}{17}T_6)J^{(a_ib_i)}
	+  A^{a_ib_i} \bigg)
\end{split}
\end{equation}
%----%
where there appear novel primary operators, $A^{a_ib_i}$, given by
%----%
\begin{equation}
	A_i^{a_ib_i} = \bigg((J_i)^{a_i}_{c}(\partial J_i)^{c b_i} -\frac{1}{8}T(J_i)^{a_i b_i} -\frac{5}{48} (\partial^2 J_i)^{a_ib_i}\bigg)~.
\end{equation}
%----%
It may be noteworthy that there seems to be some numerology involved in the residues appearing in the OPE. This is not entirely surprising, since the OPEs in the free field realisation are computed by Wick contractions---which are manifestly of a combinatorial nature. It would be interesting to look for a more elegant combinatorial description of these OPEs that makes the $n$-punctured generalisation immediate, but we will not pursue this for now.

%----------------------------------------------------%
\section{\label{sec:conclusions}Concluding remarks}
%----------------------------------------------------%

%----------------------------------------------------%
\subsection{Screening charges}
%----------------------------------------------------%

In many instances, free-field realisations of vertex algebras are accompanied by certain screening charges whose kernel recovers the subalgebra of interest (\cf\ \cite{Adamovic2019, Adamovic2003115, deBoer:1996ti, Frenkel:2007}). One can then often think of a free field realisation of a vertex algebra as a resolution whose first differential is the screening charge.

In the free field realisation of Wakimoto modules, for example, such screening charges exist and admit a geometric interpretation~\cite{Frenkel:geo2008}. Namely, they arise from the Cousin--Grothendieck spectral sequence associated to the Bruhat stratification on the flag variety; the free-field realisation is a chiral analogue of the BGG resolution for $\gf$-modules.

A slight modification of the arguments in \textit{loc.\ cit.} can be applied to our case. First, we fix some partial ordering $w_0\ge w_1\ge w_2\ge\dots w_N$ on the Weyl group $W$, with $w_0$ the longest element as before. Then $G$ has a stratification
%----%
\begin{equation}
	G=\underset{i}{\bigsqcup}\, G^i~,
\end{equation}
%----%
where $G^i = B w_i B$ and $G^0\equiv G^*$. Let $j_i:G^i\hookrightarrow G$ denote the immersions of the strata. If $M$ is a $D$-module on $G$, then the Cousin--Grothendieck spectral sequence tells us that $M$ has a resolution
%----%
\begin{equation}
	0\rightarrow M\hookrightarrow j_{0*}j_0^!M\xrightarrow{\delta_0} j_{1*}j_1^!M\xrightarrow{\delta_1}\dots\xrightarrow{\delta_{N-1}} j_{N*}j_{N}^!M\rightarrow 0~.
\end{equation}
%----%
In particular, if $M=\DD_G$ is the sheaf of differential operators on $G$, then $j_0^!\DD_G$ can be identified with the sheaf of differential operators on the big cell, $\DD_{G^*}$. From \cite[Theorem 5.2.]{Arakawa:tdo2011}, this lifts to a resolution  of the sheaf of (critical level) chiral differential operators on $G$, $\D_{G,\kappa_c}$. Taking global sections, we have a complex of vertex algebras $(C^\bullet,\delta_\bullet)$, with $C^0= \D_{G^*,\kappa_c}(G^*)$.

To construct a resolution of the chiral universal centraliser, we can perform the two-sided Drinfel'd--Sokolov reduction on $\D_G(G)$. We form the tri-complex $(C^{\bullet\bullet\bullet},\delta,d_L,d_R)$, where $d_{L}$ and $d_R$ are the BRST differentials computing the left and right DS reductions, respectively. Passing to $\delta$-cohomology and then $d_L$ and $d_R$ cohomology will, by construction, recover the chiral universal centraliser in cohomological degree zero.

Note that the DS reduction subcomplex is acyclic away from degree zero (see \cite{Frenkel:lan2006} for more detail), so a spectral sequence argument implies that we can algernatively compute the $d_L$ and $d_R$ cohomology first before computing the $\delta$-cohomology and still realise the chiral universal centraliser in degree zero. Since $G^*\cong N\times T\times N$, the two-sided DS reduction gives $\D_T(T)$ in degree zero. The resulting complex, $(\tilde{C}^{\bullet},\delta)$ has zeroth term $\tilde{C}^0 = \D_T(T)$, with
%----%
\begin{equation}
	H^0(\tilde{C},\delta)\equiv\mathrm{ker} ~ \delta_0 \cong \mathbf{I}^{cl}_G~.
\end{equation}
%----%
Of course, $\D_T(T)$ is isomorphic to the free field vertex algebra $\Pi_G$, so the above presentation suggests a reinterpretation of our free-field realisation along with some kind of screening charge $\delta_0$, \ie, the restriction of the Cousin--Grothendieck differential. However, the above analysis at this level of abstraction does not lead to an explicit expression for a screening charge in the usual sense. We leave a more detailed and concrete analysis for future work.

%----------------------------------------------------%
\subsection{Monopoles and Abelianisation}
%----------------------------------------------------%

Our realisation of the universal centraliser in terms of functions on the Kostant--Toda lattice is somewhat evocative of Abelianisation-type constructions, in that $\KT_G\cong T^*(\IC^\times)^{\rk\,\gf}\cong T^*\check{T}$ and $T^\ast\IC^\times$ may be identified as the Coulomb branch of pure $U(1)$ $\mathcal{N}=4$ gauge theory in three dimensions. Indeed, the construction of Braverman--Finkelberg--Nakajima for the Coulomb branches of gauge theories comes with an Abelianisation map. Let $\check{T}$ be the Langlands dual of the maximal torus $T$, and let $\mathrm{Gr}_{\check{T}}= \check{T}(\KK)/\check{T}(\OO)$ be the affine Grassmannian. Since $\check{T}$ is abelian, the affine Grassmannian is just a collection of points, parameterised by the coweights of $\check{T}$. There is a closed embedding $i:\mathrm{Gr}_{\check{T}}\hookrightarrow \mathrm{Gr}_{\check{G}}$, which gives rise to a dominant morphism~\cite{Braverman:2016wma},
%----%
\begin{equation}
	\Zf_G\cong \mathrm{Spec}~H^{\check{G}(\OO)}_\bullet(\mathrm{Gr}_{\check{G}})\rightarrow\mathrm{Spec}~H^{\check{T}(\OO)}_\bullet(\mathrm{Gr}_{\check{T}})/\!\!/W \cong (\mathfrak{t}\times \check{T})/\!\!/W~.
\end{equation}
%----%
This is actually an isomorphism away from the closed subspace $\mathfrak{t}^\circ = \bigcup_{\alpha\in\Delta} \mathfrak{t}_\alpha$, where $\mathfrak{t}_\alpha$ is the root hyperplane $\alpha=0$. As a result, $\Zf_G$ and $(\mathfrak{t}\times \check{T})/\!\!/W$ are birational. In~\cite{Bezrukavnikov:2003}, this was used to construct the universal centraliser in terms of a blow up of $(\mathfrak{t}\times \check{T})/\!\!/W$.

In the physics literature, $(\mathfrak{t}\times \check{T})/\!\!/W$ is just the classical description of the Coulomb branch, which receives quantum corrections from monopole operators~\cite{Bullimore:2015lsa}. It is tempting, then, to attempt to interpret our free field realisations (and the corresponding finite-dimensional constructions) as an instance of this Abelianisation paradigm. However, such an interpretation is not immediately forthcoming, as in particular for us the Poisson structure on the ``Abelian'' space ($\KT_G$) is \emph{not corrected}, by monopole operators or otherwise.

A construction of Darboux coordinates by Kato has appeared in \cite{Kato:2020jup}, however this is for the $K$-theoretic Coulomb branch of pure gauge theory, which is not the universal centraliser.

Nevertheless, it would be interesting to investigate the possibility of a link between the two pictures. This could help frame our construction in terms that are more in keeping with both physics and geometric representation theory, as well as offering the possibility for generalisation away from the case of pure gauge theory.

%----------------------------------------------------%
\subsection{Open questions}
%----------------------------------------------------%

In addition to the various conjectures outlined in the main body of the papers, there are a number of issues that remain open to which we hope someone (ourselves or otherwise) will return.

First of all, in starting our construction of chiral algebras of class $\SS$ with the universal centraliser (of type ${\mathrm SL}_2$), it is clear that the vertex algebras $\V_{{\mathrm SL}_2,s}$ are sensitive to the choice of the Lie \emph{group}. Indeed, one could replicate the construction, starting with $\mathbf{I}_{{\mathrm PGL}_2}$, and recover a series of vertex algebras $\V_{{\mathrm PGL}_2,s}$ which are not isomorphic to $\V_{{\mathrm SL}_2,s}$. It would seem natural to expect the $\V_{{\mathrm PGL}_2,s}$ to be finite quotients of the $\V_{{\mathrm SL}_2,s}$ with respect to the $\IZ/2\IZ$ action of the centre. However, it is unclear whether passing to the $\IZ/2\IZ$ invariants in $\V_{SL_2,s}$ commutes with passing to the Zhu's $C_2$ algebra.

The most natural progression is to continue the free-field program for $\V_{G,s}$ to groups of higher rank. For ${\rm SL}_3$, we already have a concrete presentation for the chiral universal centraliser and one could hope in this manner to obtain a presentation of the equivariant affine $\WW$-algebra as a strongly generated vertex algebra. However, the calculations for ${\rm SL}_2$ already become unwieldy at six punctures. To establish the free field realisation on stronger footing we would need to develop a more abstract formalism---rather than the algorithmic approach we have taken so far. We hope to develop such a formalism in upcoming work \cite{BBS:unpub}.

The free-field realisations naturally give rise to a construction of a class of vertex algebra modules. Namely, there are natural restriction functors arising from $\Pi_{{\rm SL}_2}^{\otimes s}$ modules to modules of $\V_{{\rm SL}_2,s}$. There is also, of course, a web of such restrictions, \eg, from $\Pi_{{\rm SL}_2}\otimes \V_{{\rm SL}_2,s}$ to $\V_{{\rm SL}_2,s+1}$ modules. One would hope that the properties of these functors give some partial characterisation of the modules of the chiral algebras of class $\SS$, which are of interest in both four dimensional and three dimensional physics.

\section*{Acknowledgments}
The authors are grateful to Tomoyuki Arakawa, Dylan Butson, Carlo Meneghelli, and Leonardo Rastelli for useful discussions and collaborations on related topics. Christopher Beem's work is supported in part by grant \#494786 from the Simons Foundation, by ERC Consolidator Grant \#864828 ``Algebraic Foundations of Supersymmetric Quantum Field Theory'' (SCFTAlg), and by the STFC consolidated grant ST/T000864/1. Sujay Nair is supported in part by EPSRC studentship \#2272671.

%----------------------------------------------------------------------%

\appendix

%----------------------------------------------------------------------%
\section{\label{app:proof_Poisson}Proof of Theorem~\ref{thm:Poisson_generation}}
%----------------------------------------------------------------------%

In this appendix, we prove Theorem \ref{thm:Poisson_generation} regarding the Poisson structure of the universal centraliser for $G={\rm SL}_N$. Namely, we show the following.

%-------------------------------------------%
\begin{sthm}
The Poisson algebra $I_{{\rm SL}_N}$ is Poisson generated by $(P_i)_{i=1}^{\rk\,\gf}$ and the generator $X=g_{N1}$.
\end{sthm}
%-------------------------------------------%

\vspace{-11pt}

%-------------------------------------------%
\begin{proof}
Throughout this appendix we adopt the notation $G={\rm SL}_N$ and $\gf=\slf_N$. The proof makes extensive use of the realisation of $\Zf_{G}$ as the two-sided Kostant--Whittaker reduction of $T^*G$. As such, we start with some preliminaries and also recall certain technical details regarding Poisson (coisotropic) reduction in the present setting.

%----------------------------------------------------------------------%
\subsection{\label{subapp:preliminaries}Preliminaries}
%----------------------------------------------------------------------%

Consider the triple $\gf^*_L\times \gf^*_R\times G$, where the subscripts simply distinguish the two $\gf^*$ factors. Choose generators $\widetilde{g}_{ij}$---conjugate to the matrix elements of a $G$ matrix---for $\IC[G]\subset\IC[\gf_L^*\times\gf^*_R\times G]$. This space is equipped with a Poisson structure by letting linear functions in $\IC[\gf^*_{R/L}]\cong \rm Sym~\gf$ act according to the left (right) $\gf$ action on the $\widetilde{g}_{ij}$, which then organise themselves into the bifundamental representation $V_{F}\otimes V_{F}^*$.

The cotangent bundle $T^\ast G$ is embedded in $\gf^*_L\times \gf^*_R \times G$ as a closed, Poisson subvariety,
%----%
\begin{equation}
\label{eq:cotangent_subvariety}
	T^*G = \{(x,y,g)\in \gf^*_L\times \gf^*_R \times G~|~ Ad^*_g x = y\}~.
\end{equation}
%----%
Let $\JJ$ be the corresponding defining ideal for $T^*G$; it is linear in the generators of $\IC[\gf^*_{L/R}]$ and is degree $N$ in the generators, $\widetilde{g}_{ij}$.

The moment maps for left and right $G$ actions on $T^*G$ are inherited from the projections $\pi_{L/R}:\gf^*_L\times \gf^*_R\times G\rightarrow \gf^\ast_{L/R}$, respectively, which are the $\gf_L$ and $\gf_R$ moment maps of the ambient space. We will study the two-sided KW reduction of this ambient space and carefully analyse the fate of the linear coordinate functions $\widetilde{g}_{Ni}$ in particular.

Let $\XX\colonequals(\chi+\mathfrak{b}^*)_L\times (\chi+\mathfrak{b}^*)_R\times G = \pi_L^{-1}(\chi+\mathfrak{b}^*)\cap \pi_R^{-1}(\chi+\mathfrak{b})$, which is a coisotropic subvariety of our ambient space. The two-sided KW reduction is then equivalent to the quotient $\WW\coloneqq \XX/\!\!/N_L\times N_R$. Let $\pi_{LR}$ be the corresponding projection $\pi_{LR}:\XX\rightarrow \WW$.

To realise the two-sided KW reduction of $T^*G$, we first define its intersection $\YY\colonequals T^\ast G \cap \XX$. The reduction in question is then given by $\YY/\!\!/N_L\times N_R\cong\Zf_G$. Summarising, we have the following base change,\!\footnote{By a slight abuse of notation we also denote the natural projection $\YY\rightarrow \YY/\!\!/N_L\times N_R$ by $\pi_{LR}$.}
%----%
\begin{equation}
\label{eq:base_change}
	\begin{tikzcd}
	T^*G\ar[d] 						& \YY \ar[l] \ar[d] \ar[r, "\pi_{L R}"] 	& \Zf_G \ar[d]\\
	\gf^*_L\times\gf^*_R \times G 	& \XX \ar[l,"\iota",swap] \ar[r, "\pi_{L R}"] 	& \WW
	\end{tikzcd}~.
\end{equation}
%----%
Each square is Cartesian; they both describe $\YY$ as a fibre product. The arrows directed down and to the left in the diagram are all closed embeddings. Moreover, $T^*G\rightarrow \gf_L^*\times \gf_R^*\times G$ is Poisson, so $\Zf_G\rightarrow \WW$ is also Poisson by functoriality.

The closed subvariety $S_f\times S_f\times G\subset (\chi+\mathfrak{b}^*)_L \times (\chi+\mathfrak{b}^*)_R\times G$ is a global slice to the $N_L\times N_R$ action on $\XX$, so $\WW$ can be identified with $ S_f\times S_f \times G$ as varieties. Letting $\II$ denote the defining ideal of $S_f\times S_f\times G$ inside $\XX$, the Poisson structure on $\WW\simeq S_f\times S_f \times G$ is then defined in the following (standard) manner.

For every $a,b\in \IC[\WW]$ we have canonical extensions $\pi_{LR}^\#(a), \pi_{LR}^\#(b)\in\IC[\XX]$. Choose any further extensions $A,B\in \IC[\gf^*_L\times \gf^*_R\times G]$. Then the Poisson bracket on $\WW$ is the unique bracket satisfying
%----%
\begin{equation}
	\pi_{LR}^\#(\{a,b\}_{\WW}) =  \iota^\#(\{A,B\}_{\gf_L^*\times\gf_R^*\times G})~.
\end{equation}
%----%
In the present case, things are considerably simpler than in general. Since $G\times S_f\times S_f$ is a slice, the following triangle commutes.
%----%
\begin{equation}
	\begin{tikzcd}
	\XX \ar[rr,"\pi_{LR}"]& & \WW\ar[dl,"\sim",leftrightarrow]&\\
	& \ar[ul, hookrightarrow] S_f\times S_f\times G &
	\end{tikzcd}
\end{equation}
%----%
On functions, this means that if $f\in \IC[\WW]$ is identified with to $f'\in \IC[S_f\times S_f\times G]$ under the isomorphism, then its extension $\pi^\#_{LR}(f)$ has the form $\pi^\#_{LR}(f) = f' + \II$. Consequently, determining the extension $\pi_{LR}^\#$ is equivalent to identifying an injective section
%----%
\begin{equation}
	\sigma:\IC[S_f\times S_f\times G]\cong\IC[\XX]/\II\rightarrow \IC[\XX]~,
\end{equation}
%----%
whose image is isomorphic to $\IC[\XX]^{N_L\times N_R}$.

We are free to choose any further extension to $\IC[\gf^*_L\times \gf^*_R\times G]$, and in particular we can choose these extensions to live in the Poisson subalgebra $\IC[\mathfrak{b}^*\times\mathfrak{b}^*\times G]^{N_L\times N_R}$ of $\IC[\gf^*\times \gf^*\times G]$, which is identified \textit{as a ring} with $\IC[\XX]^{N_L\times N_R}$, the target of $\sigma$.

Summing up, for $a,b\in\IC[S_f\times S_f\times G]$, their Poisson bracket can be computed as
%----%
\begin{equation}
	\sigma\big(\{a,b\}_{S_f\times S_f\times G}\big) = \{\sigma(a),\sigma(b)\}_{\gf^*_L\times \gf^*_R\times G}~,
\end{equation}
%----%
where we interpret $\sigma(a),\sigma(b)$ as elements of the Poisson subalgebra $\IC[\mathfrak{b}^*_L\times \mathfrak{b}^*_R\times G]^{N_L\times N_R}$ as described above. With this technical detail out of the way, we now move on to the proof of Theorem \ref{thm:Poisson_generation} itself.

%----------------------------------------------------------------------%
\subsection{\label{subapp:reducting_coordinates}Reducing coordinates}
%----------------------------------------------------------------------%

Let $\Phi$ be the set of positive roots of $\gf$; we want to introduce some coordinates on $\gf^*$. For $\mathfrak{t}^*$ we shall introduce the \textit{linear} duals, $(h_i)_{i=1}^{\rk\,\gf}$ in $\mathfrak{t}$ (as in the case of $\KT_\gf$). Now for $\nf^*$, if $(e_\alpha)_{\alpha\in\Phi}$ are a basis of $\nf$ then $\big((e_{\alpha},-)\big)_{\alpha\in\Phi}$ form a basis of $\nf^*$. Therefore, the \textit{linear} dual to $(e_{\alpha},-)$ is $e_{-\alpha}\in\nf_-$. We choose $(e_{-\alpha})_{\alpha\in\Phi}$ to be coordinates on $\nf^*$. The ideal $\II$ is generated by $h_{i,L}$, $h_{i,R}$ for $i=1,\ldots N$ along with a linearly independent basis of the differences,\!\footnote{As previously, we append subscripts $L$ and $R$ to distinguish between the generators of $\IC[(\chi+\mathfrak{b}^*)_{L/R}]$.}
%----%
\begin{equation}
	e_{-\alpha,L} - e_{-\beta,L}~, e_{-\alpha,R}- e_{-\beta,R}~,
\end{equation}
%----%
where $\alpha,\beta\in\Phi$ are two roots with the same height.

The generator $P_1\in \IC[S_f]$ is extended to $\IC[\chi+\mathfrak{b}^*]$ as the restricted quadratic fundamental invariant, so the associated left and right versions are realised according to
%----%
\begin{equation}
\label{eq:section_of_P1}
	\sigma(P_{1,L/R}) = \kappa^{ij}h_{i,L/R}h_{j,L/R} + \sum_{\alpha\in \Delta} e_{-\alpha,L/R}~,
\end{equation}
%----%
where we recall that $\Delta$ is the set of simple roots of $\gf$ and $\kappa^{ij}$ are the components of the Killing form.

Next consider the coordinate functions on $G$. We will use $g_{ij}$ for the generators of $\IC[G]\subset\IC[S_f\times S_f\times G]$, and continue to use $\widetilde{g}_{ij}$ for the generators of $\IC[G]\subset\XX$. The subalgebra generated by ``bottom row'' entries $(\widetilde{g}_{Ni})_{i=1}^{N}$ is $N_L$-invariant, while that generated by ``left column'' entries, $(\widetilde{g}_{i1})_{i=1}^N$ is $N_R$-invariant. In particular, $\widetilde{g}_{N1}$ is itself $N_L\times N_R$-invariant since it is the highest weight vector. So the function $g_{1N}$ on $S_f\times S_f\times G$ extends in an $N_L\times N_R$-invariant way to all of $\XX$, and $\sigma(g_{N1})=\widetilde{g}_{N1}$.

Consider now the function $\sigma(\{P_{1,R},g_{N1}\})\in \IC[\XX]$---the Poisson bracket is computed in $\IC[\gf^*_L\times \gf^*_R\times G]$ using \eqref{eq:section_of_P1}, which gives
%----%
\begin{equation}
\label{eq:extension_2N}
	\sigma\big(\{P_{1,R},g_{N1}\}\big) = \widetilde{g}_{N2} + 2\sum_{i,j=1}^{\rk\,\gf}\kappa^{ij}h_{i,R}a_i(h_j)g_{N1}~,
\end{equation}
%----%
where $a_i$ are the components of the highest weight of the fundamental representation of ${\rm SL}_N$. By construction, this Poisson bracket must lie in $\IC[\XX]^{N_L\times N_R}$, and so lies in the image of $\sigma$. Restricting the expression \eqref{eq:extension_2N} to $S_f\times S_f\times G$ returns $g_{N2}$. Since $\sigma$ is injective, this fixes $\sigma(g_{N2})$ to the above expression.

The sum of simple negative roots acts on the fundamental representation of ${\rm SL}_N$ as a lowering operator, Therefore, from \eqref{eq:section_of_P1}
%----%
\begin{equation}\label{eq:P1_raising}
	\{\sigma(P_{1,R}),\widetilde{g}_{Ni}\} = \widetilde{g}_{N,i+1} + c_ig_{N,i}~,
\end{equation}
%----%
for $i=1,\ldots,N-1$ where $c_i$ are linear in the $(h_j)_{j=1}^{N}$.

%-------------------------------------------%
\begin{lem}\label{lem:form_of_extension}
%-------------------------------------------%
For $i=1,\dots,N$, the extension of $g_{Ni}$ is of the form
%----%
\begin{equation}\label{eq:form_of_extension}
	\sigma(g_{Ni}) = \widetilde{g}_{N,i+1} + \sum_{j=1}^{i-1}K_{ij}\widetilde{g}_{Nj}~,
\end{equation}
%----%
where $K_{ij}$ are homogeneous, degree $j$ polynomials in the generators of $\II$.
%-------------------------------------------%
\end{lem}
%-------------------------------------------%

%-------------------------------------------%
\begin{proof}
%-------------------------------------------%
We proceed by induction, noting that the case $i=1,2$ have been verified. Now suppose $i<N$. From \eqref{eq:section_of_P1} and the inductive assumption, we have
%----%
\begin{equation}
	\{\sigma(P_{1,R},\sigma(g_{Ni})\} = \widetilde{g}_{N,i+1}+c_ig_{Ni} +\sum_{j=1}^{i-1}K_{ij}\{P_{1,R},\widetilde{g}_{N,j}\} + \sum_{j=1}^{i-1}\{\sigma(P_{1,R}),K_{ij}\}\widetilde{g}_{N,j}~.
\end{equation}
%----%
The subleading terms are all in $\II$, except for the final term which we split as $\{\sigma(P_{1,R}),K_{ij}\} = F_{ij} + \II$, where $F_{ij}$ are homogeneous polynomials of degree $j$ in the $(P_{k,R})_{k=1}^N$ (restricted to $S_f$). Restricting to $S_f\times S_f\times G$, we have $\{ P_{1,R}, g_{Ni}\} = g_{N,i+1} + \sum_{j}^{i-1}F_{ij} g_{Nj}$.

Therefore, the extension of $g_{N,i+1}$ is
%----%
\begin{equation}
	\sigma(g_{N,i+1}) = \widetilde{g}_{N,i+1} + c_i g_{Ni} + \sum_{j=1}^{i-1}K_j\widetilde{g}_{N,j+1}~ + \sum_{j=1}^{i-1} K_jc_j\widetilde{g}_{Nj}~,
\end{equation}
%----%
which is in the desired form.
%-------------------------------------------%
\end{proof}
%-------------------------------------------%

In the proof of Lemma~\ref{lem:form_of_extension} we argued that for $i=1,\dots, N-1$, if $\sigma(g_{Ni})$ have the form of \eqref{eq:form_of_extension} then $\sigma(g_{N,i+1})$ can be written as a linear combination of $\{\sigma(P_{1,R}),\sigma(g_{N,i+1})\}$ and the $(F_{ij}g_{Nj})_{j=1}^{i-1}$. By restricting to $S_f\times S_f\times G$, we find an analogous result for the $(g_{Ni})_{i=1}^N$. By induction, this means that each $g_{Ni}$ are Poisson generated by $g_{N1}$ and the fundamental invariants, $(P_{i,R})_{i=1}^N$.

All that remains is to consider the restriction of these functions to $\Zf_G\subset S_f\times S_f\times G$. The embedding $\Zf_G\rightarrow S_f\times S_f\times G$ has image
%----%
\begin{equation}
	\{(s,s',g)\in S_f\times S_f\times G~|~ \mathrm{Ad}^*_g s = s'\}~.
\end{equation}
%----%
Since $S_f$ is transversal to the $G$-action, two distinct points on the slice cannot be $G$-conjugate. Therefore, this embedding factors through the closed embedding $\Zf_G\hookrightarrow G\times S_f$ with $G\times S_f\subset G\times S_f \times S_f$ realised as the diagonal. This amounts to enforcing the relation $P_{i,L}=P_{i,R}$ for each $i=1,\dots,N$. We know from Proposition~\ref{prop:ring_generators_sln} that the matrix entries of the bottom row of $G$ are ring generators. Thus, we conclude that $g_{N1}$ and $P_{i,R}$ Poisson generate $I_G$.
%-------------------------------------------%
\end{proof}
%-------------------------------------------%

\providecommand{\href}[2]{#2}\begingroup\raggedright\endgroup

\end{document}